\documentclass{llncs}

\usepackage{caption}
\usepackage[english]{babel}
\usepackage[utf8]{inputenc}

\usepackage{amsmath,amsthm,amssymb}

\usepackage{color}
\usepackage{graphicx}
\usepackage{newtxtext,newtxmath}
\usepackage{microtype}
\usepackage{xspace}
\usepackage{mathtools,tensor,xcolor,afterpage,wrapfig,framed,protocolj,arydshln,tikz}
\usepackage[colorlinks, citecolor=blue]{hyperref}
\usepackage{tikz-cd}
\usepackage{float}
\usepackage{bm}

\usepackage[framemethod=tikz]{mdframed}

\usepackage[noend,nofillcomment,linesnumbered,noline]{algorithm2e}

\usepackage{enumitem}

\usepackage{tikz}
\usetikzlibrary{calc}
\usetikzlibrary{positioning, shapes.geometric}

\SetCommentSty{mycommfont}

\newcommand{\ignore}[1]{}

\newcommand{\secp}{n}

\newcommand{\F}{\mathbb{F}}
\newcommand{\CF}{C\hspace{-1pt}F}
\newcommand{\CFt}{C\hspace{-1pt}\tilde{F}}

\newcommand{\eqdef}{\stackrel{\operatorname{def}}{=}}

\newcommand{\bool}[1][\relax]{{\ensuremath{{\{0,1\}}^{#1}}}\xspace}

\DeclareMathOperator{\BadComplete}{\mathsf{BadComplete}}
\DeclareMathOperator{\NewBadComplete}{\mathsf{NewBadComplete}}
\DeclareMathOperator{\ExistingBadComplete}{\mathsf{ExistingBadComplete}}
\DeclareMathOperator{\BadEval}{\mathsf{BadEval}}
\DeclareMathOperator{\GoodStatus}{\mathsf{GoodStatus}}
\DeclareMathOperator{\tv}{\mathsf{tv}}

\newcommand{\cP}{\ensuremath{\mathcal{P}}\xspace}

\newcommand{\cF}{\ensuremath{\mathcal{F}}\xspace}
\newcommand{\cG}{\ensuremath{\mathcal{G}}\xspace}
\newcommand{\cA}{\ensuremath{\mathcal{A}}\xspace}

\newcommand{\cS}{\ensuremath{\mathcal{S}}\xspace}

\newcommand{\cAtwo}{\ensuremath{\mathcal{A}}\xspace}
\newcommand{\cStwo}{\ensuremath{\mathcal{S}_\cAtwo}\xspace}

\newcommand{\cM}{\ensuremath{\mathcal{M}}\xspace}

\newcommand{\cD}{\ensuremath{\mathcal{D}}\xspace}
\newcommand{\sD}{\ensuremath{\widehat{\cD}}\xspace}
\newcommand{\sE}{\ensuremath{\widehat{\cE}}\xspace}
\newcommand{\cE}{\ensuremath{\mathcal{E}}\xspace}

\newcommand{\cO}{\mathcal{O}}

\newcommand{\cC}{\mathcal{C}}

\newcommand{\rC}{\mathit{C}}

\newcommand{\badG}{\tilde{\cG}}

\newcommand{\negl}{\ensuremath{\mathsf{negl}}\xspace}
\newcommand{\poly}{\ensuremath{\mathrm{poly}}\xspace}

\newcommand{\source}{\operatorname{source}}
\newcommand{\target}{\operatorname{target}}

\DeclareMathOperator{\GL}{\mathsf{GL}}

\definecolor{darkred}{rgb}{0.5, 0, 0}
\definecolor{darkgreen}{rgb}{0, 0.5, 0}
\definecolor{darkblue}{rgb}{0,0,0.5}
\definecolor{webred}{rgb}{0.5,0,0}
\definecolor{webblue}{rgb}{0,0,0.8}

\usepackage{float}

\title{Crooked indifferentiability of the Feistel Construction}

\author{
Alexander Russell\thanks{University of Connecticut, acr@cse.uconn.edu} 
\and Qiang Tang\thanks{The University of Sydney, qiang.tang@sydney.edu.au} 
\and Jiadong Zhu\thanks{University of Connecticut, zhujiadong2016@163.com, Corresponding Author}
}
\institute{}

\begin{document}
\maketitle

\begin{abstract}
  The Feistel construction is a fundamental technique for building
  pseudorandom permutations and block ciphers. This paper shows that a
  simple adaptation of the construction is resistant, even to algorithm
  substitution attacks---that is, adversarial subversion---of the
  component round functions. Specifically, we establish that a
  Feistel-based construction with more than $2000n/\log(1/\epsilon)$ rounds
  can transform a subverted random function---which disagrees with
  the original one at a small fraction (denoted by $\epsilon$) of
  inputs---into an object that is \emph{crooked-indifferentiable} from
  a random permutation, even if the adversary is aware of all the
  randomness used in the transformation. We also provide a lower bound
  showing that the construction cannot use fewer than
  $2n/\log(1/\epsilon)$ rounds to achieve
  crooked-indifferentiable security. 
\end{abstract}


\pagenumbering{arabic}
\pagestyle{plain}

\section{Introduction}

Random permutations and ideal ciphers are idealized models that have proven to be powerful tools for
designing and reasoning about cryptographic schemes, particularly in the settings of symmetric key encryption (e.g., block ciphers) and hash designs.   
They consist of
the following two steps: (i) design a scheme $\Pi$ in which all
parties (including the adversary) have oracle access to (a family of) truly
random permutations (and the corresponding inversions), and establish the security of $\Pi$ in this favorable
setting; (ii) instantiate the oracle in $\Pi$ with a suitable
cipher (such as AES) to obtain an instantiated
scheme $\Pi'$.  The random permutation (ideal cipher) heuristic states that if the
original scheme $\Pi$ is secure, then the instantiated scheme $\Pi'$
is also secure.
In this work we focus
on the problem of correcting faulty---or adversarially
corrupted---ideal ciphers/random permutations so that they can be confidently applied for
such cryptographic purposes.


One particular motivation for correcting random permutations/ideal ciphers in a
cryptographic context arises from works studying design
and security in the subversion (i.e., \emph{kleptographic}) setting.
In this setting, various components of a cryptographic scheme may
be subverted by an adversary, so long as the tampering cannot be
detected via blackbox testing. This is a challenging framework
because many basic cryptographic techniques are
not directly available: in particular, the random permutation/ideal cipher paradigm is
directly undermined. In terms of the discussion above, the random permutation/ideal cipher---which is eventually to be replaced with a concrete cipher---is subject to adversarial subversion which complicates even
the first step of the random permutation/ideal cipher methodology. To see a simple example, for AES, denoted as ($\sf{AES.K, AES.E, AES.D}$), whose software/hardware implementation (denoted as $\sf{\widetilde{AES.K}, \widetilde{AES.E}, \widetilde{AES.D}}$) might be subverted as follows: ${\sf\widetilde{AES.E}}(k,m^*)=k$, for a trigger message $m^*$ randomly chosen by the adversary, while $\sf{\widetilde{AES.E}}=\sf{AES.E}$ otherwise, i.e., only when encrypting a special trigger message, the subverted encryption algorithm directly outputs the secret key. This is clearly undetectable via blackbox testing. 
Since the subverted implementation of AES now cannot be assumed to be an ideal cipher anymore, the security of constructions that previously relied on this assumption also becomes elusive.

Our goal is
to provide a generic approach that can rigorously ``protect'' the
usage of ideal cipher/random permutation from subversion and, essentially, establish a
``crooked'' ideal cipher/random permutation methodology.
Specifically, given a function $\tilde{h}$ drawn from a distribution
which {\em agrees in most places} with a uniform permutation, we would like to
produce a corrected version which appears still as an ideal cipher/random permutation to adversaries with
a polynomially bounded number of queries. This model is also analogous to 
the classical study of ``program checking and self-correcting''
\cite{Blum88,STOC:BluKan89,STOC:BluLubRub90}: the goal in this theory
is to transform a program that is faulty at a small fraction of inputs
(modeling an evasive adversary) to a program that is correct at all
points with overwhelming probability. Our setting intuitively adapts
this classical theory of self-correction to the study of ``self-correcting
a probability distribution.''  Notably, the functions to be corrected
are structureless, instead of heavily structured. 

\medskip
\noindent{\em The model of ``crooked'' indifferentiability.} The first
work in this line was \cite{C:RTYZ18}, focusing on correcting
subverted random oracles; in particular, they introduced a security
model called {\em crooked-indifferentiability} to formally capture the
problem as follows: First, a function
$h: \{0,1\}^n \rightarrow \{0,1\}^n$ is drawn uniformly at random.
Then, an adversary may \emph{subvert} the function $h$, yielding a new
function $\tilde{h}$. The subverted function $\tilde{h}(x)$ is
described by an adversarially-chosen (polynomial-time) algorithm
$\tilde{H}^h(x)$, with oracle access to $h$. This function may differ
from the original function (so that $\tilde{h}(x) \neq h(x)$) at only
a negligible fraction of inputs (to evade blackbox testing).
To show that the resulting function (construction) is ``as good as'' a
random oracle in the sense of indifferentiability
~\cite{TCC:MauRenHol04,C:CDMP05}, an \emph{$H$-crooked-distinguisher
  ${\sD}$} was introduced; it first prepares the subverted
implementation $\tilde{H}$ (after querying $H$ first); then a fixed
amount of (public) randomness $R$ is drawn and published; the
construction $\rC$ may use only the subverted implementation
$\tilde{H}$ and the randomness $R$. Now, following the
indifferentiability framework, we will ask for a simulator $\cS$ such
that $(\rC^{\tilde{H}^h}(\cdot,R),h)$ and $(\cF,\cS^{\tilde{H}}(R))$
are indistinguishable to any $H$-crooked-distinguisher $\sD$ (even one
who knows $R$).

\subsection{Our Contribution}\label{subsec:our contribution and construction}
We investigate the above question in the more restrictive random permutation/ideal cipher setting.
We first adopt the security model of {\em crooked-indifferentiability} for random permutation. 
(A formal definition appears in Section~\ref{sec:model}.)

\medskip
\noindent\textbf{A warm up construction.} To consider feasibility of
correcting a subverted random permutation, and also as an example to
explore the crooked-indifferentiability model, we start with a warm-up
construction by composing the following two components.


\smallskip
\noindent{\em Component I.} The first component is built from a source random function that was proven to be {\em crooked}-indifferentiable from a random oracle \cite{C:RTYZ18}.

The source function is expressed as a family of
$\ell+1$ independent random oracles:
\[ h_0: \{0,1\}^{n} \rightarrow \{0,1\}^{3n}\,, \text{and }
  h_i: \{0,1\}^{3n} \rightarrow \{0,1\}^{n}\,\text{for
       $i \in \{1, \ldots, \ell\}$.} \]
These can be realized as slices of a single random function
$H: \{0,1\}^{n'} \rightarrow \{0,1\}^{n'}$, with
$n' = 3n + \lceil \log \ell+1 \rceil$ by an appropriate convention for
embedding and extracting inputs and values. Given subverted implementations
$\{\tilde{h}_i\}_{i=0,\ldots,\ell}$ (defined as above by the
adversarially-defined algorithm $\tilde{H}$),
the corrected function is defined as:
\[
  \rC^{\tilde{H}^H}(x) \eqdef \tilde{h}_0\left(\bigoplus_{i = 1}^{\ell}
    \tilde{h}_i(x \oplus r_i)\right)\,,
\]
where $R = (r_1, \ldots, r_\ell)$ is sampled uniformly after
$\tilde{h}$ is provided (and then made public).

\smallskip
\noindent{\em Component II: the classical Feistel cipher.} The second component is the classical Feistel cipher with $\ell$ rounds for $\ell=14$. Coron et al.~\cite{JC:CHKPST16} proved it is indifferentiable from a random permutation.
%
The classical \emph{$\ell$-round Feistel
  cipher} transforms a sequence of functions
$F_1, \ldots, F_\ell: \{0,1\}^n \rightarrow \{0,1\}^n$ into a
permutation on the set $\{0,1\}^{2n}$. The construction logically
treats $2n$-bit strings as pairs $(x,y)$, with $x, y \in \{0,1\}^n$,
and is defined as the composition of a sequence of permutations
defined by the $F_i$. Specifically, given an input $(x_0,x_1)$, the
construction defines
\[
  x_{i+1}:=x_{i-1} \oplus F_i(x_i)
\]
for each $i = 1, \ldots, \ell$, and results in the output string
$(x_{\ell}, x_{\ell+1})$. It is easy to see that the resulting
function is a permutation. In practical settings, the ``round
functions'' ($F_i$) are often keyed functions (determined by secret
keys of length $\poly(n)$), in which case the construction results in
a keyed permutation.

\begin{figure}[ht]
\vspace{-5mm}
\begin{center}
\begin{tikzpicture}[scale=0.8]
    \node[draw,thick,minimum width=1cm] (f1) at ($1*(0,-1.5cm)$)  {$F_1$};
    \node (xor1) [left of = f1, node distance = 2cm] {$\bigoplus$};
    \draw[thick,-latex] (f1) -- (xor1);

    \node[draw,thick,minimum width=1cm] (f2) at ($2*(0,-1.5cm)$)  {$F_2$};
    \node (xor2) [left of = f2, node distance = 2cm] {$\bigoplus$};
    \draw[thick,-latex] (f2) -- (xor2);
    
 	\draw[thick,latex-latex] (f1.east) -| +(1.5cm,-0.5cm) -- ($(xor1) - (0,1cm)$) -- ($(xor1.north) - (0,1.5cm)$);
 	\draw[thick] (xor1.south) -- ($(xor1)+(0,-0.5cm)$) -- ($(f1.east) + (1.5cm,-1cm)$) -- +(0,-0.5cm);
    
 	\draw[thick,latex-] (f2.east) -| +(1.5cm,-0.5cm) -- ($(xor2) - (0,1cm)$);
 	\draw[thick] (xor2.south) -- ($(xor2)+(0,-0.5cm)$) -- ($(f2.east) + (1.5cm,-1cm)$);
	
	\draw[thick, densely dotted] ($(f2.east) + (1.5cm,-1cm)$) -- +(0,-0.5cm);
	\draw[thick, densely dotted] ($(xor2) - (0,1cm)$) -- ($(xor2.north) - (0,1.5cm)$);
	
    \node at (0,-4.5cm) {\scriptsize{for $\ell$ rounds}};

    \node[draw,thick,minimum width=1cm] (f3) at ($3*(0,-1.5cm) + (0, -.75cm)$)  {$F_{\ell-1}$};
    \node (xor3) [left of = f3, node distance = 2cm] {$\bigoplus$};
    \draw[thick,-latex] (f3) -- (xor3);

    \node[draw,thick,minimum width=1cm] (f4) at ($4*(0,-1.5cm) + (0, -.75cm)$)  {$F_{\ell}$};
    \node (xor4) [left of = f4, node distance = 2cm] {$\bigoplus$};
    \draw[thick,-latex] (f4) -- (xor4);
    
    \draw[thick,latex-latex] (f3.east) -| +(1.5cm,-0.5cm) -- ($(xor3) - (0,1cm)$) -- ($(xor3.north) - (0,1.5cm)$);
 	\draw[thick] (xor3.south) -- ($(xor3)+(0,-0.5cm)$) -- ($(f3.east) + (1.5cm,-1cm)$) -- +(0,-0.5cm);
	
	\draw[thick, densely dotted] ($(f3.east) + (1.5cm,0cm)$) -- +(0cm,0.5cm);
	\draw[thick, densely dotted] (xor3.north) -- +(0cm,0.35cm);

    \node (p0) [draw,thick,above of = f1, minimum width=5cm,minimum height=0.5cm,node distance=1cm] {$IP$}; 
    \node (l0) [above of = xor1,node distance=2cm] {$L$};
    \node (r0) [right of = l0, node distance = 4cm] {$R$};
    \draw[thick,-latex] (l0 |- p0.south) -- (xor1.north);
    \draw[thick] ($(f1.east)+(1.5cm,0)$) -- +(0,0.75cm);

    \draw[thick,latex-] (l0 |- p0.north) -- (l0);
    \draw[thick,latex-] (r0 |- p0.north) -- (r0);

    \node (p4) [draw,thick,below of = f4, minimum width=5cm,minimum height=0.5cm,node distance=1cm] {$FP$}; 
    \node (l4) [below of = xor4,node distance=2cm] {$L'$};
    \node (r4) [right of = l4, node distance = 4cm] {$R'$};
    \draw[thick,latex-latex] (f4.east) -| +(1.5cm,-0.75cm);
    \draw[thick,-latex] (xor4.south) -- ($(xor4)+(0,-0.75cm)$);

    \draw[thick,-latex] (l4 |- p4.south) -- (l4);
    \draw[thick,-latex] (r4 |- p4.south) -- (r4);

\end{tikzpicture}

\end{center}

\caption{The $\ell$ round classical Feistel construction.}
\vspace{-4mm}
\end{figure}
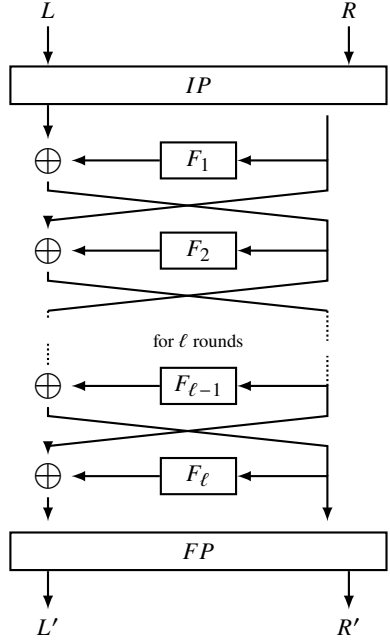


\smallskip
\noindent{\em Composing the two components.} We can compose the above two components by replacing the 14 round functions in component II with 14 independent copies of component I. The result construction, by the property of indifferentiability, is also crooked-indifferentiable from a random permutation as a corollary of the replacement theorem of crooked-indifferentiability (see Section \ref{sec:crooked-indiff}).

\medskip

\noindent{\bf Our direct and ``optimal'' construction.} 
\label{Construction} However, there are two drawbacks of this construction. First, the structure of the construction is complicated. Second, it makes at least linear number of invocations of the underlying subverted component (and also $O(n^2)$ random bits) to achieve security.  Instead, we prove that a {\em direct} Feistel-based construction can also work and remove these drawbacks. In particular, our construction involving only  \emph{public}
randomness can boost a ``subverted'' random permutation (or just a function) into a construction
that is indifferentiable from a perfect random permutation. (Section
\ref{sec:overview}, \ref{sec:proof}). Besides structure-wise simplicity (and the fact that it adopts the direct Feistel structure), our construction requires a smaller number of invocations of the underlying (subverted) random permutation, which is essentially optimal up to constant factors (at least for the Feistel structure, as we prove impossibility to have fewer rounds; there was also explicit attacks for the case of random oracle in \cite{C:RTYZ18}, but the construction in \cite{C:RTYZ18} was not ``tight'' in this sense, with a factor of $O(\log(1/\epsilon))$).

Our subversion-resistant construction on strings
of length $2n$ relies on the parameter $\ell$ and the Feistel
construction applied to $\ell$ round functions that are determined by:
\begin{itemize}
\item $\ell$ functions $F_i: \{0,1\}^n \rightarrow \{0,1\}^n$ that are
  subject to subversion as described above, 
\item an additional family of $\ell$ public, uniform affine-linear
  functions determined by $\ell$ pairs
  $(a_i,b_i) \in \GL(\F_2,n) \times \F_2^n$.\footnote{For technical reasons, we need to encode the input of the round function with the pairwise independent function, please see the proof of Lemma \ref{Lemma: No subverted chain covers a long honest chain} for detailed discussions.}
\end{itemize}
The affine-linear functions are determined by independent and uniform
selection of $a_i$ from $\GL(n,\mathbb{F}_2)$ (to be concrete, the
collection of invertible $n \times n$ matrices with elements in $\F_2$), and
$b_i \in \mathbb{F}_2^n$. The $i$-th affine linear function, defined on
an input $x \in \F_2^n$, is given by the rule
$x \mapsto a_i \cdot x \oplus b$. The final construction is given by the
Feistel construction applied to the round functions
$x \mapsto \tilde{F}_i(a_i \cdot x \oplus b)$, where $\tilde{F}$ is the
subverted version of the function $F_i$. To be concrete, with the data
$(F_i, a_i, b_i)$ (with $i = 1, \ldots, \ell$), the construction
$C: \{0,1\}^{2n} \rightarrow \{0,1\}^{2n}$ is defined by the rule
\begin{align*}
  & C(x_0,x_1):=(x_{\ell},x_{\ell+1})\,,\text{where}\\
  & x_{i+1}:=x_{i-1} \oplus \tilde{F}_i(a_i \cdot x_i \oplus b_i)\text{ , for $i=1,\ldots,\ell$}\,.
\end{align*}
where $n$-bit strings $x$ and $b_i$ are viewed as length $n$ column
vectors, $a_i \cdot x_i$ is the multiplication between matrix $a_i$
and column vector $x_i$, and $\tilde{F}_i(x)$ is the subverted
function value at $(i.x)$ using the subversion algorithm $\cA$.

\ignore{
\paragraph{\bf Some features of our construction.} Our construction has several features that are different from other crooked-indifferentiable constructions:
\begin{itemize}
\item Our construction has a simpler structure than other construction that are crooked-indifferentiable from a random permutation. All existing constructions are obtained by composing two smaller constructions, while our construction is a straightforward Feistel construction.
\item Our construction uses $O(n^2/\log(1/\epsilon)$ random bits which are better than other constructions, which use $n^2$ random bits. $O(n^2/\log(1/\epsilon)$ is also the lower bound of the number of the random bits required for crooked-indifferentiability with the Feistel construction.
\item Our construction needs a different technique than that in \cite{C:RTYZ18}. The security of the two-layer construction in \cite{C:RTYZ18} is relied on the fact that the XOR structure
\[
  \tilde{g}_R(x) \eqdef \bigoplus_{i = 1}^{\ell}
    \tilde{h}_i(x \oplus r_i)\,,
\]
is unpredictable and the simulator can always program $h_0$ at $\tilde{g}_R(x)$. By contrast, the simulator in our construction can not program at a fixed round function (because otherwise the distinguisher can always query this round function first). The simulator needs to flexibly chooses where to program according to the queries of the distinguisher. Also, the Feistel construction also needs a more sophisticated security analysis than the simple XOR structure.
\end{itemize}
}

\medskip
\noindent{\bf New techniques for proving crooked-indifferentiability of Feistel structure.}
Besides that we aim to get a random permutation, which has stricter requirements, our security analysis needs substantially more sophisticated techniques than that in \cite{C:RTYZ18}. The security of the two-layer construction for random oracle in \cite{C:RTYZ18} relies on the fact that the XOR structure
\[
  \tilde{g}_R(x) \eqdef \bigoplus_{i = 1}^{\ell}
    \tilde{h}_i(x \oplus r_i)\,
\]
is unpredictable so that the simulator can always program $h_0$ (at $\tilde{g}_R(x)$). By contrast, our simulator cannot program at one fixed round function (because otherwise the distinguisher can always query this round function first). The simulator needs to flexibly choose where to program according the queries of the distinguisher. 

We remark that some techniques in our proof are inspired by the elegant techniques of Coron et al.~\cite{JC:CHKPST16} for conventional indifferentiability; for example, we adopt the concept of ``chain'' to analyze the basic structure of Feistel construction. However, the subversion of the random function in our setting introduces multiple new challenges, because of, e.g., on-the-fly adaptive queries of  the subverted $\tilde{F}$ when the simulator runs it. 


To achieve ``crooked'' indifferentiability, our simulator needs to
ensure consistency between two ways of generating output values: one
is directly from the construction $C$; the other
calls for an ``explanation'' of $P$---a truly random permutation---via
reconstruction from related queries to $F$ (in a way consistent with
the subverted $\tilde{F}$). To ensure a correct
simulation, the simulator must suitably answer related queries
(defining one value of $C$). Essentially, the proof relies on the fact that for any Feistel ``chain'' $(x_0,\ldots,x_{\ell+1})$, the simulator can find two places $(x_u,x_{u+1})$ and program $F_u(a_u \cdot x_u \oplus b_u):=x_{u-1} \oplus x_{u+1}$, $F_{u+1}(a_{u+1} \cdot x_{u+1} \oplus b_{u+1}):=x_u \oplus x_{u+2}$ to make the Feistel chain consistent with $P(x_0,x_1)=(x_{\ell},x_{\ell+1})$.
There are two major challenges in the simulation: first, one of the
two programmed terms ($F_u(a_u \cdot x_u \oplus b_u)$ and
$F_{u+1}(a_{u+1} \cdot x_{u+1} \oplus b_{u+1})$) may be already
evaluated prior to programming by the simulator; second,
the one of the two programmed terms may be dishonest (i.e.,
$\tilde{F} \neq F$) so that programming may not be possible.

In the security proof, to analyze the difference between the construction and the ideal object (random permutation), we need to carefully design several intermediate games for transition. Using the games, we reduce the gap between the construction and the ideal object to the probability of two ``bad events'' that reflect the two challenges above. Finally, we prove the bad events are negligible by carefully analyzing the structure of our construction. We also need to give a more careful analysis of efficiency of the simulator as it has to internally generate many more terms because of the necessity of running $\tilde{F}$.



\ignore{
\begin{itemize}
\item A function $h: \{0,1\}^n \rightarrow \{0,1\}^n$ is drawn uniformly at random.
\item An adversary may \emph{subvert} the function $h$, yielding a new
  function $\tilde{h}$. The subverted function $\tilde{h}(x)$ is
  described by an adversarially-chosen algorithm $H^h(x)$, with oracle
  access to $h$, which determines the value of $\tilde{h}(x)$ after
  making a polynomial number of adaptive queries to $h$. We insist
  that $\tilde{h}(x) = h(x)$ with probability $1- \epsilon$ over
  random choice of $x$.
\item The function $\tilde{h}$ is ``corrected'' to a function
  $\tilde{h}_R$ by a procedure which may involve public randomness.
\end{itemize}
}



\medskip
\noindent{\bf Related works}. 
\noindent{\em Conventional indifferentiability of Feistel cipher.} The notion of
indifferentiability was proposed in the elegant work of Maurer et
al.~\cite{TCC:MauRenHol04}; this notably extends the classical concept
of indistinguishability to circumstances where one or more of the
relevant oracles are publicly available (such as a random oracle). It
was later adapted by Coron et al.~\cite{C:CDMP05}; several other
variants were proposed and studied in
\cite{TCC:DodPun06,EC:DodPun07}. A line of notable work applied the framework to the 
ideal cipher problem: in particular the Feistel construction (with a
small constant number of rounds) is indifferentiable from a random
permutation, see \cite{JC:CHKPST16,EC:DacKatThi16,C:DaiSte16}. Our
work adopts the indifferentiability framework applied to the subverted
case (that is, crooked-indifferentiability); the construction
aims to sanitize a subverted random function to be indifferentiable from a clean random permutation.

\smallskip
\noindent{\em Crooked-indifferentiability of random oracles.}
In \cite{C:RTYZ18}, the authors proved that a simple two-layer construction using $O(n^2)$ public random bits is crooked-indifferentiable from a random oracle (following results \cite{INDO,RO-full} gave more rigorous analysis, and show applications in subversion resistant digital signatures \cite{PKC:CRTYZZ19}). This work focuses on a  strictly stronger goal: to obtain a random {\em permutation}, and with smaller number of rounds. This line of work was motivated to defend against kleptographic attacks, originally introduced by Young and
Yung~\cite{C:YouYun96,EC:YouYun97}, with renewed recent interests (e.g.,~\cite{C:BelPatRog14,EC:DGGJR15,AC:RTYZ16}).

\smallskip
\noindent{\em Related work on non-uniformity and pre-processing.} There are several recent approaches that
study idealized objects in the auxiliary input model (or with
pre-processing)~\cite{EC:DodGuoKat17,EC:CDGS18}. As pointed out in \cite{RO-full}, crooked-indifferentiability is strictly
stronger than the pre-processing model: besides pre-processing
queries, the adversary may embed (and keep) compressed state as backdoor; more importantly, our
subverted implementation can further misbehave in ways that cannot be
captured by any single-shot polynomial-query adversary because the
subversion at each point is determined by a local adaptive
computation.

\ignore{
\smallskip
\noindent{\em Related work on kleptographic security.}  
Kleptographic attacks were originally introduced by Young and
Yung~\cite{C:YouYun96,EC:YouYun97}: In such attacks, the adversary
provides subverted implementations of the cryptographic primitive,
trying to learn secrets without being detected.
%
In recent years, several remarkable allegations of cryptographic tampering~\cite{PLS13,ReutersRSA2013}, including detailed investigations~\cite{Checkoway14,juniper_paper}, have produced a renewed interest in both kleptographic attacks and in techniques for preventing them~ e.g.,\cite{C:BelPatRog14,EC:DGGJR15,AC:RTYZ16,CCS:RTYZ17}. None of those work considered how to actually correct a subverted random permutation.
}

\section{The Model: Crooked Indifferentiability}
\label{sec:model}


\subsection{Preliminary: Indifferentiability}
\label{sec:indifferentiability}

The notion of indifferentiability proposed in the elegant work of
Maurer et al.~\cite{TCC:MauRenHol04} has proven to be a powerful tool
for studying the security of hash function and many other primitives.
The notion extends the classical concept of indistinguishability to
the setting where one or more oracles involved in the construction are
publicly available. The indifferentiability framework
of~\cite{TCC:MauRenHol04} is built around random systems providing
interfaces to other systems. 
Coron et al.~\cite{C:CDMP05} demonstrate a strengthened\footnote{Technically, the quantifiers in the security definitions in the original~\cite{TCC:MauRenHol04} and in the followup ~\cite{C:CDMP05} are different; in the former, a simulator needs to be constructed for each adversary, while in the latter a simulator needs to be constructed for {\em all} adversaries.}
 indifferentiability framework built around Interactive
Turing Machines (as in~\cite{FOCS:Canetti01}). Our presentation
borrows heavily from~\cite{C:CDMP05}.  In the next subsection, we will
introduce our new notion, \emph{crooked indifferentiability}.

\paragraph{Defining indifferentiability.}
An \emph{ideal primitive} is an algorithmic entity which receives
inputs from one of the parties and returns its output immediately to
the querying party.  We now proceed to the definition of
indifferentiability~\cite{TCC:MauRenHol04,C:CDMP05}:


\begin{definition}[Indifferentiability~\cite{TCC:MauRenHol04,C:CDMP05}]
A Turing machine $\rC$ with oracle access to an ideal primitive $\cG$ is said to be   $(t_\cD,t_\cS,q,\epsilon)$-indifferentiable from an ideal primitive $\cF$, if there is a simulator $\cS$, such that for any distinguisher $\cD$, it holds that :
$$\left| \Pr[\cD^{\rC,\cG}(1^\secp) = 1] - \Pr[\cD^{\cF,\cS}(1^\secp) = 1] \right| \leq \epsilon\,.$$
  The simulator $\cS$ has oracle access to $\cF$ and runs in time at most $t_\cS$. The distinguisher $\cD$ runs in time at most $t_\cD$ and makes at most $q$ queries. Similarly, $\rC^\cG$ is said to be (computationally) indifferentiable from $\cF$ if $\epsilon$ is a negligible function of the security parameter $\secp$ (for polynomially bounded $t_\cD$ and $t_\cS$).
  See Figure \ref{fig:indiff}.
  
\end{definition}

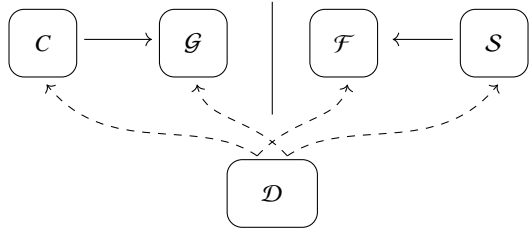
\begin{figure}[htb!]
  \begin{center}
    \begin{tikzpicture}
            \draw[thin, rounded corners=2mm] (-3.5,1.5) rectangle +(.9,.9) node[pos=.5] {$\rC$};
      \draw[->,thin] (-2.5,2) -- (-1.6,2);

      \draw[thin, rounded corners=2mm] (-1.5,1.5) rectangle +(.9,.9) node[pos=.5] {$\cG$};
      \draw[thin] (0,2.5) -- (0,1);

                \draw[thin, rounded corners=2mm] (.5,1.5) rectangle +(.9,.9) node[pos=.5] {$\cF$};
      \draw[->,thin] (2.4,2) -- (1.6,2);

      \draw[thin, rounded corners=2mm] (2.5,1.5) rectangle +(.9,.9) node[pos=.5] {$\cS$};
           
\draw[thin,dashed,->] (0.2,0.45) .. controls (1,1) and (2.2,.4) .. (3,1.4);
      \draw[thin,dashed,->] (-0.2,0.45) .. controls (.3,1) and (.9,.9) .. (1,1.4);
\draw[thin,dashed,->] (0.2,0.45) .. controls (-.3,1) and (-.9,.9) .. (-1,1.4);
      \draw[thin,dashed,->] (-0.2,0.45) .. controls (-1,1) and (-2.2,.4) .. (-3,1.4);

            \draw[thin, rounded corners=2mm] (-.6,-.5) rectangle +(1.2,.9) node[pos=.5] {$\cD$};

    \end{tikzpicture}
  \end{center}
\caption{
\label{fig:indiff} 
The indifferentiability notion: the distinguisher $\cD$ either interacts with algorithm $\rC$ and ideal primitive $\cG$, or with ideal primitive $\cF$ and simulator $\cS$. Algorithm $\rC$ has oracle access to $\cG$, while simulator $\cS$ has oracle access to $\cF$.}
\end{figure}

As illustrated in Figure \ref{fig:indiff}, the role of the simulator is to simulate the ideal primitive
$\cG$ so that no distinguisher can tell whether it is interacting with $\rC$ and $\cG$, or with $\cF$ and
$\cS$; in other words, the output of $\cS$ should look ``consistent'' with what the distinguisher
can obtain from $\cF$. Note that the simulator does not observe the distinguisher's queries to
$\cF$; however, it can call $\cF$ directly when needed for the simulation. Note that, in some sense, the simulator must ``reverse engineer'' the construction $\rC$ so that the simulated oracle appropriately induces $\cF$ and, of course, possesses the correct marginal distribution. 

\paragraph{Replacement.}
 It is shown in \cite{TCC:MauRenHol04} that if $\rC^\cG$ is indifferentiable from $\cF$, then $\rC^\cG$ can replace $\cF$ in any cryptosystem, and the resulting cryptosystem is at least as secure in the  $\cG$  model as in the $\cF$ model. 

We use the definition of \cite{TCC:MauRenHol04} to specify what it means for a cryptosystem to be at least as secure in the  $\cG$  model as in the $\cF$ model. A cryptosystem is modeled as an Interactive Turing Machine with an interface to an adversary $\cA$ and to a public oracle. 
The cryptosystem is run by an environment $\cE$ which provides a binary output and also runs the adversary. In the  $\cG$  model, cryptosystem $\cP$ has oracle access to $\rC$ whereas attacker $\cA$ has oracle access to $\cG$. In the $\cF$ model, both $\cP$ and $\cA$ have oracle access to $\cF$. The definition is illustrated in Figure~\ref{fig:composition}.

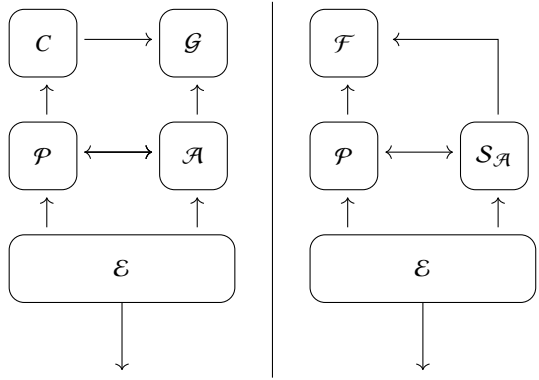
\begin{figure}
  \begin{center}
    \begin{tikzpicture}
     
            \draw[thin, rounded corners=2mm] (-3.5,1.5) rectangle +(.9,.9) node[pos=.5] {$\rC$};

      \draw[->,thin] (-2.5,2) -- (-1.6,2);

      \draw[thin, rounded corners=2mm] (-1.5,1.5) rectangle +(.9,.9) node[pos=.5] {$\cG$};
      \draw[thin] (0,2.5) -- (0,-2.5);

                \draw[thin, rounded corners=2mm] (-3.5,0) rectangle +(.9,.9) node[pos=.5] {$\cP$};

      \draw[->,thin] (-2.5,0.5) -- (-1.6,0.5);
      
            \draw[->,thin] (-3,1) -- (-3,1.4);
            \draw[->,thin] (-1,1) -- (-1,1.4);

      \draw[thin, rounded corners=2mm] (-1.5,0) rectangle +(.9,.9) node[pos=.5] {$\cAtwo$};
             \draw[thin, rounded corners=2mm] (-3.5,-1.5) rectangle +(3,.9) node[pos=.5] {${\cE}$};
                                                                              \draw[->,thin] (-2,-1.5) -- (-2,-2.4);


      \draw[->,thin] (-2.5,0.5) -- (-1.6,0.5);
      
            \draw[->,thin] (-3,-.5) -- (-3,-0.1);
            \draw[->,thin] (-1,-.5) -- (-1,-.1);

                \draw[thin, rounded corners=2mm] (.5,1.5) rectangle +(.9,.9) node[pos=.5] {$\cF$};


                                  
                                     \draw[thin, rounded corners=2mm] (.5,0) rectangle +(.9,.9) node[pos=.5] {$\cP$};

      \draw[<->,thin] (-2.5,0.5) -- (-1.6,0.5);
      
            \draw[->,thin] (1,1) -- (1,1.4);
            \draw[->,thin] (3,1) |- (1.6,2);

      \draw[thin, rounded corners=2mm] (2.5,0) rectangle +(.9,.9) node[pos=.5] {$\cStwo$};

                      \draw[thin, rounded corners=2mm] (0.5,-1.5) rectangle +(3,.9) node[pos=.5] {${\cE}$};
                                                        \draw[->,thin] (2,-1.5) -- (2,-2.4);

      \draw[<->,thin] (1.5,0.5) -- (2.4,0.5);
      
           \draw[->,thin] (1,-.5) -- (1,-0.1);
            \draw[->,thin] (3,-.5) -- (3,-.1);

    \end{tikzpicture}
  \end{center}
\caption{
\label{fig:composition} 
The environment $\cE$ interacts with cryptosystem $\cP$ and attacker $\cAtwo$. In the $\cG$ model (left), $\cP$ has oracle access to $\rC$ whereas $\cAtwo$ has oracle access to $\cG$. In the $\cF$ model, both $\cP$ and $\cStwo$ have oracle access to $\cF$.}
\end{figure}

\begin{definition}
A cryptosystem is said to be at least as secure in the $\cG$ model with algorithm $\rC$ as in the $\cF$ model, 
if for any environment $\cE$ and any attacker $\cAtwo$ in the $\cG$ model, there exists an attacker $\cStwo$ in the $\cF$ model, such that: 
\[
\Pr[\cE(\cP^{\rC^{}},\cAtwo^{\cG})=1]-\Pr[\cE(\cP^\cF,\cStwo^\cF)=1]\leq\epsilon.
\]
where $\epsilon$ is a negligible function of the security parameter $\secp$. 
Similarly, a cryptosystem is said to be computationally at least as secure, etc., if $\cE$, $\cAtwo$ and $\cStwo$ are polynomial-time in $\secp$.
\end{definition}

 

We have the following security preserving (replacement) theorem, which says that when an ideal primitive is replaced by an indifferentiable one, the security of the ``big'' cryptosystem remains:
 \begin{theorem}[\cite{TCC:MauRenHol04,C:CDMP05}]
 \label{theorem:composition}
Let $\cP$ be a cryptosystem with oracle access to an ideal primitive $\cF$. Let $\rC$ be an algorithm such that $\rC^{\cG}$ is indifferentiable from $\cF$. Then cryptosystem $\cP$ is at least as secure in the $\cG$ model with algorithm $\rC$ as in the $\cF$ model. 
\end{theorem}

\subsection{Crooked indifferentiability}
\label{sec:crooked-indiff}

The ideal primitives that we focus on in this paper are random permutations.  
A random permutation is an ideal primitive which provides an independent random output for each new query so that the resulting function is a permutation.  We next formalize a new notion called \emph{crooked indifferentiability} to reflect the challenges in our setting with subversion; our formalization is for random permutations, but the formalization can be naturally extended to other ideal primitives.    

\paragraph{Crooked indifferentiability for random permutations.}

As mentioned in the introduction, we consider the problem of ``repairing'' a subverted random permutation in such a way that the corrected construction can be used as a drop-in replacement for an unsubverted random permutation. This immediately suggests invoking and appropriately adapting the indifferentiability notion. Specifically, we need to adjust the notion to reflect subversion.

We model the act of \emph{subversion of a permutation $H$} as creation of an ``implementation'' $\tilde{H}$ of the new, subverted permutation; in practice, this would be the source code of the subverted version of the function $H$. In our setting, however, where $H$ is modeled as a random permutation, we define $\tilde{H}$ as a polynomial-time algorithm with oracle access to $H$; thus the subverted function is $x \mapsto \tilde{H}^H(x)$. 
We proceed to survey the main modifications between crooked indifferentiability and the original notion of indifferentiability. 
\begin{enumerate}
\item The deterministic construction will have oracle access to the random permutation only via the subverted implementation $\tilde{H}$ but not via the ideal primitive $H$. (Operationally, the construction has oracle access to the function $x \mapsto \tilde{H}^H(x)$.)
The construction depends on access to trusted, but public, randomness $R$.  
\item The simulator is provided, as input, the subverted implementation $\tilde{H}$ (a Turing machine) and the public randomness $R$; it has oracle access to the target ideal functionaltiy ($\cF$).
\end{enumerate}
Point (2) is necessary, and desirable, as it is clearly impossible to achieve indifferentiability using a simulator that has no access to $\tilde{H}$ (the distinguisher can simply query an input such that $\rC$ will use a value that is modified by $\tilde{H}$ while $\cS$ has no way to reproduce this). More importantly, we will show below that security will be preserved by replacing an ideal random oracle with a construction satisfying our definition (with an augmented simulator). Specifically, we prove a security preserving (i.e., replacement) theorem akin to those of \cite{TCC:MauRenHol04} and \cite{C:CDMP05} for our adapted notions. 

\begin{definition}[$H$-crooked indifferentiability] 
\label{def:indiff-crooked}
We define the notion of $H$-crooked indifferentiability by the following experiment.

\begin{mdframed}
\begin{center}
  \textsc{Real Execution}
\end{center}
\begin{enumerate}
\item Consider a distinguisher $\sD$ and the following multi-phase real
execution. Initially, the distinguisher $\sD$ commences the first
phase: with oracle access to ideal primitive $H$ the distinguisher
constructs and publishes a \emph{subverted implementation} of $H$;
this subversion is described as a deterministic polynomial time
algorithm denoted $\tilde{H}$. (Recall that the algorithm $\tilde{H}$
implicitly defines a subverted version of $H$ by providing $H$ to
$\tilde{H}$ as an oracle---thus $\tilde{H}^H(x)$ is the value taken by
the subverted version of $H$ at $x$.) Then, a uniformly random string
$R$ is sampled and published.
\item Then the second phase begins involving
a deterministic construction $\rC$: the construction
$\rC$ requires the random string $R$ as input and has oracle access to
$\tilde{H}$ (the crooked version of $H$); explicitly this is the
oracle $x \mapsto \tilde{H}^H(x)$. 
\item Finally, the distinguisher $\sD$,
now with random string $R$ as input and full oracle access to the pair
$(\rC, H)$, returns a decision bit $b$.  Often, we call $\sD$ the
$H$-crooked-distinguisher.
\end{enumerate}
\begin{center}
  \textsc{Ideal execution}
\end{center}
\begin{enumerate}
\item Consider now the corresponding multi-phase ideal execution with the same $H$-crooked-distinguisher $\sD$. The ideal execution introduces a simulator $\cS$ responsible for simulating the behavior of $H$; the simulator is provided full oracle access to the ideal object $\cF$. Initially, the simulator must answer any queries made to $H$ by $\sD$ in the first phase. Then the simulator is given the random string $R$ and the algorithm $\langle \tilde{H}\rangle$ (generated by $\sD$ at the end of the first phase) as input.
\item In the second phase, the $H$-crooked-distinguisher $\sD$, now with random string $R$ as input and oracle access to the alternative pair $(\cF, \cS)$, returns a decision bit $b$.
\end{enumerate}
\end{mdframed} 

We say that construction $\rC$ is
$(n_{\source},n_{\target},q_{\sD},q_{\tilde{H}},r,\epsilon)$-$H$-crooked-indifferentiable
from ideal primitive $\cF$ if there is an efficient simulator $\cS$ so
that for any $H$-crooked-distinguisher $\sD$ making no more than
$q_{\sD}(\secp)$ queries and producing a subversion $\tilde{H}$ making
no more than $q_{\tilde{H}}(\secp)$ queries, the real execution and
the ideal execution are indistinguishable.  Specifically,
\[
  \left| \Pr_{u,R,H} \left[\tilde{H} \leftarrow\sD^H(1^\secp)\ ; \
      \sD^{\rC^{\tilde{H}}(R),H}(1^\secp, R) = 1\right] -
    \Pr_{u,R,\cF} \left[\tilde{H} \leftarrow\sD^H(1^\secp)\ ; \
      \sD^{\cF,\cS_{}^{\cF}(R,\langle \tilde{H}\rangle)}(1^\secp, R) = 1\right]
  \right| \leq \epsilon(\secp)\,.
\]
Here $R$ denotes a random string of length $r(\secp)$ and both
$H: \bool^{n_{\source}} \rightarrow \bool^{n_{\source}}$ and
$\cF: \bool^{n_{\target}} \rightarrow \bool^{n_{\target}}$ denote
random functions where $n_{\source}(\secp)$ and $n_{\target}(\secp)$
are polynomials in the security parameter $\secp$. We let $u$ denote
the random coins of $\sD$. The simulator is efficient in the sense
that it is polynomial in $n$ and the running time of the supplied
algorithm $\tilde{H}$ (on inputs of length $n_{\source}$). See
Figure~\ref{fig:indiff-crooked} for detailed illustration of the last
phase in both real and ideal executions. (While it is not explicitly
captured in the description above, the distinguisher $\sD$ is
permitted to carry state from the first phase to the second phase.)
The notation $C^{\tilde{H}}(R)$ denotes oracle access to the function
$x \mapsto \tilde{H}(x)$.
\end{definition}

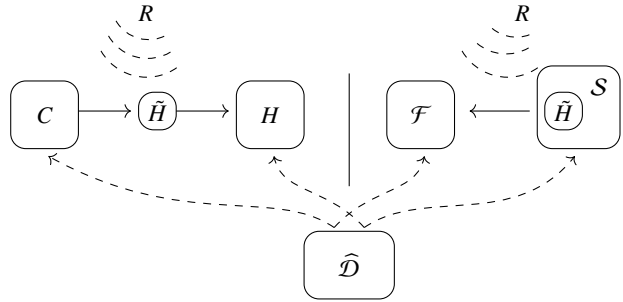
\begin{figure}[h]
  \begin{center}
    \begin{tikzpicture}
            \draw[thin, rounded corners=2mm] (-4.5,1.5) rectangle +(0.9,0.9) node[pos=.5] {$\rC$} node[pos=2] {$R$};
                        \draw[thin, rounded corners=2mm] (-2.8,1.75) rectangle +(.5,.5) node[pos=.5] {$\tilde{H}$};

      \draw[->,thin] (-3.6,2) -- (-2.9,2);
      \draw[->,thin] (-2.3,2) -- (-1.6,2);

      \draw[thin, rounded corners=2mm] (-1.5,1.5) rectangle +(.9,.9) node[pos=.5] {$H$};
      \draw[thin] (0,2.5) -- (0,1);

                \draw[thin, rounded corners=2mm] (.5,1.5) rectangle +(.9,.9) node[pos=.5] {$\cF$} node[pos=2] {$R$};
      \draw[->,thin] (2.4,2) -- (1.6,2);

\draw [thin,dashed] (-3.1,3.2) arc [radius=.5, start angle=210, end angle= 300]; 

\draw [thin,dashed] (-3.2,3) arc [radius=.6, start angle=210, end angle= 310]; 

\draw [thin,dashed] (-3.3,2.8) arc [radius=.7, start angle=210, end angle= 310]; 

\draw [thin,dashed] (1.9,3.1) arc [radius=.5, start angle=230, end angle= 300]; 

\draw [thin,dashed] (1.7,2.9) arc [radius=.6, start angle=230, end angle= 300]; 

\draw [thin,dashed] (1.5,2.8) arc [radius=.7, start angle=210, end angle= 310]; 

      \draw[thin, rounded corners=2mm] (2.5,1.5) rectangle +(1.1,1.1) node[pos=.75] {$\cS$};
                                  \draw[thin, rounded corners=2mm] (2.6,1.75) rectangle +(.5,.5) node[pos=.5] {$\tilde{H}$};
 
\draw[thin,dashed,->] (0.2,0.45) .. controls (1,1) and (2.2,.4) .. (3,1.4);5
      \draw[thin,dashed,->] (-0.2,0.45) .. controls (.3,1) and (.9,.9) .. (1,1.4);
\draw[thin,dashed,->] (0.2,0.45) .. controls (-.3,1) and (-.9,.9) .. (-1,1.4);
      \draw[thin,dashed,->] (-0.2,0.45) .. controls (-1,1) and (-2.2,.4) .. (-4,1.4);

            \draw[thin, rounded corners=2mm] (-.6,-.5) rectangle +(1.2,.9) node[pos=.5] {$\sD$};

    \end{tikzpicture}
  \end{center}
  \caption{
\label{fig:indiff-crooked} 
The $H$-crooked indifferentiability notion: the distinguisher $\sD$, in the first phase, manufactures and publishes a subverted implementation denoted as $\tilde{H}$, for ideal primitive $H$; then in the second phase, a random string $R$ is published; 
after that, in the third phase, algorithm $\rC$, and simulator $\cS$ are developed; 
the $H$-crooked-distinguisher $\sD$, in the last phase,   
either interacting with algorithm $\rC$ and ideal primitive $H$, or with ideal primitive $\cF$ and simulator $\cS$, return a decision bit. 
Here, algorithm $\rC$ has oracle access to $\tilde{H}$, while simulator $\cS$ has oracle access to $\cF$ and $\tilde{H}$.
}
\end{figure}

Our main security proof will begin by demonstrating that in our
particular setting, security in a simpler model suffices: this is the
\emph{abbreviated crooked indifferentiability} model, articulated
below. We then show that---in light of the special structure of our
simulator---it can be effectively lifted to the full model above.

\begin{definition}[Abbreviated $H$-crooked indifferentiability]
  \label{def:abbrev-indiff-crooked}
  The abbreviated model calls for the distinguisher to provide the
  subversion algorithm $\tilde{H}$ at the outset (without the
  advantage of any preliminary queries to $H$). Thus, the abbreviated
  model consists only of the last phase of the full model.  Formally,
  in the abbreviated model the distinguisher is provided as a pair
  $(\sD,\tilde{H})$, the random string $R$ is drawn (as in the full
  model), and insecurity is expressed as the difference between the
  behavior of $\sD$ on the pair $(C^{\tilde{H}}(R),H)$ and the pair
  $(\cF,\cS^\cF(R,\langle\tilde{H}\rangle))$. Specifically, the
  construction $\rC$ is
  $(n_{\source},n_{\target},q_{\sD},q_{\tilde{H}},r,\epsilon)$-Abbreviated-$H$-crooked-indifferentiable
  from ideal primitive $\cF$ if there is an efficient simulator $\cS$
  so that for any $H$-crooked-distinguisher $\sD$ making no more than
  $q_{\sD}(\secp)$ queries and subversion algorithm $\tilde{H}$ making
  no more than $q_{\tilde{H}}(\secp)$ queries, the real execution and
  the ideal execution are indistinguishable:
  \[
    \left| \Pr_{u,R,H} \left[\sD^{\rC^{\tilde{H}}(R),H}(1^\secp, R) = 1\right] - 
      \Pr_{u,R,\cF} \left[\sD^{\cF,\cS_{}^{\cF}(R,\langle\tilde{H}\rangle)}(1^\secp, R) = 1\right] 
    \right| \leq \epsilon(\secp)\,.
  \]
  Here $R$ denotes a random string of length $r(\secp)$ and both
  $H: \bool^{n_{\source}} \rightarrow \bool^{n_{\source}}$ and
  $\cF: \bool^{n_{\target}} \rightarrow \bool^{n_{\target}}$ denote
  random functions where $n_{\source}(\secp)$ and $n_{\target}(\secp)$
  are polynomials in the security parameter $\secp$. We let $u$ denote
  the random coins of $\sD$. The simulator is efficient in the sense
  that it is polynomial in $n$ and the running time of the supplied
  algorithm $\tilde{H}$ (on inputs of length $n_{\source}$).
\end{definition}

Observe that while the abbreviated simulator is a fixed algorithm, its
running time may depend on the running time of $\tilde{H}$---in
particular, the definition permits $\cS$ sufficient running time to
simulate $\tilde{H}$ on a polynomial number of inputs.

Regarding the difference between these notions, observe that the
distinguisher can ``compile into'' the subversion algorithm
$\tilde{H}$ any queries and pre-computation that might have been
advantageous to carry out in phase I; such queries and pre-computation
can also be mimicked by the distinguisher itself. This technique can
effectively simulate the two phase execution with a single
phase. Nevertheless, the models do make slightly different demands on
the simulator which must be prepared to answer some queries (in Phase
I) prior to knowledge of $R$ and $\tilde{H}$.

\paragraph{Replacement with crooked indifferentiability.}

Security preserving (replacement) has been shown in the indifferentiability framework~\cite{TCC:MauRenHol04}: if $\rC^\cG$ is indifferentiable from $\cF$, then $\rC^\cG$ can replace $\cF$ in any cryptosystem, and the resulting cryptosystem in the  $\cG$  model is at least as secure  as that in the $\cF$ model. 
We next show that the replacement property also holds in the crooked indifferentiability framework. 

 Recall that, in the ``standard'' indifferentiability framework~\cite{TCC:MauRenHol04,C:CDMP05},  a cryptosystem can be modeled as an Interactive Turing Machine with an interface to an adversary $\cA$ and to a public oracle. 
There the cryptosystem is run by a  ``standard'' environment $\cE$.  
In our ``crooked'' indifferentiability framework, a cryptosystem has the interface to an adversary $\cA$ and to a public oracle. However, now  the cryptosystem is run by a  crooked-environment $\sE$.

Consider an ideal primitive $\cG$. 
Similar to the $\cG$-crooked-distinguisher, we can define the $\cG$-crooked-environment $\sE$ as follows: 
Initially, the crooked-environment $\sE$ manufactures and then publishes a subverted implementation of the ideal primitive $\cG$, and denotes it $\badG$. Then $\sE$ runs the attacker $\cAtwo$, and 
the cryptosystem $\cP$ is developed.   
In the  $\cG$  model, cryptosystem $\cP$ has oracle access to $\rC$ whereas attacker $\cAtwo$ has oracle access to $\cG$; note that, $\rC$ has oracle access to $\badG$, not to directly $\cG$. In the $\cF$ model, both $\cP$ and $\cAtwo$ have oracle access to $\cF$. Finally, the crooked-environment $\sE$ returns a binary decision output. 
The definition is illustrated in Figure \ref{fig:composition-crooked}.


\begin{definition}
Consider ideal primitives $\cG$ and $\cF$. 
A cryptosystem $\cP$ is said to be at least as secure in the $\cG$-crooked model with algorithm $\rC$ as in the $\cF$ model, if for any $\cG$-crooked-environment $\sE$ and any attacker $\cAtwo$ in the $\cG$-crooked model, there exists an attacher $\cStwo$ in the $\cF$ model, such that: 
$$\Pr[\sE(\cP^{\rC^{\badG}},\cAtwo^{\cG})=1]-\Pr[\sE(\cP^\cF,\cStwo^\cF)=1]\leq\epsilon.$$
where $\epsilon$ is a negligible function of the security parameter $\secp$. 
\end{definition}

\begin{figure}[htbp!]
 \begin{center}
 \begin{tikzpicture}

 \draw[thin, rounded corners=2mm] (-3.5,1.5) rectangle +(.9,.9) node[pos=.5] {$\rC$} node[pos=2] {$R$};
 \draw[thin, rounded corners=2mm] (-2.25,1.75) rectangle +(.5,.5) node[pos=.5] {$\badG$};

 \draw[->,thin] (-1.8,2) -- (-1.5,2);
 \draw[->,thin] (-2.6,2) -- (-2.25,2);

 \draw[thin, rounded corners=2mm] (-1.5,1.5) rectangle +(.9,.9) node[pos=.5] {$\cG$};
 \draw[thin] (0,2.5) -- (0,-2.5);
 
 \draw[thin, rounded corners=2mm] (-3.5,0) rectangle +(.9,.9) node[pos=.5] {$\cP$};

 \draw[->,thin] (-2.5,0.5) -- (-1.6,0.5);
 
 \draw[->,thin] (-3,1) -- (-3,1.4);
 \draw[->,thin] (-1,1) -- (-1,1.4);

\draw [thin,dashed] (-2.1,3.2) arc [radius=.5, start angle=210, end angle= 300]; 

\draw [thin,dashed] (-2.2,3) arc [radius=.6, start angle=210, end angle= 310]; 

\draw [thin,dashed] (-2.3,2.8) arc [radius=.7, start angle=210, end angle= 310]; 

\draw [thin,dashed] (2.1,3.1) arc [radius=.5, start angle=230, end angle= 300]; 

\draw [thin,dashed] (1.9,2.9) arc [radius=.6, start angle=230, end angle= 300]; 

\draw [thin,dashed] (1.7,2.8) arc [radius=.7, start angle=210, end angle= 310]; 

 \draw[thin, rounded corners=2mm] (-1.5,0) rectangle +(.9,.9) node[pos=.5] {$\cAtwo$};

 \draw[thin, rounded corners=2mm] (-3.5,-1.5) rectangle +(3,.9) node[pos=.5] {${\sE}$};
 
 \draw[->,thin] (-2,-1.5) -- (-2,-2.4);


 \draw[->,thin] (-2.5,0.5) -- (-1.6,0.5);
 
 \draw[->,thin] (-3,-.5) -- (-3,-0.1);
 \draw[->,thin] (-1,-.5) -- (-1,-.1);
 


 

 \draw[thin, rounded corners=2mm] (.5,1.5) rectangle +(.9,.9) node[pos=.5] {$\cF$} node[pos=2] {$R$};


 
 \draw[thin, rounded corners=2mm] (.5,0) rectangle +(.9,.9) node[pos=.5] {$\cP$};

 \draw[<->,thin] (-2.5,0.5) -- (-1.6,0.5);
 
 \draw[->,thin] (1,1) -- (1,1.4);
 \draw[->,thin] (3,1) |- (1.6,2);

 \draw[thin, rounded corners=2mm] (2.5,0) rectangle +(.9,.9) node[pos=.5] {$\cS_\cA$};

 \draw[thin, rounded corners=2mm] (0.5,-1.5) rectangle +(3,.9) node[pos=.5] {${\sE}$};
 \draw[->,thin] (2,-1.5) -- (2,-2.4);

 \draw[<->,thin] (1.5,0.5) -- (2.4,0.5);
 
 \draw[->,thin] (1,-.5) -- (1,-0.1);
 \draw[->,thin] (3,-.5) -- (3,-.1);

 \end{tikzpicture}
 \end{center}
 \caption{
\label{fig:composition-crooked} 
The environment $\sE$ interacts with cryptosystem $\cP$ and attacker $\cAtwo$:
In the $\cG$ model (left), $\cP$ has oracle accesses to $\rC$ whereas $\cAtwo$ has oracle accesses to $\cG$; the algorithm $\rC$ has oracle accesses to the subverted $\badG$. 
In the $\cF$ model, both $\cP$ and $\cStwo$ have oracle accesses to $\cF$. 
In addition, in both $\cG$ and $\cF$ models, randomness $R$ is publicly available to all entities.}
\end{figure}
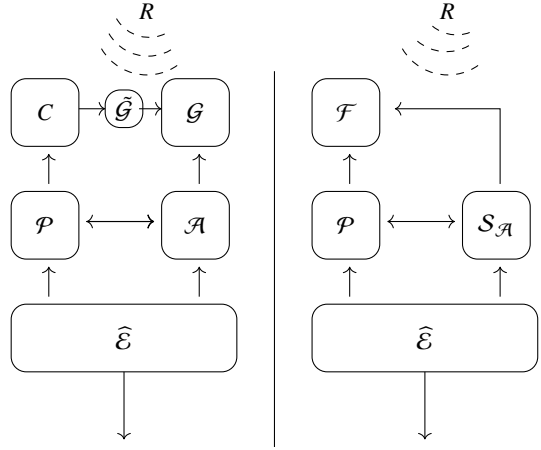



 We now demonstrate the following theorem which 
 shows that security is preserved when replacing an ideal primitive by a crooked-indifferentiable one:

 \begin{theorem} \label{theorem:crooked-composition}
 Consider an ideal primitive $\cG$ and a $\cG$-crooked-environment $\sE$. 
Let $\cP$ be a cryptosystem with oracle access to an ideal primitive $\cF$. Let $\rC$ be an algorithm such that $\rC^{\cG}$ is $\cG$-crooked-indifferentiable from $\cF$. Then cryptosystem $\cP$ is at least as secure in the $\cG$-crooked model with algorithm $\rC$ as in the $\cF$ model. 
\end{theorem}

\begin{proof}
The proof is very similar to that in~\cite{TCC:MauRenHol04,C:CDMP05}. 
Let $\cP$ be any cryptosystem, modeled as an Interactive Turing Machine. Let $\sE$ be any crooked-environment, and $\cAtwo$ be any attacker in the $\cG$-crooked model. In the $\cG$-crooked model, $\cP$ has oracle access to $\rC$ (who has oracle access to $\badG$, not to directly $\cG$.)
 whereas $\cAtwo$ has oracle access to ideal primitive $\cG$; moreover crooked-environment $\sE$ interacts with both $\cP$ and $\cAtwo$. This is illustrated in Figure \ref{fig:composition-crooked-proof} (left part).

Since $\rC$ is crooked-indifferentiable from $\cF$ (see Figure \ref{fig:indiff-crooked}), one can replace $(\rC^{\badG}, \cG)$ by $(\cF, \cS)$ with only a negligible modification of the crooked-environment $\sE$'s output distribution. As illustrated in Figure \ref{fig:composition-crooked-proof}, by merging attacker $\cAtwo$ and simulator $\cS$, one obtains an attacker $\cStwo$ in the $\cF$ model, and the difference in $\sE$'s output distribution is negligible.
\end{proof}

\begin{figure}[htbp!]
 \begin{center}
 \begin{tikzpicture}
 \draw[thin, rounded corners=2mm] (-3.5,1.5) rectangle +(.9,.9) node[pos=.5] {$\rC$} node[pos=2] {$R$};
 \draw[thin, rounded corners=2mm] (-2.25,1.75) rectangle +(.5,.5) node[pos=.5] {$\badG$} ;

 \draw[->,thin] (-1.8,2) -- (-1.5,2);
 \draw[->,thin] (-2.6,2) -- (-2.25,2);

 \draw[thin, rounded corners=2mm] (-1.5,1.5) rectangle +(.9,.9) node[pos=.5] {$\cG$};
 \draw[thin] (0,2.5) -- (0,-2.5) ;
 
 
\draw [thin,dashed] (-2.1,3.2) arc [radius=.5, start angle=210, end angle= 300]; 

\draw [thin,dashed] (-2.2,3) arc [radius=.6, start angle=210, end angle= 310]; 

\draw [thin,dashed] (-2.3,2.8) arc [radius=.7, start angle=210, end angle= 310]; 

\draw [thin,dashed] (2.1,3.1) arc [radius=.5, start angle=230, end angle= 300]; 

\draw [thin,dashed] (1.9,2.9) arc [radius=.6, start angle=230, end angle= 300]; 

\draw [thin,dashed] (1.7,2.8) arc [radius=.7, start angle=210, end angle= 310]; 

 \draw[thin, rounded corners=2mm] (-3.5,0) rectangle +(.9,.9) node[pos=.5] {$\cP$};

 \draw[->,thin] (-2.5,0.5) -- (-1.6,0.5);
 
 \draw[->,thin] (-3,1) -- (-3,1.4);
 \draw[->,thin] (-1,1) -- (-1,1.4);

 \draw[thin, rounded corners=2mm] (-1.5,0) rectangle +(.9,.9) node[pos=.5] {$\cAtwo$};

 \draw[thin, rounded corners=2mm] (-3.5,-1.5) rectangle +(3,.9) node[pos=.5] {${\sE}$};
 
 \draw[->,thin] (-2,-1.5) -- (-2,-2.4);


 \draw[->,thin] (-2.5,0.5) -- (-1.6,0.5);
 
 \draw[->,thin] (-3,-.5) -- (-3,-0.1);
 \draw[->,thin] (-1,-.5) -- (-1,-.1);
 
 \draw[thin, dashed, rounded corners=2mm] (-3.7,-2.1) rectangle +(3.4,3.3) node[pos=.1] {$\sD$};



 \draw[thin, rounded corners=2mm] (.5,1.5) rectangle +(.9,.9) node[pos=.5] {$\cF$} node[pos=2] {$R$};
 \draw[->,thin] (2.4,2) -- (1.6,2);

 \draw[thin, rounded corners=2mm] (2.5,1.5) rectangle +(1.1,1.1) node[pos=.75] {$\cS$};
 \draw[thin, rounded corners=2mm] (2.6,1.75) rectangle +(.5,.5) node[pos=.5] {$\badG$};

 \draw[thin, rounded corners=2mm] (.5,0) rectangle +(.9,.9) node[pos=.5] {$\cP$};

 \draw[<->,thin] (-2.5,0.5) -- (-1.6,0.5);
 
 \draw[->,thin] (1,1) -- (1,1.4);
 \draw[->,thin] (3,1) -- (3,1.4);

 \draw[thin, rounded corners=2mm] (2.5,0) rectangle +(.9,.9) node[pos=.5] {$\cAtwo$};

 \draw[thin, rounded corners=2mm] (0.5,-1.5) rectangle +(3,.9) node[pos=.5] {${\sE}$};
 \draw[->,thin] (2,-1.5) -- (2,-2.4);


 \draw[<->,thin] (1.5,0.5) -- (2.4,0.5);
 
 \draw[->,thin] (1,-.5) -- (1,-0.1);
 \draw[->,thin] (3,-.5) -- (3,-.1);
 
 \draw[thin, dashed, rounded corners=2mm] (.3,-2.1) rectangle +(3.4,3.3) node[pos=.1] {$\sD$};
 
 
 \draw[thin, dashed,red, rounded corners=2mm] (2.3,-.3) rectangle +(1.9,3) node[pos=.85] {$\cStwo$};



 \end{tikzpicture}
 \end{center}
\caption{
\label{fig:composition-crooked-proof} 
Construction of attacker $\cStwo$ from attacker $\cAtwo$ and simulator $\cS$}
\end{figure}
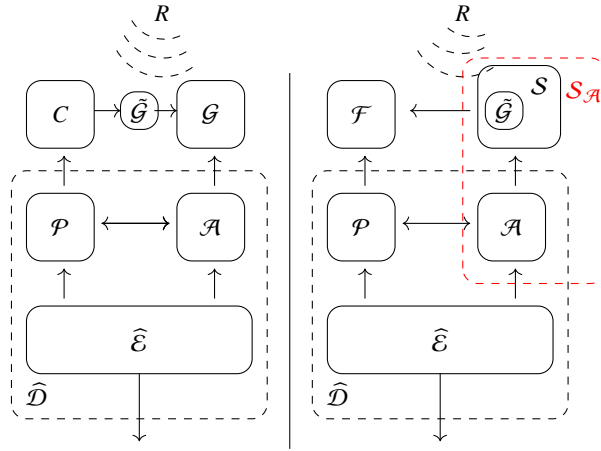



Similar proof can be used to show the following corollary.

\begin{corollary}[Proof of the warm-up construction.]
Let $\cG$, $\cF_1$ and $\cF_2$ be three ideal primitives. Suppose there are two algorithms $\rC_1$ and $\rC_2$ such that $\rC_1^{\cG}$ is crooked-indifferentiable from $\cF_1$ and $\rC_2^{\cF_1}$ is indifferentiable from $\cF_2$. Then $\rC^{\cG}$ is crooked-indifferentiable from $\cF_2$ for $\rC:=\rC_1^{\rC_2}$.
\end{corollary}

\paragraph{Restrictions (of using crooked indifferentiability).}
Ristenpart et al.~\cite{EC:RisShaShr11} has demonstrated that the replacement/composition theorem (Theorem~\ref{theorem:composition}) in the original indifferentiability framework only holds in single-stage settings. We remark that, the same restriction also applies to our replacement/composition theorem (Theorem~\ref{theorem:crooked-composition}). We leave it as our future work to extend our crooked indifferentiability to the multi-stage settings where disjoint adversaries are split over several stages.


\section{Main Result and Technical Overview}
\label{sec:overview}

\subsection{The Construction and Main Result}

From this point on, we use $\cD$ rather than $\sD$ to denote the
distinguisher in our crooked indifferentiability model. For a security
parameter $n$ and a (polynomially related) parameter $\ell$, the
construction depends on public randomness $R = ((a_1,b_1), \ldots, (a_\ell,b_\ell))$.

The source function of the construction is expressed as a family of
$\ell$ independent random oracles:
\[
  F_i: \{0,1\}^{n} \rightarrow \{0,1\}^{n}\,,\qquad\text{for
       $i \in \{1, \ldots, \ell\}$.}
\]

These can be realized as slices of a single random permutation
$F': \{0,1\}^{n'} \rightarrow \{0,1\}^{n'}$, with
$n' = n + \lceil \log \ell+1 \rceil$ by an appropriate convention for
embedding and extracting inputs and values.(Note that $F_i$ may not be permutations.) The family of $\ell$ public, uniform affine-linear
  functions are determined by $R = ((a_1,b_1), \ldots, (a_\ell,b_\ell))$ where $(a_i,b_i) \in \GL(\F_2,n) \times \F_2^n$ for each $i=1,\ldots,\ell$. $a_i$ and $b_i$ are selected independently and uniformly from $\GL(n,\mathbb{F}_2)$ (to be concrete, the
collection of invertible $n \times n$ matrices with elements in $\F_2$) and $\mathbb{F}_2^n$, respectively. The $i$-th affine linear function, defined on
an input $x \in \F_2^n$, is given by the rule
$x \mapsto a_i \cdot x \oplus b$. The final construction is given by the
Feistel construction applied to the round functions
$x \mapsto \tilde{F}_i(a_i \cdot x \oplus b)$, where $\tilde{F}$ is the
subverted version of the function $F_i$. To be concrete, with the data
$(F_i, a_i, b_i)$ (with $i = 1, \ldots, \ell$), the construction
$C: \{0,1\}^{2n} \rightarrow \{0,1\}^{2n}$ is defined by the rule
\begin{align*}
  & C(x_0,x_1):=(x_{\ell},x_{\ell+1})\,,\text{where}\\
  & x_{i+1}:=x_{i-1} \oplus \tilde{F}_i(a_i \cdot x_i \oplus b_i)\text{ , for $i=1,\ldots,\ell$}\,.
\end{align*}
where $n$-bit strings $x$ and $b_i$ are viewed as length $n$ column
vectors, $a_i \cdot x_i$ is the multiplication between matrix $a_i$
and column vector $x_i$, and $\tilde{F}_i(x)$ is the subverted
function value at $(i.x)$ using the subversion algorithm $\cA$.

We wish to show that such a construction is indifferentiable from an actual random permutation (with the proper input/output length).

\begin{theorem}
\label{thm:ind}
We treat a function $F':\{0,1\}^{n'} \rightarrow \{0,1\}^{n'}$, with
$n' = n + \lceil \log \ell+1 \rceil$, as implicitly defining a family
of random oracles
\[
  F_i :\{0,1\}^n\rightarrow\{0,1\}^{n}\,, \qquad\text{for $i>0$,}
\]
by treating $\{0,1\}^{n'} = \{0, \ldots, L-1\} \times \{0,1\}^{n}$
and defining $F_i(x) = F(i,\cdot)$, for $i=0,\ldots,\ell \leq L-1$.
(Output lengths are achieved by removing the appropriate number of
trailing symbols). For convenience, we will use the setting $\ell = 8n$.  Consider
a (subversion) algorithm $\cA$ so that
  it defines a subverted random oracle $\tilde{F}$. Assume that for every $F$ (and every $i$),  
  \begin{equation}\label{eq:equal}
    \Pr_{x \in \{0,1\}^n}[\tilde{F}(i,x) \neq F(i,x)] \leq \epsilon(n) = \negl(n)\,.
  \end{equation}

The above Feistel-based construction is $(n',2n,q_{\cD},q_{\cA},r,\epsilon')$-indifferentiable from a random permutation $P:\{0,1\}^{2n}\rightarrow\{0,1\}^{2n}$, where $\epsilon'=\negl(n)$, $q_{\cD}$ is 
the number of queries made by the distinguisher $\cD$ and $q_{\cA}$ is 
the number of queries made by $\cA$ as in Definition \ref{def:indiff-crooked}. $q_{\cD}$ and $q_{\cA}$ are both polynomial functions of $n$.

\end{theorem}

\paragraph{Remark.} It will be clear later that the construction is still secure if $\ell=8n$ is replaced by $\ell = 2000n/\log(1/\epsilon)$.

To somewhat simplify the notation, we define the function
$\CF_i: \{0,1\}^n \rightarrow \{0,1\}^n$ to be the unsubverted analog
of the round function $\CF_i(x)= F_i(a_i \cdot x \oplus b_i)$ and,
similarly, define $\CFt_i(x)=\tilde{F}_i(a_i \cdot x \oplus b_i)$ to
be actual round function.  Since the function
$x \rightarrow a_i \cdot x \oplus b_i$ is a permutation (note that
$a_i$ is an invertable linear function), reasoning about $\CF_i$ (and
$\CFt_i$, respectively) is effectively equivalent to reasoning about
$F_i$ (and $\CF_i$). For convenience, we will focus on $\CF_i$
($\CFt_i$) for the bulk of the paper (i.e., we will treat the query and
evaluation of $F_i(x)$ as the query and evaluation of $\CF_i(x')$ such
that $x=a_i \cdot x' \oplus b_i$).

When evaluating $\CFt_i(x)$, the subversion algorithm queries $\CF$ at
a set of points of polynomial size. We define the set of these points
to be
\[
  Q_i(x)=\{ (j,x') \mid \text{the evaluation of $\CFt_i(x)$ queries $\CF_j(x')$}\}\,.
\]

\subsection{$2n/\log(1/\epsilon)$ rounds are not enough}

We first show that the above construction is insecure with fewer than $2n/\log(1/\epsilon)$ rounds.

\begin{lemma}\label{Lemma:Attack}
Let $n$ be a positive integer and $\epsilon$ be a real number with $1/8 \geq \epsilon \geq 2^{-n}$. Let $\ell$ be a positive integer not greater than $2n/\log(1/\epsilon)$ and $\lambda$ be the greatest integer that is smaller than or equal to $n/\ell+1$. Consider selecting a uniform element $B \in \mathbb{F}_2^{\lambda\ell/2}$ and a $\lambda\ell/2$ by $n$ matrix $A$ of elements in the field $\mathbb{F}_2$ such that
\begin{itemize}
\item Each row vector $w_i$($i=1,\ldots,\lambda\ell/2$) of $A$ are nonzero.
\item For each $j=1,\ldots,\ell/2$, the $\lambda$ row vectors $w_{j\lambda+1},\ldots,w_{j\lambda+\lambda}$ are linearly independent.
\end{itemize}
Then, over the randomness of the choice of $A$ and $B$,
\[
 \Pr \left [\parbox{8cm}{There does not exist a length $n$ column vector $X$ of elements in the field $\mathbb{F}_2$ such that $A \cdot X=B$.} \right ]=O(n^2 \cdot 2^{-n/4})\,,
\]
where $A \cdot X$ is the multiplication between a matrix and a column vector and $B$ is viewed as a column vector.
\end{lemma}
\begin{proof}
It is easy to see that there exists $X$ such that $A\cdot X=B$ if $A$ has full rank.

For any $i=1,\ldots,\lambda\ell/2$, we denote by $W_i$ the set of first $i$ row vectors of $A$. For a set $S$ of vectors, we use $\langle S \rangle$ to denote the vector space spanned by the elements of $S$.

For the sake of proof, consider the following process of generating each row vectors $w_i$ ($i=1,\ldots,\lambda\ell/2$) of $A$:
\begin{align*}
        & \mkern30mu \textbf{For } \text{$i=1,\ldots,\lambda\ell/2$:}\\
        & \mkern50mu \text{Write $i=j\lambda+k$ for some $j \in \{1,\ldots,\ell/2\}$ and $k \in \{1,\ldots,\lambda\}$ }\\
        & \mkern50mu \textbf{While } \text{$w_i$ is not selected or $w_{i} \in \langle W_{i}/W_{j\lambda} \rangle$.:}\\
        & \mkern70mu \text{Select a nonzero $w_i$ uniformly in $\mathbb{F}_2^n$.}\\
    \end{align*}
    For each $i=1,\ldots,\lambda\ell/2$, we define the event
    \[
      E^1_i:=\{\text{$w_i$ is resampled for at least once in the above generation of $A$}\}
    \]
    and
    \[
      E^2_i:=\{\text{$w_i \in \langle W_{i-1} \rangle$}\}.
    \]

Then the probability in the lemma is bounded by
\begin{align*}
  &  \mkern20mu \Pr[\text{$A$ does not have full rank}]\\
  & \leq \sum_{i=1}^{\lambda\ell/2}\Pr[E^2_i]\\
  & \leq \sum_{i=1}^{\lambda\ell/2}(\Pr[E^2_i \cap \neg E^1_i]+\Pr[E^1_i])\\  
  & \leq \sum_{i=1}^{\lambda\ell/2}((|W_{i-1}|-1)/(2^n-1)+(2^{\lambda}-1)/(2^n-1))\\
  & \leq \sum_{i=1}^{\lambda\ell/2}(2^{i-1}+2^{\lambda})/(2^n)\\
  & < \sum_{i=1}^{\lambda\ell/2}(2^{(n+\ell)/2}+2^{\lambda})/(2^n)\\
  & =O(n^2(2^{3n/4}/2^n))\\
  & =O(n^2 \cdot 2^{-n/4}),   
\end{align*}
where the last inequality relies on the fact that $\ell \leq 2n/\log(1/\epsilon) \leq n/2$ and $\lambda \leq n/\ell+1 \leq n/2+1$. 
\end{proof}

\begin{theorem}
The construction is not crooked-indifferentiable from a random permutation if $\ell \leq 2n/\log(1/\epsilon)$.
\end{theorem}
\begin{proof}
Let $\lambda$ be the greatest integer that is smaller than or equal to $n/\ell+1$ (so $\lambda \geq n/\ell$). Consider the following subversion algorithm $\cA$: for each $F_i$ ($i=1,\ldots,\ell$) and any $n$ bit string $x$, define $\tilde{F}_i(x):=0^n$ if the first $\lambda$ bits of $x$ are $0$s. Otherwise, define $\tilde{F}_i(x):=F_i(x)$. (Notice that this subversion algorithm is legitimate since the dishonest fraction is $2^{-\lambda}\leq 2^{-n/\ell} \leq \epsilon$.)

Now we prove the distinguisher can launch the following attack with the subversion algorithm above. We will show that, with overwhelming probability over the choice of $R$, there is a pair of $n$-bit strings $(x_0,x_1)$ such that for the Feistel chain $(x_1,x\ldots,x_{\ell})$ related to $(x_0,x_1)$, $\CFt_i(x_i)=0^n$ for all $i=1,\ldots,\ell$. (We use the terminology ``with overwhelming probability'' in the paper to mean ``with all but negligible probability.'')

Notice that the fact that such a pair $(x_0,x_1)$ exists is equivalent to the fact that there is a pair $(x_0,x_1)$  such that the first $\lambda$ bits of $a_{2i+1} \cdot x_1 \oplus b_{2i+1}$ and the first $\lambda$ bits of $a_{2j} \cdot x_0 \oplus b_{2j}$ are $0$s for all $0<2i+1,2j\leq \ell$. And this is true with overwhelming probability due to Lemma \ref{Lemma:Attack}. (Also, the attack can be launched by a polynomial running time adversary since the linear equations in Lemma \ref{Lemma:Attack} can be solved efficiently.)
\end{proof}





\subsection{Technical Overviews and Notations}

In this section we give a technical overview of proving Theorem \ref{thm:ind}. For convenience, we will prove the crooked-indifferentiability in the abbreviated model (Definition \ref{def:abbrev-indiff-crooked}) first and then lift it to the full model (\ref{def:indiff-crooked}).  

\noindent\emph{Our strategy: Simulation via judicious preemptive chain
  completion.} To convey the main idea, suppose that a distinguisher
queries the \emph{simulated} round functions in order to determine the
value of the permutation $P$ on input $(x_0,x_1) \in \{0,1\}^{2n}$; in
particular, the resulting output $(x_{8n},x_{8n+1})$ is obtained by
computing $x_{i+1}:=x_{i-1} \oplus \CFt_i(x_i)$ for all
$i=1,\ldots,8n$. Then, of course, $(x_{8n},x_{8n+1})$ must equal the
output of $P$ on input $(x_0,x_1)$: otherwise the distinguisher can
easily detect that it is not interacting with the real Feistel
construction. To ensure such consistency, the simulator must recognize
that the queries $x_1,\ldots,x_{8n}$ belong to an evaluation of $\rC$,
and must set the values $\CF_i(x_i)$ to enforce consistency with
$P$. The mechanism for this is described informally below and
in more detail in the next section.

The simulator maintains an internal table for each function $\CF_i$
that indicates a partial definition of this function: these tables
typically expand during interaction with the distinguisher and satisfy
the invariant that once a $\CF_i$ value is defined in the table for a
particular element $x$ of the domain, this is never removed or altered
later in the computation. While the tables define the $\CF_i$ values
used to respond to any query answered by the simulator, the table may
record additional $\CF_i$ values not as yet queried by the
distinguisher as a bookkeeping tool. Of course, distinguisher queries
are always answered consistently with the values in the tables.

\paragraph{Subverted and unsubverted chains; honest chains.} In the
following, an index $s$, combined with a sequence of values $x_s,\ldots,x_{s+r}$ ($r \geq 1$,
$1 \leq s < s+r \leq 8n$) such that $\CF_i(x_i)$ is defined by the
simulator for all $s \leq i \leq s+r$ and such that
$x_{i+1}:=x_{i-1} \oplus \CF_i(x_i)$ for all $s+1 \leq i \leq s+r-1$,
will be called an \emph{unsubverted chain}. For each index $i$ and an
element $x \in \cS.\CF_i$, we say $\CFt_i(x)$ is \emph{defined} if its
value can be determined by the subversion algorithm and the $\CF$
values that are already defined by the simulator. We assume without
loss of generality that the subversion algorithm always evaluates
$\CF_i(x)$ when called upon to evaluate $\CFt_i(x)$ (for any $i$ and
$x$). Therefore, $\CF_i(x)$ must be defined when $\CFt_i(x)$ is
defined. An index $s$, combined with a sequence of values $x_s,\ldots,x_{s+r}$ ($r \geq 1$,
$1 \leq s < s+r \leq 8n$) such that $\CFt_i(x_i)$ is defined by the
simulator for all $s \leq i \leq s+r$, and such that
$x_{i+1}:=x_{i-1} \oplus \CFt_i(x_i)$ for all $s+1 \leq i \leq s+r-1$,
will be called a \emph{subverted chain}. The \emph{length} $L(\cdot)$
of an unsubverted (or subverted) chain is defined to be the number of
the elements in the chain. An unsubverted (or subverted) chain is
called a \emph{full chain} if it has length $8n$. Note, in general, that
chains always have length at least two (as $r \geq 1$).

For a chain $c=(s,x_s,\ldots,x_{s+r})$, we define
$Q_{c}= \bigcup_{i=s}^{s+r} Q_i(x_i)$ if $\CFt_{i}(x_i)$ is defined for $i=s,\ldots,s+r$. For any $i \in \{1,\ldots,8n\}$
and $x \in \{0,1\}^n$, if $\CFt_i(x)$ is defined, we say $(i,x)$ is
\emph{honest} when $\CF_i(x)=\CFt_i(x)$; similarly, we say it is
\emph{dishonest} when $\CF_i(x) \neq \CFt_i(x)$. We say a subverted
chain is \emph{honest} if all the elements on the chain are honest.

For a chain $c = (s,x_{s},\ldots,x_{s+r})$ and a term $(i,x)$, we say $(i,x)$ is an element of $c$(or $(i,x) \in c$) if $s \leq i \leq s+r$ and $x_i=x$.  For two chains $c_1=(s_1,x_{s_1},\ldots,x_{s_1+r_1})$ and
$c_2=(s_2,y_{s_2},\ldots,y_{s_2+r_2})$, we say $c_1 \subset c_2$ if each element of $c_1$ is also an element of $c_2$. We say $c_1$ and $c_2$ are \emph{disjoint} if there is no chain $c$ for which $c \subset c_1$ and $c \subset c_2$ (i.e., the chains $c_1$ and $c_2$ do not share any pair of adjacent elements).

\paragraph{The definition of the simulator $\mathcal{S}$.}
Our simulation strategy will consider a carefully chosen set of
relevant unsubverted chains as ``triggers'' for completion: once a
chain of this family is defined in the simulator's table, the
simulator will preemptively ``complete'' the chain to ensure
consistency of the resulting full chain with $P$.
Upon a query for $\CF_i$ with input $x_i$ (in fact, the query is a
query for $F_i$ with input $x'_i$ such that
$a_i \cdot x'_i \oplus b_i=x_i$), the simulator sets $\CF_i(x_i)$ to a
fresh random value and looks for new relevant partial chains involving
$x_i$, adding them to a FIFO queue. (In general, many new chains may
be added by this process.) The simulator then repeats the following,
until the queue is empty: It removes the first unsubverted chain from the queue.  If the chain satisfies a certain property (will be described later), the simulator \emph{completes} it to a full subverted chain
$x_1,\ldots,x_{8n}$ in such a way so as to guarantee that
$P(x_0,x_1)=(x_{8n},x_{8n+1})$, where $x_0=x_2 \oplus \CFt_1(x_1)$ and
$x_{8n+1}=x_{8n-1} \oplus \CFt_{8n}(x_{8n})$. In particular, it sets
each undefined $\CF$ in $Q_i(x_i)$ to a fresh uniform random string,
with the exception of two consecutive values $\CF_u(x_u)$ and
$\CF_{u+1}(x_{u+1})$ which are set adaptively to ensure consistency
with $P$. We refer to this step as \emph{adapting} the values of
$\CF_u(x_u)$ and $\CF_{u+1}(x_{u+1})$. Establishing that such adapting
is always possible (for some carefully chosen $u$) will be a major
challenge of our analysis below. 

We now face four main challenges. Our choice of which chains are relevant and how they are completed will be crucial in order to solve them:
\begin{enumerate}\setlength{\itemsep}{-0.1cm}
\item \textbf{Efficiency.} We need to show that the simulation
  terminates with high probability when answering a query; i.e., the
  queue becomes empty after a small (polynomial) number of
  completions.
 \item \textbf{Freshness.} We need to show that the values of
  $\CF_u(x_u)$ and $\CF_{u+1}(x_{u+1})$ are always undefined when
  these values are selected for adaptation.
\item \textbf{Honesty.} We need to show that the values of
  $\CF_u(x_u)$ and $\CF_{u+1}(x_{u+1})$ which are adapted to ensure
  consistency are always honest; i.e., they are always equal to their
  subverted versions, i.e., $\CF_u(x_u)=\CFt_u(x_u)$ and
  $\CF_{u+1}(x_{u+1})=\CFt_{u+1}(x_{u+1})$.

\item \textbf{Indistinguishability.} Finally, with the three
  demands above in hand, it is still necessary to show that the simulated
  world cannot be distinguished from the real world.
\end{enumerate}

\paragraph{Addressing Challenge 1} To see why it is possible the queue may not become empty after a small number of completions, notice that the completion of a certain chain forces the evaluation of many terms that have not been queried by the distinguisher. These newly evaluated terms may generate another chain that triggers completion. The same efficiency problem also appears in the proof of classical indifferentiability of a constant round Feistel construction. (See Coron et al.~\cite{JC:CHKPST16}) The efficiency problem in our case (the crooked-indifferentiability model) is more complex than that in \cite{JC:CHKPST16} (the classical indifferentiability model) because when completing a chain in the crooked-indifferentiability model, the simulator needs to evaluate $\CFt$ instead of just $\CF$ values in the chain, which in general, generates many more terms than the classical model. 

To prove efficiency, we will show that the recursion stops after at most $\poly(q_{\cD})$ steps, where $q_{\cD}$ is the number of the queries made by the distinguisher $\cD$. The proof relies on the observation that, for the chains that are completed, on average, all but a constant number of elements in each chain were once queried by $\cD$. (Notice that not all the elements in these chains are evaluated because they are queried by $\cD$. For example, some elements are evaluated when the simulator completes a chain.) Hence, the total number of the chains that are completed is in fact asymptotically equivalent to $q_{\cD}/\ell$. 

\paragraph{Addressing Challenges 2 and 3} To understand why proving freshness and honesty is hard, consider the following example. During the interaction with the distinguisher, suppose the simulator $\cS$ sees an unsubverted chain $c=(s,x_s,...,x_{s+r})$ that triggers completion. Let us call the current $\cS.CF$ table $T_{\text{Initial}}$. 
For the full chain $c'=(1,x_1,\ldots,x_{8n})$ determined by $c$ ($c \subset c'$), $\cS$ hopes that it can find an index $u$ so that $(u,x_u)$ and $(u+1,x_{u+1})$ are undefined before adaption. It is easy to find an index $u$ so that these two terms are not in $T_{\text{Initial}}$. However, before $\cS$ determines $x_u$ and $x_{u+1}$, it needs to evaluate $\CFt_i(x_i)$ for $i \neq u,u+1$. And there may exist an index $i$ so that $(u,x_u)$ or $(u+1,x_{u+1})$ is in $Q_i(x_i)$, which breaks the freshness. It is also not obvious how to find $u$ so that $(u,x_u)$ and $(u+1,x_{u+1})$ are honest since the distinguisher can subvert the round functions of any index. In our analysis, we will have to find $u,u+1$ such that {\em both} freshness and honesty can be satisfied.

To prove honesty, we will show that, for any term $(i,x_i)$ in $c'$, it is honest if $i$ is much smaller than $s$ or is much greater than $s+r$ (i.e., the term is far away from the initial chain $c$ that triggers completion). Therefore, there is a long subchain of $c'$ that is honest. The simulator will select the index $u$ in this honest area. To prove freshness, we will show that, inside a long enough honest chain $c''$($c'' \subset c'$), for any term $(i,x_i)$ in the ``middle area'' of $c''$, $\CF_i(x_i)$ is not queried by $\CFt_j(x_j)$ for any $j \neq i$. To achieves freshness, the simulator only needs to pick $u$ in the middle part of the honest area.

\paragraph{Addressing Challenge 4} It is still not easy to establish
crooked-indifferentiability after we understand efficiency, honesty,
and freshness. The reason is that the $\CF$ values that are
maintained by $\cS$ are not perfectly uniform conditioned on the
distinguisher's query to the ideal object $P$, which is a crucial
property in the proofs of efficiency, freshness and honesty.

To see why the $\CF$ values held by $\cS$ are not perfectly uniform,
imagine that the distinguisher queries $P(x_0,x_1)$ for some
$(x_0,x_1)$ and then makes several $\CF$ queries to trigger the
completion of the chain corresponding to $(x_0,x_1)$. The two adapted
values $\CF_u(x_u)$ and $\CF_{u+1}(x_{u+1})$ are not uniform because
they are, of course, adapted to maintain consistency.

To break down the proof, we introduce a sequence of game transitions
involving 6 games, beginning with the simulator game (Game 1) and
ending with the construction (Game 6). By mapping the randomness from
one game to another, we prove that the gap between the 6 games is
negligible if the gap between Games 5 and 6 is negligible. (In
particular, we explain how to treat the games as coupled random
variables that can be investigated with the same underlying
randomness; this provides a convenient way to identify differences in
the dynamics and conclusions of the games.) Then we turn our attention
to Game 5, which maintains an explicit, additional table of uniform
$\CF$ values. This table (in Game 5) provides a vantage point from
which all future $\CF$ values are in fact uniform, and simplifies
reasoning about many of the critical events of interest. Finally, we
formally prove honesty and freshness in Game 5 to show the gap
between Game 5 and 6 is negligible.

\smallskip
\noindent{\em Technical differences between \cite{JC:CHKPST16} and this paper.} In \cite{JC:CHKPST16} (the classical indifferentiability model), Coron et al. used a simulation strategy similar to ours—simulation via judicious preemptive chain completion—to demonstrate the classical indifferentiability of a constant round Feistel structure. Despite using similar simulation strategy, there are some significant technical differences between our security proof and the proof in \cite{JC:CHKPST16}.
\begin{enumerate}
\item In the security proof of \cite{JC:CHKPST16}, the authors only need to show freshness and efficiency of the simulation since there is no subversion; they are not required to prove honesty.
\item The efficiency proof in \cite{JC:CHKPST16} is quite straightforward. By contrast, in our case, it is not that obvious how to upper bound the number of the terms generated in the simulation. The difference is due to the existence of the subversion algorithm. In our case, the chains that are completed are subverted chains, while the classical case has no subversion algorithm and therefore only completes ``unsubverted'' chains. The evaluation of a subverted chain generates many more terms than the evaluation of an unsubverted chain, which in general, may generate many more chains that trigger completion.  
\item The proof of freshness is challenging in both \cite{JC:CHKPST16} and our work, but for quite different reasons. The chains in \cite{JC:CHKPST16} are very short(i.e., have only constant length), and when two of them are intersected, the terms of one chain can easily occupy the ``adaption space'' for the other, which hinders freshness. In our case, however, we are not that worried about the intersection of chains since we our construction has many more than constant rounds. The difficulty of our freshness proof arises again from the subversion algorithm: to prove freshness, we need to rule out the case that when completing a chain, the two adapated terms are queried by some previously evaluated $\CFt$.
\end{enumerate} 

\paragraph{Roadmap of the proof} The formal proof consists of four steps:
\begin{enumerate}
\item First, we define the simulator for the security proof in Section \ref{subsec:definition of simulator} and introduce a sequence of intermediate games from Section \ref{subsubsection:game transition preparation} to Section \ref{fifth game} that build connections between the simulator's game and the construction. Using the game transition, we show that the security proof can be reduced to controlling the probability of two bad events (assuming the simulator is efficient). We address challenge 4 here.
\item Second, from Section \ref{Subsec: preparation} to \ref{section:security}, we formally prove the bad events are negligible Section, which corresponds to challenge 2 and 3.
\item Third, we prove the simulator is efficient in Section \ref{section:efficiency}, which is challenge 1. 
\item Finally, we lift the security in the abbreviated model to the full model in Section \ref{section:full model}.
\end{enumerate}


\section{Proof}\label{sec:proof}
\subsection{The Detailed Definition of the Simulator}\label{subsec:definition of simulator}
The simulator provides an interface $\cS.\CF(i,x)$ to query the
simulated random function $\CF_i$ on input $x$. As mentioned above,
for each $i$ the simulator internally maintains a table whose entries
are pairs $(x,y)$ of $n$-bit strings; each such entry intuitively
determines a simulated value of $\CF$ at a particular point: in
particular, if the pair $(x,y)$ appears then any query to
$\cS.\CF(i,x)$ returns the value $y$. The simulator maintains the
natural invariants described previously: responses provided to the
distinguisher are always consistent with the table and, furthermore,
once an entry has been added to the table, it is never removed or
changed. Note that in many cases the table will reflect function
values that have not been queried by the distinguisher. We denote the
$i$th table by $\cS.\CF_i$ and write $x \in \cS.\CF_i$ whenever $x$ is
a preimage in this table, often identifying $\cS.\CF_i$ with the set
of preimages stored. When $x \in \cS.\CF_i$, $\CF_i(x)$ denotes the
corresponding image. We also denote the collection of all these
$\cS.\CF_i$ tables by $\cS.\CF$. We use the notation $(i,x) \in \cS.\CF$ when $x \in \cS.\CF_i$.

For each $i$, we additionally define a table $\cS.\CFt_i$ induced implicitly by $\cS.\CF$.  As with $\cS.CF_i$, the table $\cS.\CFt_i$ consists of pairs of inputs and outputs of $\CFt_i$. We write $x \in \cS.\CFt_i$ when all queries generated by evaluation of $\CFt_i(x)$ are defined in $\cS.\CF$; naturally, the corresponding function value determines the pair $(x,y)$ in the table. The collection of all of these $\cS.\CFt_i$ is denoted by $\cS.\CFt$. (Note that this table is not maintained explicitly by the simulator, but rather determined implicitly by $\cS.\CF$.)

\paragraph{Handling queries to $\cS.\CF$.} On a query $\cS.\CF(i,x)$, the simulator first checks whether $x \in \cS.\CF_i$. If so, it answers with $\CF_i(x)$. Otherwise the simulator picks a random value $y$ and inserts $(x,y)$ into $\cS.\CF_i$. The process above is done by a procedure called $\cS.\CF^{\text{Inner}}$ which takes input $(i,x)$.) After this, the simulator takes further steps to ensure that its future answers are consistent with the permutation $P$. Only after this consistency maintenance step is the value $y$ finally returned.

To ensure consistency, the simulator considers all newly generated unsubverted chains with length $n/10$ that terminate at the
last-queried position; for a newly evaluated term $\CF_s(x_s)$, these
chains of interest either have the form $(s,x_s,...,x_{s+n/10-1})$ or
$(s-n/10+1,x_{s-n/10+1},...,x_s)$. Each such detected chain is enqueued by the
simulator in a ``completion queue,'' idenifying the chain for future
completion. 

The simulator then repeats the following completion step until the
queue is emptied. (When the queue is finally empty, the simulator
returns the answer $y$ to the initial query.) 
\begin{enumerate}
\item \textbf{Detection Step.} The first chain $c=(s,x_s,...,x_{s+n/10-1})$ is removed from the queue. A procedure called $\cS.\text{HonestyCheck}$ is then run on the chain. The procedure $\cS.\text{HonestyCheck}$ evaluates $\CFt$ values of the elements of $c$ and generates a four-tuple $(s,x_s,x_{s+1},u)$ for future completion if all the elements in $c$ are honest. (In fact, not all chains removed from the queue are processed by $\cS.\text{HonestyCheck}$. A chain removed from the queue is processed by $\cS.\text{HonestyCheck}$ only if it is disjoint with all the chains that are previously processed by $\cS.\text{HonestyCheck}$ and is disjoint with all the previously completed full subverted chains. Any chain that is not processed by $\cS.\text{HonestyCheck}$ is discarded. The procedure that decides whether a chain is going to be discarded or processed by $\cS.\text{HonestyCheck}$ is called $\cS.\text{Check}$.) In the tuple $(s,x_s,x_{s+1},u)$, the value $s$ ensures that later the simulator
knows that the first value $x_s$ corresponds to $\CF_s$. The value $u$
describes where to adapt (that is, program) the values of $\CF_*$ in
order to ensure consistency with the given permutation: this will
occur at positions $u$ and $u+1$. The convention for determining $u$
is straightforward: If $s > 5n$ or $s+n/10-1 < 3n$, then there is
``plenty of space around $4n$,'' and $u=4n$; otherwise, $u=7n$.
\item \textbf{Completion Step.} Finally, the simulator takes the four-tuple $(s,x_s,x_{s+1},u)$ and \emph{completes} the subverted chain related to $(s,x_s,x_{s+1})$. Intuitively, this means that the chain is determined
by iteratively determining neighbouring values of $\CFt(x)$ by
evaluating the subversion algorithm and, when necessary, carrying out
internal calls to $\CF_i()$ in order to answer queries made by that
algorithm to the $F_i$. This iterative process is continued, using $P$
to ``wrap around,'' until the only remaining undetermined values
appear at positions $u$ and $u+1$; at this point, the values at $u$
and $u+1$ are programmed to ensure consistency. In more detail: Assuming that $u < s$, the completion process (conducted by a procedure called $\cS.\text{Complete}$) proceeds as follows.
\begin{itemize}[topsep=1pt]
\item The initial chain consists of the two adjacent values
  $x_s, x_{s+1}$.
\item $\CFt_{s+1}(x_{s+1})$ is determined by simulating the subversion algorithm which generates oracle queries to $\CF$ to be answered using $\cS.\CF$. (Note that this process may enqueue new chains for completion.) The value $x_{s+2} = x_s \oplus \CFt_{s+1}(x_{s+1})$ is then determined, yielding the enlarged chain $(x_s, x_{s+1}, x_{s+2})$. This process is repeated until the chain is extended maximally ``to the right'' so that it has the form $(x_s, x_{s+1}, \ldots, x_{8n}, x_{8n+1})$.
\item $P^{-1}$ is then applied to $x_{8n},x_{8n+1}$ to yield $x_0, x_1$.
\item Starting from the pair $(x_0, x_1)$, this process is repeated,
  as above, to yield values for $x_2, \ldots, x_u$. Note that
  $x_u = x_{u-2} \oplus \CFt(x_{u-1})$ so that $\CFt(x_u)$ is never
  evaluated during this process (which is to say that the subversion
  algorithm is never simulated on $x_u$).
\item Similarly, the original pair $x_s, x_{s-1}$ is extended ``to the
  left'' to determine the values $x_{s-1}, ..., x_{u+1}$; as above,
  $x_{u+1}$ is determined by $x_{u+3} \oplus \CFt(x_{u+2})$, so that
  $\CFt(x_{u+1})$ is never evaluated.
\item Then, the simulator defines $\CF_u(x_u)$ and
  $\CF_{u+1}(x_{u+1})$ that is consistent with $P$, i.e., $\CF_u(x_u):=x_{u-1} \oplus x_{u+1}$ and
  $\CF_{u+1}(x_{u+1}):=x_{u} \oplus x_{u+2}$. The game aborts if
  either of these is defined from a previous action of $\cS$. If the game does not abort, the simulator evaluates the
  subversion algorithm on both $x_u$ and $x_{u+1}$. During this
  evaluation, the values $\CF_u(x_u)$ and $\CF_{u+1}(x_{u+1})$ are
  already determined; other queries are answered using $\cS.\CF$ as
  above. The game aborts if $(u,x_u)$ and $(u+1,x_{u+1})$ is dishonest; otherwise, the chain is a valid subverted chain (and consistent with $P$).
\item {A set $\cS.\text{CompletedChains}$ is maintained to store the chains that are completed: for any $(i,x_i,x_{i+1})$ ($1 \leq i \leq 8n-1$), $\cS$ updates
    \[
      \cS.\text{CompletedChains}:=\cS.\text{CompletedChains} \cup (i,x_i,x_{i+1}).
    \]}
\end{itemize}
The alternative case, when $u > s+1$, is treated analogously.
\end{enumerate}

For a detailed description of the simulator, please see Appendix on page \pageref{Game 1}.

\subsection{Game Transition Approach: Some preparations}\label{subsubsection:game transition preparation}

Our overall purpose is to show that for any deterministic
distinguisher $\cD$ that make at mosts $q_{\cD}$ queries (where
$q_{\cD}$ is some polynomial function in $n$), the probability that
$\cD$ outputs 1 when interacting with $(P,\cS^P)$ differs negligibly
from the probability it outputs 1 when interacting with
$(C^F,F)$. Here $C$ is the construction in Section~\ref{sec:overview}, and $F$ is a collection of $8n$ uniform
functions. We also wish to establish that, with overwhelming probability,
only a polynomial number of terms are evaluated by $\cS$ (or $P$) when
$\cD$ interacts with $(P,\cS^P)$.

We denote the game where the distinguisher $\cD$ interacts with
$(P,\cS^P)$ by $G_1$ (the ideal world), and the game where $\cD$ interacts with
$(C^F,F)$ by $G_6$ (the real world). We will introduce four intermediate games, $G_2$,
$G_3$, $G_4$ and $G_5$, in the following narrative to study the
relationship between $G_1$ and $G_6$. When we use the term ``$G_i$''
($i=1,...,6$), we always have a fixed deterministic distinguisher
(denoted by $\cD$) in mind, without mentioning it explicitly. Also,
whenever we prove a statement about ``$G_i$,'' this is understood to
mean that the statement holds every fixed distinguisher that issues at
most $q_{\cD}$ queries in the setting defined by $G_i$ (including the
queries to $\CF$ and to the ideal object).

In the description of $G_1$, we defined various concepts (e.g.,
unsubverted/subverted chains, honest, $Q_i(x_i)$, etc.) that are used
to describe the behavior of the simulator $\cS$ (or its table
$\cS.\CF$). In the rest of the paper, we likewise will apply these
concepts to describe the simulators in other games. To avoid
confusion, we will specify the simulator (or its counterpart) we are
working with when using these concepts.

In the following, we say a game is \emph{efficient} if, with overwhelming probability, only a polynomial size of terms are evaluated by the simulator or the ideal object when the game ends. To prove $G_1$ is efficient and is indistinguishable from $G_6$, our plan is to show the following statements:
\begin{enumerate}\setlength{\itemsep}{-0.1cm}
\item $G_1$ vs $G_2$: assuming $G_2$ is efficient, $G_1$ is efficient and is indistinguishable from $G_2$.
\item $G_2$ vs $G_3$: assuming $G_3$ is efficient, $G_2$ is efficient and the total variation distance between the transcript of $G_2$ and $G_3$ is 0.
\item $G_3$ vs $G_4$: assuming $G_4$ is efficient, the total variation distance between the transcript of the distinguisher in $G_3$ and $G_4$ is bounded by the probability of the bad events in $G_4$, and $G_3$ is efficient if the probability of the bad events in $G_4$ is negligible.
\item $G_4$ vs $G_5$: assuming $G_5$ is efficient, $G_4$ is efficient and the total variation distance between the transcript of $G_4$ and $G_5$ is 0.
\item $G_5$ vs $G_6$: the total variation distance between the transcript of $G_5$ and $G_6$ is bounded by the probability of the bad events in $G_5$.
\item Bounding bad events: assuming $G_5$ is efficient, the probability of the bad events in $G_4$ and $G_5$ is negligible.
\item Efficiency: $G_5$ is efficient.
\item Full model: The crooked-indifferentiability in the abbreviated model can be lifted to the full model.
\end{enumerate} 
It is easy to see how these results together imply the efficiency of simulator $\cS$ in $G_1$ and the crooked-indifferentiability of our construction ($G_1$ is indistinguishable from $G_6$). Statements 1 to 5 will be formally defined and proved as we introduce new games. Statement 6, 7 and 8 will be proved in \ref{section:security}, \ref{section:efficiency} and \ref{section:full model}, respectfully.

\subsection{The Second Game}
\subsubsection{Description of the Second Game}
To obtain $G_2$ from $G_1$, we replace the random permutation $P$ by a
``two-sided random function'' $RF$. $RF$ is a system that provides
answers to ``$RF$ queries'' and ``$RF^{-1}$ queries.'' Ideally, the
system would maintain a consistent, partially defined bijection that
is extended independently and uniformly whenever a ``fresh value'' is
requested. These constraints are, of course, impossible to guarantee
perfectly beyond the first query, so the system adopts a particular
convention for subsequent answers that yields a probability
distribution that is nearly indistinguishable from a random
permutation so long as the number of queries is
polynomial. Specifically, the system maintains an initially empty
table with entries of the form $(\downarrow, \alpha, \beta)$,
intuitively indicating that $RF(\alpha) = \beta$, or the form
$(\uparrow, \alpha, \beta)$, intuitively indicating that
$RF^{-1}(\beta) = \alpha$. Under expected circumstances, these tables
will provide the guarantees mentioned above, defining a consistent,
partially defined permutation. A query of the form $RF(\alpha)$ is
answered as follows: (i.) if there is a tuple
$(\downarrow, \alpha, \beta)$ in the table, $\beta$ is returned, (ii.)
otherwise, a uniformly random value $\beta$ is drawn and both
$(\downarrow, \alpha, \beta)$ and $\uparrow, \alpha, \beta)$ are added
to the table and, (ii.') in the event that there is a tuple
$(\uparrow, \alpha', \beta)$ in the table for some
$\alpha \neq \alpha'$, this is removed. Queries to $RF^{-1}$ are
handled similarly. Of course, it is natural to expect and easy to
prove that such collisions occur only with negligible probability. In
the absense of such collisions $RF$ behaves like a random permutation.

See page \pageref{Game 2} for the pseudocode description of $G_2$.

\subsubsection{The Gap Between the First Game and the Second Game}

To understand the gap between $G_1$ and $G_2$, we notice that the only difference between the two games is the ideal object. The following lemma shows that the two ideal objects, $P$ and $RF$ are indistinguishable for a distinguisher that issues polynomial queries.

\begin{lemma}[$P$ and $RF$ are indistinguishable]\label{Lemma: P vs RF}
  A distinguisher $\cD'$ that issues $q'$ queries
  to either $P$ or $RF$ has advantage at most $(q')^2/2^{2n}$ to distinguish the two objects.
\end{lemma}
\begin{proof}
We denote the transcripts of $\cD'$ interacting with $P$ and $RF$ by $T_{P}$ and $T_{RF}$ respectively. We say $T_{P}$ ($T_{RF}$) is equal to $t$ for a certain bitstring $t$ if it reflects the transcript correctly. We say a query $A(\alpha)$ ($A=P,P^{-1},RF,RF^{-1}$) is \emph{redundant} if $\cD'$ has queried $A(\alpha)$ or $A^{-1}(\beta)$ ($=\alpha$) before. In the proof, without loss of generality, we assume $\cD'$ does not make any redundant queries.

We bound the advantage of the distinguisher by bounding the total variation distance between the two scenarios. We say a bitstring $t$ is \emph{good} if when $T_{RF}=t$, there is no overwrite in the execution of $RF$. $t$ is called \emph{bad} if it is not good. We notice that if $t$ is good it behaves like a transcript of a permutation. Then, for any good $t$,
\[
   \Pr[T_{P}=t]=\prod_{i=1}^{q'} \frac{1}{2^{2n}-i+1}>\frac{1}{2^{2nq'}}=\Pr[T_{RF}=t].
\]

Let $\alpha_{P}$ and $\alpha_{RF}$ denote the distribution of the $T_P$ and $T_{RF}$. The above observation shows that
\[
 \|\alpha_{P}-\alpha_{RF}\|_{\tv}\leq \Pr[\text{$T_{RF}=t$ for a bad $t$}].
\]

Now we proceed to bound the probability that $T_{RF}$ is equal to a bad string. By definition, $T_{RF}$ is bad if and only if there is an overwrite in the execution of $RF$. Since $RF$ are queried for $q'$ times, the probability that there is an overwrite is less than $(q')^2/2^{2n}$.
\end{proof}

\begin{lemma}[The Gap between $G_1$ and $G_2$]\label{Lemma: Game 1 and 2 are indis.}
If $G_2$ is efficient, then $G_1$ is efficient and the probability of $\cD$ outputting 1 in $G_2$ differs negligibly from that in $G_1$.
\end{lemma}
\begin{proof}
The lemma is implied directly by Lemma \ref{Lemma: P vs RF}.
\end{proof}

\subsection{The Third and the Fourth Game}

We will put the descriptions of $G_3$ and $G_4$ together since they have similar structures. For convenience, we will introduce $G_4$ first and then make some changes to $G_4$ to arrive at $G_3$.

\subsubsection{Description of the Fourth Game}

In $G_4$, we adopt a new simulator $\cO^2$ with two inbound parties $\cS^2$ and $\cM^2$. (In the rest of the paper, we will use different superscripts for the simulators and parties in different games.) $\cS^2$ is a direct analogue of $\cS$ in $G_2$. $\cM^2$ is a party that knows the distinguisher's query to the ideal object $RF$. In other words, the simulator $\cO^2$ in $G_4$ is a "public" simulator, which means it answers the $\CF$ queries from $\cD$ and at the same time, knows all queries made by $\cD$ to $RF$. 

$\cO^2$ maintains two tables of $\CF$ values (denoted by $\cS^2.\CF$ and $\cM^2.\CF$), which have the same format as $\cS.\CF$ in $G_2$. The $\CF$ values in the table $\cS^2.\CF$ are evaluated in a way such that it is a subset of $\cM^2.\CF$ unless the game aborts: roughly speaking, when $\cO^2$ evaluates a term $\CF_i(x)$ in $\cS^2.\CF$, it inserts the same value to the table $\cM^2.\CF$; when $\cO^2$ needs to evaluate a previously undefined term $\CF_i(x)$ in $\cS^2.\CF$ and $(i,x) \in \cM^2.\CF$, it typically copies $\CF_i(x)$ in $\cM^2.\CF$ to $\cS^2.\CF$. 

Like $\cS$ in $G_2$, both $\cS^2$ and $\cM^2$ have two important procedures: the ``evaluation procedure'' $\cS^2.\CF^{\text{Inner}}$ ($\cM^2.\CF^{\text{Inner}}$, respectively) takes care of query to the table $\cS^2.\CF$ ($\cM^2.\CF$, respectively), and the ``completion procedure'' $\cS^2.\text{Complete}$ ($\cM^2.\text{Complete}$, respectively) is used to \emph{complete} chains in $\cS^2.\CF$ ($\cM^2.\CF$, respectively). 

The two procedures of $\cM^2$, $\cM^2.\CF^{\text{Inner}}$ and $\cM^2.\text{Complete}$, only insert values into the table $\cM^2.\CF$, while the two procedures of $\cS^2$ can insert values to both $\cS^2.\CF$ and $\cM^2.\CF$. For convenience, in the following description of $\cM^2$'s procedures, we will just say a certain term $\CF(i,x)$ is evaluated instead of saying it is evaluated in the table $\cM^2.\CF$. However, in the description of the procedures of $\cS^2$, we will mention clearly whether $\CF$ is inserted into $\cS^2.\CF$ or $\cM^2.\CF$.

\paragraph{Handling $\cD$'s query to the ideal object.} When the distinguisher makes a query to $RF$ on a $2n$-bit string $(x_0,x_1)$, $\cO^2$ does nothing if $\cS^2$ or $\cM^2$ has queried $RF(x_0,x_1)$ or $RF^{-1}(x_{8n},x_{8n+1})$ (for some $(x_{8n},x_{8n+1})$ with $RF^{-1}(x_{8n},x_{8n+1})=(x_0,x_1)$) before. Otherwise, $\cM^2$ proceeds to \emph{complete} a chain in $\cM^2.\CF$. In more detail: Assuming that the distinguisher queries $RF(x_0,x_1)$, the ``completion procedure'' $\cM^2.\text{Complete}$ proceeds as follows.

\begin{itemize}
\item $\CFt_{1}(x_{1})$ is determined by simulating the subversion algorithm. This will generate oracle queries to $\CF_*()$ which are answered using $\cM^2.\CF$. (The way a query to $\cM^2.\CF$ is answered will be explained later.) The value $x_{2} = x_0 \oplus \CFt_{1}(x_{1})$ is then determined, yielding the chain $(1, x_{1}, x_{2})$. This process is repeated until the chain is extended ``to the right'' to $x_{4n}$ so that it has the form $(x_1, x_{2}, \ldots, x_{4n-1}, x_{4n})$. Note that
  $x_{4n} = x_{4n-2} \oplus \CFt(x_{4n-1})$ so that $\CFt(x_{4n})$ is never
  evaluated during this process (which is to say that the subversion
  algorithm is never simulated on $x_u$).
\item $RF$ is then applied to $x_{0},x_{1}$ to yield $x_{8n}, x_{8n+1}$.
\item Similarly, the pair $(x_{8n}, x_{8n+1})$ is extended ``to the
  left'' to determine the values for $x_{4n+1}, \ldots, x_{8n-1}$; as above,
  $x_{4n+1}$ is determined by $x_{4n+3} \oplus \CFt(x_{4n+2})$, so that
  $\CFt(x_{4n+1})$ is never evaluated.
\item Then, $\cM^2$ defines $\CF_{4n}(x_{4n})$ and
  $\CF_{4n+1}(x_{4n+1})$ in such a way that consistency with $P$ is
  ensured, i.e., $\CF_{4n}(x_{4n}):=x_{4n-1} \oplus x_{4n+1}$ and
  $\CF_{4n+1}(x_{4n+1}):=x_{4n} \oplus x_{4n+2}$. The game aborts if
  either of these is defined from a previous action of the
  simulator. If the game does not abort, $\cM^2$ evaluates the
  subversion algorithm on both $x_{4n}$ and $x_{4n+1}$. During this
  evaluation, the values $\CF_{4n}(x_{4n})$ and $\CF_{4n+1}(x_{4n+1})$ are
  already determined; other queries are answered using $\cM^2.\CF$ as
  above. The game aborts if for any $i$ with $3n \leq i \leq 5n$, $(i,x_i) \in \bigcup_{j=1}^{8n}Q_j(x_j)/Q_i(x_i)$ or $(i,x_i)$ is dishonest; otherwise, the chain is a valid subverted chain (and
  consistent with $RF$).
\item A set $\cM^2.\text{CompletedChains}$ is maintained to store the chains that are completed: for any $(i,x_i,x_{i+1})$ ($1 \leq i \leq 8n-1$), $\cM^2$ updates $\cM^2.\text{CompletedChains}:=\cM^2.\text{CompletedChains} \cup (i,x_i,x_{i+1}).$
\item A set $\cM^2.\text{MiddlePoints}$ is maintained to store the points with index between $3n$ and $5n$: for any $(i,x_i)$ with $3n \leq i \leq 5n$, $\cM^2$ updates $\cM^2.\text{MiddlePoints}:=\cM^2.\text{MiddlePoints} \cup (i,x_i).$
\item A set $\cM^2.\text{AdaptedPoints}$ is maintained to store the points with index $4n$ and $4n+1$: for $(i,x_i)$ with $i=4n$ or $4n+1$, $\cM^2$ updates $\cM^2.\text{AdaptedPoints}:=\cM^2.\text{AdaptedPoints} \cup (i,x_i).$
\end{itemize}

$\cO^2$ reacts to the distinguisher's query to $RF^{-1}(\cdot)$ similarly.

\paragraph{Handling queries to $\cM^2.\CF$.} On a query $\cM^2.\CF(i,x)$, $\cM^2$ implements the evaluation procedure $\cM^2.\CF^{\text{Inner}}$:
\begin{itemize}
\item First, $\cM^2$ checks whether $x \in \cS^2.\CF_i$. If so, it answers with $\cS^2.\CF_i(x)$.
\item If $x \notin \cS^2.\CF_i$, $\cM^2$ checks whether $x \in \cM^2.\CF_i$ and $(i,x) \notin \cM^2.\text{MiddlePoints}$. If so, it answers with $\cM^2.\CF_i(x)$.
\item If $x \notin \cS^2.\CF_i$, $x \in \cM^2.\CF_i$ and $(i,x) \in \cM^2.\text{MiddlePoints}$, the game aborts.
\item If $x \notin \cS^2.\CF_i$ and $x \notin \cM^2.\CF_i$, $\cM^2$ picks a random value $y$ and inserts $(x,y)$ into $\cM^2.\CF_i$.
\end{itemize}

Notice that in the first case above, we do not need to check whether $x \in \cM^2.\CF_i$ because $\cS^2.\CF$ is a subset of $\cM^2.\CF$ (which will be clear later).

\paragraph{Handling $\cD$'s query to $\CF$}
In $G_1$ and $G_2$, we view the distinguisher's queries to $\CF$ as queries to the simulator's table $\cS.\CF$. Similarly, in $G_4$, the distinguisher's query to $\CF$ is viewed as a query to the table $\cS^2.\CF$. See below for a description of how $\cO^2$ handles queries to $\cS^2.\CF$.

\paragraph{Handling queries to $\cS^2.\CF$.} On a query $\cS^2.\CF(i,x)$, $\cS^2$ implements the procedure $\cS^2.\CF^{\text{Inner}}$:
\begin{itemize}
\item First, $\cS^2$ checks whether $x \in \cS^2.\CF_i$. If so, it answers with $\cS^2.\CF_i(x)$. 
\item If $x \notin \cS^2.\CF_i$, $\cS^2$ checks whether $x \in \cM^2.\CF_i$ and $(i,x) \notin \cM^2.\text{MiddlePoints}$. If so, it assigns and answers with $\cS^2.\CF_i(x):=\cM^2.\CF_i(x)$. {(This is the reason that $\cS^2.\CF$ is always a subset of $\cM^2.\CF$.)}
  
\item If $x \notin \cS^2.\CF_i$, $x \in \cM^2.\CF_i$ and $(i,x) \in \cM^2.\text{MiddlePoints}$, the game aborts.
\item If $x \notin \cS^2.\CF_i$ and $x \notin \cM^2.\CF_i$, $\cS^2$ picks a random value $y$ and inserts $(x,y)$ into $\cS^2.\CF_i$ and $\cM^2.\CF_i$.
\item After this, $\cS^2$ takes further steps to ensure that its future answers are consistent with the permutation $RF$. Only after this consistency maintenance step is the value $y$ finally returned.
\end{itemize}

\paragraph{Completing chains in $\cS^2.\CF$.}
For any input $(s,x_s,x_{s+1},u)$, the procedure $\cS^2.\text{Complete}$ proceeds as follows.
\begin{itemize}
\item if $(s,x_s,x_{s+1}) \notin \cM^2.\text{CompletedChains}$
\begin{itemize}
\item $\cS^2.\text{Complete}$ first generates the two sequences $(x_0,\ldots,x_{u})$ and $(x_{u+1},\ldots,x_{8n+1})$ in the way as  $\cS.\text{Complete}$ does in $G_2$.
\item Then, $\cS^2$ defines $\CF_u(x_u):=x_{u-1} \oplus x_{u+1}$ and
  $\CF_{u+1}(x_{u+1}):=x_{u} \oplus x_{u+2}$ in both $\cS^2.\CF$ or $\cM^2.\CF$. The game aborts if
  either of these is defined previously in $\cS^2.\CF$ or $\cM^2.\CF$. If the game does not abort, $\cS^2$ evaluates the
  subversion algorithm on both $x_u$ and $x_{u+1}$. During this
  evaluation, the values $\CF_u(x_u)$ and $\CF_{u+1}(x_{u+1})$ are
  already determined; other queries are answered using $\cS^2.\CF$ as
  above. The game aborts if for $i=u$ or $u+1$, $(i,x_i) \in \bigcup_{j=1}^{8n}Q_j(x_j)/Q_i(x_i)$ or $(i,x_i)$ is dishonest.
\end{itemize}
\item if $(s,x_s,x_{s+1}) \in \cM^2.\text{CompletedChains}$ and $u=7n$, the game aborts.
\item if $(s,x_s,x_{s+1}) \in \cM^2.\text{CompletedChains}$ and $u=4n$, $\cS^2$ copies $Q_c$ from $\cM^2.\CF$ to $\cS^2.\CF$, where $c$ is the full subverted chain containing $(s,x_s,x_{s+1})$ in $\cM^2.\CF$. The game aborts if for either $i=u$ or $u+1$, $(i,x_i)$ was in $\cS^2.\CF$ before this execution of $\cS^2.\text{Complete}$, $(i,x_i) \in \bigcup_{j=1}^{8n}Q_j(x_j)/Q_i(x_i)$, or $(i,x_i)$ is dishonest.
\item A set $\cS^2.\text{CompletedChains}$ is maintained to store the chains that are completed: for any $(i,x_i,x_{i+1})$ ($1 \leq i \leq 8n-1$), the simulator updates $\cS^2.\text{CompletedChains}:={\cS^2.\text{CompletedChains}} \; \cup (i,x_i,x_{i+1})$.
\end{itemize}

See page \pageref{Game 4} for the pseudocode description of $G_4$.

\subsubsection{Description of the Third Game}

$G_3$ is a small pivot from $G_4$ that does not abort in some cases
where $G_4$ aborts. For clarity, we denote by $\cO^1$ the simulator in
$G_3$ and by $\cS^1$ and $\cM^1$ the two parties of $\cO^1$. The
parties of $\cO^1$ maintain similar tables and sets to their
counterparts in $G_4$. (The formal notation for these tables and sets
is determined by changing the identifier in the notations from $\cS^2$ to
$\cS^1$ and from $\cM^2$ to $\cM^1$.)

The procedures in $G_3$ are similar to their counterparts in $G_4$ except that they abort in fewer cases. To aid the reader, we color the differences between $G_4$ and $G_3$ in red below.

\paragraph{The completion procedure $\cM^1.\text{Complete}$.} $\cM^1.\text{Complete}$ is same as its counterpart in $G_4$ except that it never aborts:
\begin{itemize}
\item When $\cM^1.\text{Complete}$ programs $\CF_{4n}(x_{4n})$ and
  $\CF_{4n+1}(x_{4n+1})$, \textcolor{red}{it does not abort even if
    either $\CF_{4n}(x_{4n})$ or $\CF_{4n+1}(x_{4n+1})$ is previously
    defined}. It just assigns
  $\CF_{4n}(x_{4n}):=x_{4n-1} \oplus x_{4n+1}$ and
  $\CF_{4n+1}(x_{4n+1}):=x_{4n} \oplus x_{4n+2}$ (this may overwrite
  the value).
\item Also, {\color{red}the game does not abort even if for any $i$ with $3n \leq i \leq 5n$, $(i,x_i) \in \bigcup_{j=1}^{8n}Q_j(x_j)/Q_i(x_i)$ or $(i,x_i)$ is dishonest}.   
\end{itemize}

\paragraph{The evaluation procedure $\cM^1.\CF^{\text{Inner}}$.} On a query $\cM^1.\CF(i,x)$:
\begin{itemize}
\item First, $\cM^1$ checks whether $x \in \cS^1.\CF_i$. If so, it answers with $\cS^1.\CF_i(x)$. 
\item If $x \notin \cS^1.\CF_i$, $\cM^1$ checks whether $x \in \cM^1.\CF_i$ and $(i,x) \notin \cM^1.\text{MiddlePoints}$. If so, it answers with $\cM^1.\CF_i(x)$.
\item If $x \notin \cS^1.\CF_i$, $x \in \cM^1.\CF_i$ and $(i,x) \in \cM^1.\text{MiddlePoints}$, \textcolor{red}{it answers with $\cM^1.\CF_i(x)$}.
\item If $x \notin \cS^1.\CF_i$ and $x \notin \cM^1.\CF_i$, $\cM^1$ picks a random value $y$ and inserts $(x,y)$ into $\cM^1.\CF_i$.
\end{itemize}

\paragraph{The evaluation procedure $\cS^1.\CF^{\text{Inner}}$.} On a query $\cS^1.\CF(i,x)$:
\begin{itemize}
\item First, $\cS^1$ checks whether $x \in \cS^1.\CF_i$. If so, it answers with $\cS^1.\CF_i(x)$. 
\item If $x \notin \cS^1.\CF_i$, $\cS^1$ checks whether $x \in \cM^1.\CF_i$ and \textcolor{red}{$(i,x) \notin \cM^1.\text{AdaptedPoints}$}. If so, it assigns and answers with $\cS^1.\CF_i(x):=\cM^1.\CF_i(x)$.
\item If $x \notin \cS^1.\CF_i$, $x \in \cM^1.\CF_i$ and \textcolor{red}{$(i,x) \in \cM^1.\text{AdaptedPoints}$, $\cS^1$ picks a random value $y$ and inserts $(x,y)$ into $\cS^1.\CF_i$}.
\item If $x \notin \cS^1.\CF_i$ and $x \notin \cM^1.\CF_i$, $\cS^1$ picks a random value $y$ and inserts $(x,y)$ into $\cS^1.\CF_i$ and $\cM^1.\CF_i$.
\item After this, $\cS^1$ takes further steps to ensure that its future answers are consistent with the permutation $RF$. Only after this consistency maintenance step is the value $y$ finally returned.
\end{itemize}

\paragraph{The completion procedure $\cS^1.\text{Complete}$.} For any input $(s,x_s,x_{s+1},u)$, the procedure $\cS^1.\text{Complete}$ proceeds as follows.
\begin{itemize}
\item if $(s,x_s,x_{s+1}) \notin \cM^1.\text{CompletedChains}$
\begin{itemize}
\item $\cS^1.\text{Complete}$ first generates the two sequences $(x_0,\ldots,x_{u})$ and $(x_{u+1},\ldots,x_{8n+1})$ in the same way as  $\cS.\text{Complete}$ does in $G_2$.
\item Then, $\cS^1$ defines $\CF_u(x_u):=x_{u-1} \oplus x_{u+1}$ and
  $\CF_{u+1}(x_{u+1}):=x_{u} \oplus x_{u+2}$ in both $\cS^1.\CF$ or $\cM^1.\CF$. \textcolor{red}{(This may overwrite the values in $\cM^1.\CF$.) The game aborts if
  either of these was defined previously in $\cS^1.\CF$}. If the game does not abort, $\cS^1$ evaluates the subversion algorithm on both $x_u$ and $x_{u+1}$. During this
  evaluation, the values $\CF_u(x_u)$ and $\CF_{u+1}(x_{u+1})$ are
  already determined; other queries are answered using $\cS^1.\CF$ as
  above. The game aborts if for $i=u$ or $u+1$, $(i,x_i) \in \bigcup_{j=1}^{8n}Q_j(x_j)/Q_i(x_i)$ or $(i,x_i)$ is dishonest.
\end{itemize}
\item if $(s,x_s,x_{s+1}) \in \cM^1.\text{CompletedChains}$, {\color{red}the procedure proceeds in the same way as the last case}.
\item A set $\cS^1.\text{CompletedChains}$ is maintained to store the chains that are completed: for any $(i,x_i,x_{i+1})$ ($1 \leq i \leq 8n-1$), the simulator updates $\cS^1.\text{CompletedChains}:=\cS^1.{\text{CompletedChains}} \; \cup (i,x_i,x_{i+1})$.
\end{itemize} 

See page \pageref{Game 3} for the pseudocode description of $G_3$.

\subsubsection{The Gap Between the Second and the Third Game}

We define the \emph{status} of a game $G_i$ ($i=1,\ldots,5$) at some point (if the game does not abort at this moment) to be the collection of the transcript of the distinguisher, the tables of the simulator and the table of the ideal object at this moment. {(To make the definition work for $G_1$, we imagine that a table of the ideal object $P$ is maintained in $G_1$ to store the evaluated terms of $P$.)} We say the status of the game is equal to ``Abort'' at some moment if the game already aborts at this moment. Note that in $G_i$ ($i=1,\ldots,6$) we always assume a fixed distinguisher $\cD$ that issues at most $q_{\cD}$ queries.

We say the status of $G_i$ ($i=1,\ldots,5$) at some moment $t_1$ is same as the status of $G_j$ ($j=1,\ldots,5$, $j \neq i$) at some moment $t_2$ if $G_4$ at $t_1$ and $G_5$ at $t_2$ have the same transcripts of the distinguisher, the same $RF$ tables, and the same simulator tables. (When we compare the simulator tables between different games, we only compare a table with its counterpart in another game. For example, when we compare the status of $G_2$ and $G_3$, we only compare $\cS.\CF$ to $\cS^1.\CF$ and ignore $\cM^1.\CF$.)

\begin{lemma}[The gap between $G_2$ and $G_3$]
If $G_3$ is efficient, then $G_2$ is efficient and the probability of $\cD$ outputting 1 in $G_3$ equals that in $G_2$.
\end{lemma}
\begin{proof}
To understand the gap between the transcripts of $\cD$ in $G_2$ and $G_3$, it is sufficient to compare the tables $\cS.\CF$ and $\cS^1.\CF$ since $\cD$ queries $\cS$ ($\cS^1$) in $G_2$ ($G_3$) for $\CF$ values. 

First, we compare the evaluation procedures $\cS.\CF^{\text{Inner}}$
and $\cS^1.\CF^{\text{Inner}}$ between the two games. The evaluation
procedure $\cS.\CF^{\text{Inner}}$ in $G_2$ always selects $\CF$
uniformly. A quick check reveals that $\cS^1.\CF^{\text{Inner}}$ also
evaluates $\CF$ uniformly because all the values it takes from table
$\cM^1.\CF$ were selected by $\cM^1$ uniformly. Then, the lemma
follows by the fact that the procedure $\cS.\text{Complete}$ has the
same abort condition with $\cS^1.\text{Complete}$ with respect to the
status of table $\cS.\CF$ and $\cS^1.\CF$.
\end{proof}

\subsubsection{The Gap Between the Third and the Fourth Game}

To analyze the relationship between $G_3$ and $G_4$, we define the following two bad events in $G_4$.

The first bad event happens when $G_4$ aborts during the execution of the completion procedure $\cS^2.\text{Complete}$ or $\cM^2.\text{Complete}$, which means $\cS^2$ or $\cM^2$ fails to complete a chain. Remember that the simulators fail to complete a chain when some terms on the subverted chain they are evaluating have been evaluated before or are dishonest.

The second bad event happens when $G_4$ aborts during the execution of the evaluation procedure $\cS^2.\CFt^{\text{Inner}}$ or $\cS^2.\CFt^{\text{Inner}}$, which means when $\cS^2$ or $\cM^2$ wants to set a uniform $\CF$ value to a certain term $(i,x)$, $\CF(i,x)$ is a middle point of $\cM^2$, and it has not been put in the table $\cS^2.\CF$.

The formal names and definitions of the bad events are:

\begin{align*}  
  \BadComplete_4 &=\left \{\parbox{8cm}{$G_4$ aborts during the execution of the procedure $\cS^2.\text{Complete}$ or $\cM^2.\text{Complete}$} \right \}\,,\\
  \BadEval_4 &=\left \{\parbox{8cm}{$G_4$ aborts during the execution of the procedure $\cS^2.\CFt^{\text{Inner}}$ or $\cM^2.\CFt^{\text{Inner}}$} \right \}\,.
\end{align*}

\paragraph{Remark.} Notice that the completion procedure $\cS^2.\text{Complete}$ ($\cM^2.\text{Complete}$) sometimes calls the evaluation procedure $\cS^2.\CFt^{\text{Inner}}$ ($\cM^2.\CFt^{\text{Inner}}$), so the completion procedure may abort because its \text{Inner} evaluation procedure aborts. We stress that we categorize this event as $\BadEval_4$.

It is clear that a bad event ($\BadComplete_4$ or $\BadEval_4$) happens if and only if the game $G_4$ aborts.

\begin{lemma}[The gap between $G_3$ and $G_4$]
The probability of $\cD$ outputting 1 in $G_4$ differs by at most $\Pr[\BadComplete_4]+\Pr[\BadEval_4]$ from that in $G_3$. Moreover, if $G_4$ is efficient and $\Pr[\BadComplete_4]+\Pr[\BadEval_4]$ is negligible, $G_3$ is efficient. 
\end{lemma}
\begin{proof}
From the descriptions of $G_3$ and $G_4$, we can see that there are two differences between the two games:
\begin{itemize}
\item In $G_3$, when $\cS^1$ completes a chain that has been completed by $\cM^1$ before and $u=4n$, it completes the chain as usual (as if it is not completed by $\cM^1$). However, in the same situation, $\cS^2$ will just copy $Q_c$ to $\cS^2.\CF$, where $c=(1,x_1,\ldots,x_{8n})$ is the full subverted chain completed by $\cM^2$.
\item $G_3$ does not abort in several cases where $G_4$ aborts.
\end{itemize}

We want to show gap between the two games caused the above two differences are bounded by the probability that $G_4$ aborts.

\begin{itemize}
\item The first difference causes zero gap between the two games unless $G_4$ aborts. Imagine that $G_3$ and $G_4$ have the same status (assuming the status of $G_4$ is not ``Abort'') at some point and $\cS^1.\text{Complete}$ ($\cS^2.\text{Complete}$) is going to complete a chain that has been completed by $\cM^1$ ($\cM^2$). In this case, $\cS^2.\text{Complete}$ copies $Q_c$ to $\cS^2.\CF$. We want to show that what $\cS^1.\text{Complete}$ does is equivalent to copying $Q_c$ to $\cS^1.\CF$. For any $(i,x) \in Q_c$ ($(i,x) \neq (4n,x_{4n})$ or $(4n+1,x_{4n+1})$), if $\cS^1$ assigns a different $\CF_i(x)$ in $\cS^1.\CF$ than in $\cM^1.\CF$, then by definition of $\cS^1.\CF^{\text{Inner}}$, $(i,x)$ is an adapted point that is not in $\cS^1.\CF$, which violates the assumption that $G_4$ does not abort. For $(i,x)=(4n,x_{4n})$ or $(4n+1,x_{4n+1})$, the values of $\CF_i(x)$ are equal in $\cS^1.\CF$ and $\cM^1.\CF$ by programming rules.
\item The second difference also causes zero gap between the two games unless $G_4$ aborts, which is obviously true.
\end{itemize}

As a result, the total variation distance between the transcripts of the distinguisher in $G_3$ and $G_4$ is smaller than $\Pr[\BadComplete_4]+\Pr[\BadEval_4]$, the probability that $G_4$ aborts.  

To show the second claim, we observe that if $G_4$ is efficient and aborts with negligible probability, $G_3$ will be efficient since the total variation distance between $G_3$ and $G_4$ is negligible.
\end{proof}

\subsection{The Fifth Game}\label{fifth game}
To obtain $G_5$ from $G_4$, we update the simulator $\cO^2$ to $\cO^3$ which has two parties $\cM^3$ and $\cS^3$.

$\cO^3$ has a significantly different structure than the simulators in
previous games. In the previous games, when the simulators complete a
chain, they query $RF$ first and then adapt two particular terms
($(u,x_u)$ and $(x_{u+1},u+1)$) on the chain to ensure consistency. In
contrast, in $G_5$ we imagine the ideal object $RF$ as simulated by
$\cM^3$ and $\cS^3$. During the execution of the completion procedure,
$\cM^3$ and $\cS^3$ set all the $\CF$ values of the target chain
uniformly and then program the $RF$ value to ensure consistency.

\paragraph{The evaluation procedures.}
The evaluation procedures $\cM^3.\CF^{\text{Inner}}$ and $\cS^3.\CF^{\text{Inner}}$ in $G_5$ are the same as their counterparts in $G_4$.

\paragraph{Handling $\cD$'s query to the ideal object.} When the distinguisher makes a query to $RF$ on a $2n$-bit string $(x_0,x_1)$, if $\cS^3$ or $\cM^3$ has evaluated $RF(x_0,x_1)$ or $RF^{-1}(x_{8n},x_{8n+1})$ (for some $(x_{8n},x_{8n+1})$ with $RF^{-1}(x_{8n},x_{8n+1})=(x_0,x_1)$) before, $\cO^3$ answers correspondingly. Otherwise, $\cM^3$ proceeds to \emph{complete} a chain in $\cM^3.\CF$. In more detail: Assuming that the distinguisher queries $RF(x_0,x_1)$, the ``completion procedure'' $\cM^3.\text{Complete}$ proceeds as follows. (Similar to the situation in $\cM^2.\text{Complete}$, all the $\CF$ evaluations made by $\cM^3.\text{Complete}$ are in the table $\cM^3.\CF$.)
\begin{itemize}
\item $\CFt_{1}(x_{1})$ is determined by simulating the subversion algorithm. This will generate oracle queries to $\CF_*()$ which are answered using $\cM^3.\CF$. The value $x_{2} = x_0 \oplus \CFt_{1}(x_{1})$ is then determined, yielding the chain $(1, x_{1}, x_{2})$. This process is repeated until the chain is extended maximally ``to the right'' so that it has the form $(x_1, x_{2}, \ldots, x_{8n-1}, x_{8n})$. Set $x_{8n+1}:=x_{8n-1} \oplus \CFt_{8n}(x_{8n}).$
\item Set $RF(x_0,x_1):=(x_{8n}, x_{8n+1})$ and return $(x_{8n}, x_{8n+1})$ to $\cD$.
\item The game aborts if for any $i=u$ or $u+1$, $(i,x_i)$ was in $\cM^3.\CF$ before this execution of $\cM^3.\text{Complete}$, $(i,x_i) \in \bigcup_{j=1}^{8n}Q_j(x_j)/Q_i(x_i)$, or $(i,x_i)$ is dishonest.
\item A set $\cM^3.\text{CompletedChains}$ is maintained to store the chains that are completed: for any $(i,x_i,x_{i+1})$ ($1 \leq i \leq 8n-1$), $\cM^3$ updates $\cM^3.\text{CompletedChains}:=\cM^3.\text{CompletedChains}\; \cup (i,x_i,x_{i+1})$.
\item A set $\cM^3.\text{MiddlePoints}$ is maintained to store the points with index between $3n$ and $5n$: for any $(i,x_i)$ with $3n \leq i \leq 5n$, $\cM^3$ updates $\cM^3.\text{MiddlePoints}:=\cM^3.\text{MiddlePoints} \cup (i,x_i)$.
\item A set $\cM^3.\text{AdaptedPoints}$ is maintained to store the points with index $4n$ and $4n+1$: for $(i,x_i)$ with $i=4n$ or $4n+1$, $\cM^3$ updates $\cM^3.\text{AdaptedPoints}:=\cM^3.\text{AdaptedPoints} \cup (i,x_i)$.
\end{itemize}

\paragraph{The completion procedure $\cS^3.\text{Complete}$.}
For input $(s,x_s,x_{s+1},u)$, the completion procedure $\cS^3.\text{Complete}$ (assuming $(s,x_s,x_{s+1}) \notin \cM^3.\text{CompletedChains}$) proceeds as follows.
\begin{itemize}
\item The initial chain consists of the two adjacent values
  $x_s, x_{s+1}$.
\item $\CFt_{s+1}(x_{s+1})$ is determined by simulating the subversion algorithm. This will generate oracle queries to $\CF_*()$ which are answered using $\cS^3.\CF$. (Note that this process may enqueue new chains for completion.) The value $x_{s+2} = x_s \oplus \CFt_{s+1}(x_{s+1})$ is then determined, yielding the enlarged chain $(x_s, x_{s+1}, x_{s+2})$. This process is repeated until the chain is extended maximally ``to the right'' so that it has the form $(x_s, x_{s+1}, \ldots, x_{8n}, x_{8n+1})$.
\item Similarly, the original pair $x_s, x_{s-1}$ is extended ``to the
  left'' to determine the values $x_{s-1}, ..., x_{0}$.
\item Set $RF(x_0,x_1)=(x_{8n},x_{8n+1})$.  
\item The game aborts if for either $i=u$ or $u+1$, $(i,x_i)$ was in $\cS^3.\CF$ or $\cM^3.\CF$ before this execution of $\cS^3.\text{Complete}$, $(i,x_i) \in \bigcup_{j=1}^{8n}Q_j(x_j)/Q_i(x_i)$, or $(i,x_i)$ is dishonest.
\item A set $\cS^3.\text{CompletedChains}$ is maintained to store the chains that are completed: for any $(i,x_i,x_{i+1})$ ($1 \leq i \leq 8n-1$), the simulator updates $\cS^3.\text{CompletedChains}:=\cS.\text{CompletedChains} \cup (i,x_i,x_{i+1}).$
\end{itemize}

The case when $(s,x_s,x_{s+1}) \in \cM^3.\text{CompletedChains}$ is taken care of in the same way as in $G_4$.

See page \pageref{Game 5} for the pseudocode description of $G_5$.

\subsubsection{The Gap Between the Fourth and the Fifth Game}

We wish to show that the distinguisher can not distinguish between $G_4$ and $G_5$. To prove this claim, we note that these two games behaves differently only when they complete chains: when completing a chain, $G_4$ queries the ideal object $RF$ and adapts the $\CF$ values of $(x_u,x_{u+1})$ to ensure consistency; in contrast, $G_5$ assigns $\CF$ of $x_u$ and $x_{u+1}$ uniformly, and then programs $RF$. In the following lemma, we are going to show that these two different conventions of completing a chain will yield the same distribution of the status of the game.

\begin{lemma}[The gap between $G_4$ and $G_5$]\label{Lemma:4 vs 5}
If $G_5$ is efficient, then $G_4$ is efficient and the probability of $\cD$ outputting 1 in $G_5$ equals that in $G_4$.
\end{lemma}
\begin{proof}
It is sufficient to show that, for any $0 < k \leq q_{\cD}$, the total variation distance between the status of $G_4$ and $G_5$ is 0 at the end of $k$-th round of interaction between the distinguisher and the simulators.

Consider a proof by induction. Denote by $\alpha^s_4$ and $\alpha^s_5$ the distribution of the status of $G_4$ and $G_5$ at the end of the $s$-th interaction. For any $0 < k \leq q_{\cD}$, assume $\alpha^s_4=\alpha^s_5$ when $s=k-1$. We will show 
\[
 \|\alpha^k_4-\alpha^k_5\|_{\tv}=0.
\]

\begin{itemize}
\item[-] If $G_4$ ($G_5$) already aborts before the end of the $k-1$-th round of the game, then $G_4$ ($G_5$) also aborts at the end of the $k$-th round and $\|\alpha^t_4-\alpha^t_5\|_{\tv}=0$.
\item[-] If $G_4$ ($G_5$) does not abort before the end of the $k-1$-th round, and the $k$-th query made by $\cD$ does not activate the procedures $\cS^2.\text{Complete}$ or $\cM^2.\text{Complete}$ ($\cS^3.\text{Complete}$ or $\cM^3.\text{Complete}$), it is obvious that $\|\alpha^k_4-\alpha^k_5\|_{\tv}=0$.
\item[-] If $G_4$ ($G_5$) does not abort before the end of the $k-1$-th round, and the $k$-th query made by $\cD$ activates the procedures $\cS^2.\text{Complete}$ or $\cM^2.\text{Complete}$ ($\cS^3.\text{Complete}$ or $\cM^3.\text{Complete}$), the situation is more complicated. We will deal with this case in the rest of the proof.
\end{itemize}

To show $\|\alpha^k_4-\alpha^k_5\|_{\tv}=0$ in the last case, it
suffices to prove that if the ``initial statuses'' before the
execution of the completion procedures are identical in $G_4$ and
$G_5$, the ``resulting statuses'' after the execution are identically
distributed. Formally speaking, assume that at some moment of $G_4$
and $G_5$, the two games have the same status and begin to execute
procedure $\cM^2.\text{Complete}$ and
$\cM^3.\text{Complete}$. (Without loss of generality, we assume this
execution is activated by a query from $\cD$ to $RF(x_0,x_1)$ for some
bitstring pair $(x_0,x_1)$.) We will show the distribution of
$\alpha_4$ and $\alpha_5$ at the end of the execution are
identical. The proof of the two games executing
$\cS^2.\text{Complete}$ and $\cS^3.\text{Complete}$ are omitted
because it is the same as the proof for $\cM^2.\text{Complete}$ and
$\cM^3.\text{Complete}$.

To prove the above claim, we introduce the following transition of the
four different completion procedures. $\cM^2.\text{Complete}$ and
$\cM^3.\text{Complete}$ are the first and the last procedure. Two
middle procedures are used to build their connections. We assume the
four procedures share a common initial table of the simulator,
$T_{\text{initial}}$. ($T_{\text{initial}}$ is $\cM^2.\CF$ in $G_4$ and
$\cM^3.\CF$ in $G_5$.) Our goal is to prove that the distributions of the
resulting table, $T_{\text{final}}$, are identical among the four
procedures. Roughly speaking, the job of a completion procedure is to
generate a sequence of terms $(i,x_i)$ (for $i=1,\ldots,8n$) that are
supposed to form a full subverted chain, a pair of two bitstrings
$((x_0,x_1),(x_{8n},x_{8n+1}))$ such that
$RF(x_0,x_1)=(x_{8n},x_{8n+1})$, and checks whether the game needs to
abort in the process. Our proof will focus mainly on the distribution
of these variables and when the game aborts.

\begin{itemize}
\item \textbf{Procedure 1}: Procedure 1 is $\cM^2.\text{Complete}$.
\begin{enumerate}
\item To generate $(i,x_i)$ (for $i=2,\ldots,u$), define $x_i:=x_{i-2} \oplus \CFt_{i-1}(x_{i-1})$ for $2 \leq i \leq u$ (each $\CF$ is evaluated uniformly).
\item Query $RF$ at $(x_0,x_1)$ and receive an uniform pair of $n$-bit strings $(x_{8n},x_{8n+1})$. To generate $(i,x_i)$ ($i=u,\ldots,8n$), define $x_{i-2}:=x_{i} \oplus \CFt_{i-1}(x_{i-1})$ for $u+3 \leq i \leq 8n+1$ (each $\CF$ is evaluated uniformly).
\item Define $\CF_u(x_u):=x_{u-1} \oplus x_{u+1}$ and $\CF_{u+1}(x_{u+1}):=x_{u} \oplus x_{u+2}$. The game aborts if there is an index $j$ such that $3n \leq j \leq 5n$ and $(j,x_j)$ is in $T_{\text{initial}}$ or $\bigcup_{i=1}^{8n}Q_i(x_i)/Q_j(x_j)$.
\item Evaluate $\CFt_u(x_u)$ and $\CFt_{u+1}(x_{u+1})$; the game aborts if there is an index $j$ such that $3n \leq j \leq 5n$ and $(j,x_j)$ is dishonest.
\end{enumerate}
\item \textbf{Procedure 2}: 
\begin{enumerate}
\item Same as step 1 of Procedure 1.
\item Select $x_{u+1}$ and $x_{u+2}$ uniformly. To generate $(i,x_i)$ ($i=u+3,\ldots,8n+1$), define $(i,x_i):=x_{i-2} \oplus \CFt_{i-1}(x_{i-1})$ for $u+3 \leq i \leq 8n+1$ (each $\CF$ is evaluated uniformly). Assign $RF(x_0,x_1)=(x_{8n},x_{8n+1})$.
\item Same as step 3 of Procedure 1.
\item Same as step 4 of Procedure 1.
\end{enumerate}
\item \textbf{Procedure 3}: Procedure 3 is a small pivot from Procedure 2. 
\begin{enumerate}
\item Same as step 1 of Procedure 2.
\item Select $\CF_u(x_u)$ uniformly and define $x_{u+1}:=x_{u-1} \oplus \CF_{u}(x_{u})$. The game aborts if $\CF_u(x_u)$ is previously assigned. Select $\CF_{u+1}(x_{u+1})$ uniformly and define $x_{u+2}:=x_{u} \oplus \CF_{u+1}(x_{u+1})$. The game aborts if $\CF_{u+1}(x_{u+1})$ is previously assigned. To generate the sequence $(i,x_i)$ (for $i=u+3,\ldots,8n+1$), define $x_i:=x_{i-2} \oplus \CFt_{i-1}(x_{i-1})$ for $u+3 \leq i \leq 8n+1$ (each $\CF$ is evaluated uniformly). Assign $RF(x_0,x_1)=(x_{8n},x_{8n+1})$.
\item The game aborts if there is an index $j$ such that $3n \leq j \leq 5n$ and $(j,x_j)$ is in $T_{\text{initial}}$ or $\bigcup_{i=1}^{8n}Q_i(x_i)/Q_j(x_j)$.
\item Same as step 4 of Procedure 2.
\end{enumerate}
\item \textbf{Procedure 4}: Procedure 4 is $\cM^3.\text{Complete}$.
\begin{enumerate}
\item To generate $(i,x_i)$ (for $i=2,\ldots,8n+1$), define $x_i:=x_{i-2} \oplus \CFt_{i-1}(x_{i-1})$ for $2 \leq i \leq 8n+1$ (each $\CF$ is evaluated uniformly). The game aborts if there is an index $j$ such that $3n \leq j \leq 5n$ and $(j,x_j)$ is in $T_{\text{initial}}$ or $\bigcup_{i=1}^{8n}Q_i(x_i)/Q_j(x_j)$.
\item The game aborts if there is an index $j$ such that $3n \leq j \leq 5n$ and $(j,x_j)$ is dishonest.
\item Assign $RF(x_0,x_1)=(x_{8n},x_{8n+1})$.
\end{enumerate}
\end{itemize}

For $i=1,2,3,4$, We denote the distribution of $T_{\text{final}}$ in Procedure $i$ by $\alpha_{P_i}$. We will prove $\|\alpha_{P_1}-\alpha_{P_4}\|_{\tv}=0$ by showing that $\|\alpha_{P_j}-\alpha_{P_{j+1}}\|_{\tv}=0$ for $j=1,2,3$. 

To see why $\|\alpha_{P_1}-\alpha_{P_2}\|_{\tv}=0$, we rewrite Procedure 1 and 2 as:
\begin{itemize}
\item \textbf{Procedure 1'}: Procedure 1' is a rewrite of Procedure 1.
\begin{enumerate}
\item For all $x \in \{0,1\}^n$ and $u+2 \leq i \leq 8n$, evaluate $\CFt_i(x)$ (each $\CF$ is evaluated uniformly).
\item Same as step 1 of Procedure 1.
\item Select an uniform pair of $n$-bit strings $(x_{8n},x_{8n+1})$ and assign $RF(x_0,x_1)=(x_{8n},x_{8n+1})$. To generate $(i,x_i)$ (for $i=u+1,\ldots,8n$), define $x_{i-2}:=x_{i} \oplus \CFt_{i-1}(x_{i-1})$ for $u+3 \leq i \leq 8n+1$.
\item Define $\CF_u(x_u):=x_{u-1} \oplus x_{u+1}$ and $\CF_{u+1}(x_{u+1}):=x_{u} \oplus x_{u+2}$. For all $x \in \{0,1\}^n$ and $1 \leq i \leq 8n$, remove $\CF_i(x)$ from the table if $\CF_i(x)$ is not in $T_{\text{initial}}$ and is not in $Q_j(x_j)$ for any $1 \leq j \leq 8n$ ($j \neq u,u+1$). The game aborts if there is an index $j$ such that $3n \leq j \leq 5n$ and $(j,x_j)$ is in $T_{\text{initial}}$ or $\bigcup_{i=1}^{8n}Q_i(x_i)/Q_j(x_j)$.
\item Same as step 4 of Procedure 1.
\end{enumerate}
\item \textbf{Procedure 2'}: Procedure 2' is a rewrite of Procedure 2.
\begin{enumerate}
\item Same as step 1 of Procedure 1'.
\item Same as step 2 of Procedure 1'.
\item Select $x_{u+1}$ and $x_{u+2}$ uniformly. To generate $(i,x_i)$ (for $i=u+3,\ldots,8n+1$), define $x_i:=x_{i-2} \oplus \CFt_{i-1}(x_{i-1})$ for $u+3 \leq i \leq 8n+1$. Assign $RF(x_0,x_1)=(x_{8n},x_{8n+1})$.
\item Same as step 4 of Procedure 1'.
\item Same as step 5 of Procedure 1'.
\end{enumerate}
\end{itemize}

Notice that the step 3 of Procedure 1' is equivalent to that of Procedure 1' because the Feistel structure gives a permutation of $2n$-bit strings: selecting an uniform ``input'' string $(x_{u+1},x_{u+2})$ is equivalent to selecting an uniform ``output'' string $(x_{8n},x_{8n+1})$. Therefore, $\|\alpha_{P_1}-\alpha_{P_2}\|_{\tv} = \|\alpha_{P_{1'}}-\alpha_{P_{2'}}\|_{\tv} = 0$, where $\alpha_{P_{1'}}$ and $\alpha_{P_{2'}}$ are the distributions of $T_{\text{final}}$ in Procedure 1' and 2'.

The fact that $\|\alpha_{P_2}-\alpha_{P_3}\|_{\tv}=0$ and $\|\alpha_{P_3}-\alpha_{P_4}\|_{\tv}=0$ are clear from the definitions of Procedures 2,3, and 4.
\end{proof}

Similar to $G_4$, we define the following two bad events in $G_5$:

\begin{align*}  
  \BadComplete_5 &=\left\{\parbox{8cm}{$G_5$ aborts during the execution of the procedure $\cS^3.\text{Complete}$ or $\cM^3.\text{Complete}$} \right\}\,,\\
  \BadEval_5 &=\left\{\parbox{8cm}{$G_5$ aborts during the execution of the procedure $\cS^3.\CFt^{\text{Inner}}$ or $\cS^3.\CFt^{\text{Inner}}$} \right\}\,.
\end{align*}

\paragraph{Remark.} Same as the bad events in $G_4$, if $\cS^3.\text{Complete}$ or $\cM^3.\text{Complete}$ aborts because its \text{inner} procedure $\cS^3.\CFt^{\text{Inner}}$ or $\cS^3.\CFt^{\text{Inner}}$ aborts, we say $\BadEval_5$ happens rather than $\BadComplete_5$ happens. 

It is clear that a bad event ($\BadComplete_5$ or $\BadEval_5$) happens if and only if the game $G_5$ aborts.

The proof of Lemma~\ref{Lemma:4 vs 5} also implies:

\begin{lemma}\label{Lemma: Bad4 is same as Bad5}
$\Pr[\BadComplete_4]=\Pr[\BadComplete_5]$, and $\Pr[\BadEval_4]=\Pr[\BadEval_5]$.
\end{lemma}

\paragraph{Remark.} Recall that in Section~\ref{sec:overview}, we mentioned the following statement: Assuming $G_5$ is efficient, the probability that $G_4$ or $G_5$ aborts is negligible. Due to Lemma~\ref{Lemma: Bad4 is same as Bad5}, the statement can be reduced to: assuming $G_5$ is efficient, the probability that $G_5$ aborts is negligible.

\begin{lemma}\label{Lemma:S is subset of M}
In $G_5$, $\cS^3.\CF$ is a subset of $\cM^3.\CF$ unless the game aborts.
\end{lemma}
\begin{proof}
Clear from the definition.
\end{proof}

Because of Lemma~\ref{Lemma:S is subset of M}, we can treat unsubverted (subverted) chains in $\cS^3.\CF$ as unsubverted (subverted) chains in $\cM^3.\CF$.

\subsubsection{The Gap Between the Fifth and the Sixth Game}

\begin{lemma}[The gap between $G_5$ and $G_6$]
The probability of $\cD$ outputting 1 in $G_5$ differs that in $G_6$ by at most $\Pr[\BadComplete_5]+\Pr[\BadEval_5]$.
\end{lemma}
\begin{proof}
The proof of the lemma is clear from the definition of $G_5$.
\end{proof}

\subsection{Preparations for Security and Efficiency Proof}\label{Subsec: preparation}

In this section, we will make some technical preparations for the proof of the following two statements:
\begin{itemize}
\item Security: Assuming $G_5$ is efficient, the probability that $G_5$ aborts is negligible.
\item Efficiency: $G_5$ is efficient.
\end{itemize}

We will focus on $G_5$ and its tables (especially $\cM^3.\CF$) in this section since both statements above are about $G_5$. The key property of $G_5$ we use to understand its distribution is that all the $\CF$ values in $\cM^3.\CF$ are selected uniformly and independently. Formally speaking, to understand the property of $G_5$ and its table $\cM^3.\CF$, we will consider the following simple probability model:
\begin{itemize}
\item Consider uniformly selecting a table $\CF_{\text{Full}}$ that contains $\CF$ values for all $(i,x)$.
\item In $G_5$, when the simulator proceeds to assign a new value to the table $\cM^3.CF$, the simulator takes the value from $\CF_{\text{Full}}$ instead of selecting it uniformly as usual.
\end{itemize}

It is easy to see that the above model does not change the distribution of $G_5$ at all. Under this model, $\cM.\CF$ is an uniform table in the sense that, at any moment of $G_5$, conditioned on the current exposed (evaluated) terms in $\cM.\CF$, any unexposed term, if it is ever evaluated, will be evaluated uniformly.

We will work on two main results in this section. First, we introduce
the notion of monotone increasing (and decreasing) chains, prove that
each unsubverted chain can be viewed as a union of a decreasing chain
and an increasing chain, and then show several nice properties of
increasing (or decreasing) chains. Second, we prove all the dishonest
terms on a subverted chain are located on an interval shorter than
$n/6$.

\paragraph{The order function; monotone chains.}
To record the order in which $\cM^3$ sets $\CF$ values,  we define the following order function $O_{\cM^3}$ from $\{1,\ldots,8n\} \times \{0,1\}^n$ to positive integers (with an additional symbol $\bot$):
\[
  O_{\cM^3}(i,x) = \begin{cases} t & \text{if $\CF_i(x)$ is the $t$-th evaluated $\CF$ value by $\cM^3$},\\
    \perp & \text{if $\CF_i(x)$ is undefined in $\cM^3.\CF$.}
  \end{cases}
\]
An unsubverted chain $(s,x_s,\ldots,x_{s+r})$ in $\cM^3.\CF$ is said to be \emph{monotone
  increasing} (or \emph{monotone decreasing}) if
$O_{\cM^3}(i,x_i)<O_{\cM^3}(i+1,x_{i+1})$ for all $s \leq i < s+r$ (or,
likewise, $O_{\cM^3}(j,x_j)>O_{\cM^3}(j+1,x_{j+1})$ for all
$s \leq j < s+r$).

In the rest of the paper, without loss of generality, we focus our
analytic efforts on increasing chains; the results related to
increasing chains can be easily transitioned into those related to
decreasing chains. We first show that any unsubverted chain can be viewed as a union of a deceasing chain and an increasing chain.

\begin{lemma}\label{Lemma:add one point at a time}
If $G_5$ is efficient, then with overwhelming probability, any unsubverted chain $c=(s,x_s,\ldots,x_{s+r})$ in $\cM^3.\CF$ will satisfy one of the three conditions below:
\begin{enumerate}
\item $c$ is increasing,
\item $c$ is decreasing,
\item There exists an index $s < v < s+r$ such that $(s,x_s,\ldots,x_v)$ is decreasing and $(v,x_v,\ldots,x_{s+r})$ is increasing. 
\end{enumerate}
\end{lemma}
\begin{proof}
It suffices to show that in $\cM^3.\CF$ there is no unsubverted length three chain $(s,x_s,x_{s+1},x_{s+2})$  such that $\CF_{s+1}(x_{s+1})$ is evaluated after both $\CF_s(s_s)$ and $\CF_{s+2}(x_{s+2})$ are evaluated. Suppose that throughout $G_5$, there are no more than $P$ ($=\poly(n)$) elements in $\cM^3.\CF$. Then in $G_5$,
\begin{align*}\allowdisplaybreaks
 & \mkern20mu  \Pr \left [\parbox{10cm}{There is a length 3 chain $(s,x_s,x_{s+1},x_{s+2})$ such that $O_{\cM^3}(s+1,x_{s+1})>\max\{O_{\cM^3}(s+2,x_{s+2}),O_{\cM^3}(s,x_{s}).\}$}\right ]\\
 & = \sum_{i=2}^{P} \Pr \left [\parbox{10cm}{There is a length 3 chain $(s,x_s,x_{s+1},x_{s+2})$ such that $O_{\cM^3}(s+1,x_{s+1})=i>\max\{O_{\cM^3}(s+2,x_{s+2}),O_{\cM^3}(s,x_{s})\}$.}\right ]\\
 & = \sum_{i=2}^{P} \sum_{\substack{j,k<i \\ j \neq k}} \Pr \left [\parbox{10cm}{There is a length 3 chain $(s,x_s,x_{s+1},x_{s+2})$ such that $O_{\cM^3}(s+1,x_{s+1})=i$, $O_{\cM^3}(s+2,x_{s+2})=j$ and $O_{\cM^3}(s,x_{s})=k$.}\right ]\\
 & < \sum_{i=2}^{P} \sum_{\substack{j,k<i \\ j \neq k}} \frac{1}{2^n}\\ 
 & < \frac{P^3}{2^n}\\ 
 & = \negl(n),
\end{align*}
where the first inequality is based on the fact that and $\CF_{s+1}(x_{s+1})$ is selected uniformly and is independent of $\CF_s(s_s)$ and $\CF_{s+2}(x_{s+2})$.
\end{proof}

\paragraph{Advantages of $\cM^3$ over $\cS^3$.} Note that Lemma
\ref{Lemma:add one point at a time} does not work for the simulator
$\cS^3$ because not all terms in $\cS^3.CF$ are selected
independently. Classifying the unsubverted chains in $\cM^3.\CF$ as
increasing and decreasing is extremely useful to analyze the property
of the chains in $G_5$.

Next, we will use a sequence of lemmas to establish the following major theorem that describes the nice properties of increasing chains.

\begin{theorem}\label{Th: nice properties of increasing chain}
If $G_5$ is efficient, then with overwhelming probability, any unsubverted increasing chain $c=(s,x_s,\ldots,x_{s+r})$ ($r>8$) in $\cM^3.\CF$ will satisfy:
\begin{enumerate}
\item for any $0 < i<j$ and $8<j \leq r$, $({s+j},x_{s+j}) \notin Q_{s+i}(x_{s+i})$;
\item for any $7 \leq i<j \leq r$, $({s+i},x_{s+i}) \notin Q_{s+j}(x_{s+j})$;
\item for any $7 < i \leq r$, $({s+i},x_{s+i})$ is honest if $\CFt_{s+i}(x_{s+i})$ is defined.
\end{enumerate}
\end{theorem}

\begin{lemma}\label{Lemma: No subverted chain covers a long honest chain}
In $G_5$, with overwhelming probability, there are not an unsubverted (or subverted) chain $c=(i,x_i,\ldots,x_j)$ and a length 10 unsubverted chain $c'=(s,y_s,\ldots,y_{s+9})$ in $\cM^3.\CF$ such that
\begin{itemize}
\item for all $(j,x) \in c$, $\CFt_i(x)$ is defined;
\item $c$ and $c'$ are disjoint;
\item for each $s \leq k \leq s+9$, $(k,y_k) \in Q_c$.
\end{itemize}
\end{lemma}
\begin{proof}
  Consider proving the following stronger statement: Imagine we fill
  the entire table $\cM^3.\CF$ by uniformly selecting all the $F$
  values and $(a_i,b_i)$ ($i=1,\ldots,8n$). We will prove that with
  overwhelming probability, there are not two chains $c$ and $c'$ that
  satisfy the properties in the lemma.

  Let $(x_{i+1},x_{i+2})$, $(y_s,y_{s+1})$ be two pairs of $n$-bit
  strings and $(i,j,s)$ be three positive indices. We denote by $c$
  the length $(j-i)$ chain starting with $(i+1,x_{i+1},x_{i+2})$
  (without loss of generality, we assume $c$ is a subverted chain for
  convenience) and denote by $c'$ the length 10 unsubverted chain
  starting with $(s,y_s,y_{s+1})$. We denote by $x_v$
  ($v=i+1,\ldots,j$) the elements of $c$ and denote by $y_k$
  ($k=s,\ldots,s+9$) the elements of $c'$. We define the event:
  \[
    E_{i,j,s}(x_{i+1},x_{i+2},y_s,y_{s+1}):=\left\{\parbox{85mm}{$c$
        and $c'$ are disjoint, and for each $s \leq k \leq
        s+9$, $(k,y_k) \in Q_c$} \right\}\,.
  \] 
  For $s \leq t \leq s+9$, we also define:
  \[
    E_{i,j,s}^{t}(x_{i+1},x_{i+2},y_s,y_{s+1}):=\left\{\text{$c$
        and $c'$ are disjoint, and for each $s \leq k \leq
        t$, $(k,y_k) \in Q_c$} \right\}\,.
  \] 
  To analyze the probability of $E_{i,j,s}(x_1,x_2,y_s,y_{s+1})$ over
  the choice of $F$ and $(a_i,b_i)$ ($i=1,\ldots,8n$), we consider
  selecting uniformly the values of $F_i(x)$ for all $i=1,\ldots,8n$
  and $x \in \{0,1\}^n$ and selecting uniformly
  $a_v \cdot x_v \oplus b_v$ for $v=i,\ldots,j$. Since the function
  $x_v \rightarrow a_i \cdot x_v \oplus b_i$ is pairwise independent,
  the values of $a_k \cdot y_k \oplus b_k$($k=s,\ldots,s+9$) are
  uniformly random. (For convenience, in the following, we will write
  $E_{i,j,s}$ for $E_{i,j,s}(x_1,x_2,y_s,y_{s+1})$ and $E_{i,j,s}^t$
  for $E_{i,j,s}^t(x_1,x_2,y_s,y_{s+1})$.) Over the randomness of
  $a_k \cdot y_k \oplus b_k$ ($k=s,\ldots,s+9$), we have
\begin{align*}\allowdisplaybreaks
 & \mkern20mu  \Pr[\, E_{i,j,s}]\\
 & = \Pr[E_{i,j,s} \mid E_{i,j,s}^{s+1} \,] \cdot \Pr[E_{i,j,s}^{s+1}]\\
 & < \Pr[\, E_{i,j,s} \mid E_{i,j,s}^{s+1} \,]\\
 & \mkern20mu \cdot(\Pr[\,\CF_s(y_s) \in \cup_{v=i}^jQ_v(x_v) \mid y_s \neq x_s \,]+\Pr[\,\CF_{s+1}(y_{s+1}) \in \cup_{v=i}^jQ_v(x_v) \mid y_{s+1} \neq x_{s+1} \,])\\
 & < \Pr[E_{i,j,s} \mid E_{i,j,s}^{s+3} \,] \cdot \Pr[E_{i,j,s}^{s+3} \mid E_{i,j,s}^{s+1} \,] \cdot 2 \cdot (8n \cdot q_{\cA}/2^n)\\ 
 & < \Pr[E_{i,j,s} \mid E_{i,j,s}^{s+5} \,] \cdot \Pr[E_{i,j,s}^{s+5} \mid E_{i,j,s}^{s+3} \,] \cdot (16n \cdot q_{\cA}/2^n)^2\\
 & < \Pr[E_{i,j,s} \mid E_{i,j,s}^{s+7} \,] \cdot \Pr[E_{i,j,s}^{s+7} \mid E_{i,j,s}^{s+5} \,] \cdot (16n \cdot q_{\cA}/2^n)^3\\
 & < \Pr[E_{i,j,s} \mid E_{i,j,s}^{s+7} \,] \cdot (16n \cdot q_{\cA}/2^n)^4\\
 & < (16n \cdot q_{\cA}/2^n)^5\,.
\end{align*}
The lemma is implied by taking the union bound over the choice of $(x_1,x_2,y_s,y_{s+1})$.
\end{proof}

A similar proof can be used to prove the following lemma:
\begin{lemma}\label{Lemma: No point covers length 8 unsubverted chain}
  With overwhelming probability over the choice of all the $F$ values
  and $(a_i,b_i)$ ($i=1,\ldots,8n$), there are not a term $(i,x_i)$ and
  a length 8 unsubverted chain $c=(s,y_s,\ldots,y_{s+7})$ in
  $\cM^3.\CF$ such that $(k,y_k) \in Q_i(x_i)$ for all $k=s,s+2,s+4,s+6$.
\end{lemma}

\begin{lemma}\label{Lemma: querying odd point implies querying all the even points before}
  If $G_5$ is efficient, then with overwhelming probability, for any
  unsubverted increasing chain $c=(s,x_s,\ldots,x_{s+r})$ in
  $\cM^3.\CF$, if
  $(s+2t+1,x_{s+2t+1}) \in Q_{s+2k}(x_{s+2k})$ (assuming
  $\CFt_{s+2k}(x_{s+2k})$ is defined) for some $t,k$ with
  $0 < 2t+1,2k \leq r$, then $(s+2i,x_{s+2i}) \in Q_{s+2k}(x_{s+2k})$
  for all $0 < i \leq t$.
\end{lemma}
\begin{proof}
We give a simple example to show the idea of the proof. Take $s=1$, $r=7$, $k=2$ and $t=3$ for example. We want to show that for any chain $c=(1,x_1,\ldots,x_8)$, if $(8,x_8) \in Q_5(x_5)$, then with overwhelming probability, $(1+2i,x_{1+2i}) \in Q_5(x_5)$ for $i=1$.

Consider the following two ways of determining a length 8 unsubverted chain:
\begin{itemize}
\item  \textbf{Procedure 1}:
\begin{enumerate}
\item Pick an arbitrary moment in $G_5$ and abort the game. Denote the table $\cM^3.\CF$ at this moment by $T_{\text{initial}}$. Pick a length 2 increasing chain $(1,x_1,x_2)$ in $T_{\text{initial}}$ such that it is not a subchain of a length 3 unsubverted chain.
\item For $2 \leq i \leq 7$, select $\CF_i(x_i)$ uniformly, set $x_{i+1}:=\CF_i(x_i) \oplus x_{i-1}$ and abort the procedure if $(i+1,x_{i+1})$ is already in the table $T_{\text{initial}}$.
\item Evaluate $\CFt_5(x_5)$. 
\end{enumerate}
\item  \textbf{Procedure 2}:
\begin{enumerate}
\item Pick an arbitrary moment in $G_5$ and abort the game. Denote the table $\cM^3.\CF$ at this moment by $T_{\text{initial}}$. Pick a length 2 increasing chain $(1,x_1,x_2)$ in $T_{\text{initial}}$ such that it is not a subchain of a length 3 unsubverted chain.
\item Select $\CF_2(x_2)$ uniformly and set $x_3:=a_2 \oplus x_1$.
\item Select 4 uniform $n$-bit strings $a_4$, $a_5$, $a_6$ and $a_7$. Set  $x_5:=a_4 \oplus x_3$, $x_7:=a_6 \oplus x_5$ and abort the procedure if either of them is in $T_{\text{initial}}$. Set $\CF_5(x_5):=a_5$ and $\CF_7(x_7):=a_7$.
\item Evaluate $\CFt_5(x_5)$.
\item  Select $\CF_3(x_3)$ uniformly (use the existing value if it has been evaluated), set $x_4:=\CF_3(x_3) \oplus x_2$, $x_6:=a_5 \oplus x_4$, $x_8:=a_7 \oplus x_6$, and abort the procedure if any one of $x_4$, $x_4$ and $x_8$ is in $T_{\text{initial}}$.
\end{enumerate} 
\end{itemize}

A quick thought reveals that the above two procedures are equivalent in terms of the distribution of the chain and, furthermore, the probability they abort is negligible because of Lemma \ref{Lemma:add one point at a time}. We use the second procedure to analyze the distribution of the first one. In the second procedure, we can see that if $(3,x_3) \notin Q_5(x_5)$, then $\CF_3(x_3)$ is still uniform conditioned on $Q_5(x_5)$, which implies that $x_8=a_7 \oplus x_6=a_7 \oplus a_5 \oplus x_4=a_7 \oplus a_5 \oplus  \CF_3(x_3) \oplus x_2$ is uniform. Therefore, if $(3,x_3) \notin Q_5(x_5)$, $(8,x_8) \in Q_5(x_5)$ with negligible probability.

The full proof can be achieved by replacing the concrete numbers in the last example by more general parameters $s$, $r$, $k$ and $t$ and taking the union bound over the various values of these parameters.
\end{proof}

\begin{lemma}\label{Lemma:early does not query late}
If $G_5$ is efficient, then with overwhelming probability, for any unsubverted increasing chain $c=(s,x_s,\ldots,x_{s+r})$ in $\cM^3.\CF$ and any index $i,j$ with $0 < i<j$ and $8<j \leq r$, $(s+j,x_{s+j}) \notin Q_{s+i}(x_{s+i})$ (if $\CFt_{s+i}(x_{s+i})$ is defined).
\end{lemma}

\begin{proof}
Without loss of generality, assume $i=0$. Suppose $({s+j},x_{s+j}) \in Q_{s}(x_{s})$. Notice that $(s+j-1,x_{s+j-1}) \in Q_{s+i}(x_{s+i})$ with overwhelming probability because otherwise the randomness of $\CF_{s+j-1}(x_{s+j-1})$ will cause the event $({s+j},x_{s+j}) \notin Q_{s}(x_{s})$. Then,
\begin{itemize}
\item if $j$ is odd, since $j>8$ and $({s+j},x_{s+j}) \in Q_{s}(x_{s})$, by Lemma~\ref{Lemma: querying odd point implies querying all the even points before}, $(s+2k,x_{s+2k}) \in Q_{s}(x_{s})$ for $k=1,2,3,4$. This contradicts Lemma~\ref{Lemma: No point covers length 8 unsubverted chain}.
\item if $j$ is even, since $j>8$ and $({s+j-1},x_{s+j-1}) \in Q_{s}(x_{s})$, by Lemma~\ref{Lemma: querying odd point implies querying all the even points before}, $(s+2j,x_{s+2j}) \in Q_{s}(x_{s})$ for $j=1,2,3,4$, which contradicts with Lemma~\ref{Lemma: No point covers length 8 unsubverted chain}. \qedhere
\end{itemize}
\end{proof}

\begin{lemma}\label{Lemma: late points querying odd point implies querying all the early even points}
If $G_5$ is efficient, then with overwhelming probability, for any unsubverted increasing chain $c=(s,x_s,\ldots,x_{s+r})$ in $\cM^3.\CF$, if $(s+2t,x_{s+2t}) \in Q_{s+k}(x_{s+k})$ (assuming $\CFt_{s+k}(x_{s+k})$ is defined) for some $t,k$ with $0 < 2t < k \leq r$, then $(s+2i-1,x_{s+2i-1}) \in Q_{s+k}(x_{s+k})$ for all $0<i \leq t$.
\end{lemma}
\begin{proof}
The proof of the lemma is similar to that of Lemma \ref{Lemma: querying odd point implies querying all the even points before}. Consider the example where $s=1$, $r=8$, $t=2$ and $k=8$. We want to show that for any chain $c=(1,x_1,\ldots,x_9)$, if $(5,x_5) \in Q_9(x_9)$, then with overwhelming probability, $(2i,x_{2i}) \in Q_9(x_9)$ for $i=1$.

Consider the following two ways of determining a length 9 unsubverted chain:
\begin{itemize}
\item  \textbf{Procedure 1}:
\begin{enumerate}
\item Pick an arbitrary moment in $G_5$ and abort the game. Denote the table $\cM^3.\CF$ at this moment by $T_{\text{initial}}$. Pick a length 2 increasing chain $(1,x_1,x_2)$ in $T_{\text{initial}}$ such that it is not a subchain of a length 3 unsubverted chain.
\item For $2 \leq i \leq 8$, select $\CF_i(x_i)$ uniformly, set $x_{i+1}:=\CF_i(x_i) \oplus x_{i-1}$ and abort the procedure if $(i+1,x_{i+1})$ is already in the table $T_{\text{initial}}$.
\item Evaluate $\CFt_9(x_9)$. 
\end{enumerate}
\item  \textbf{Procedure 2}:
\begin{enumerate}
\item Pick an arbitrary moment in $G_5$ and abort the game. Denote the table $\cM^3.\CF$ at this moment by $T_{\text{initial}}$. Pick a length 2 increasing chain $(1,x_1,x_2)$ in $T_{\text{initial}}$ such that it is not a subchain of a length 3 unsubverted chain.
\item Select 3 uniform $n$-bit strings $a_3$, $a_4$ and $a_5$. Set $x_4:=a_3 \oplus x_2$, $x_6:=a_5 \oplus x_4$ and aborts the procedure if either of them is in $T_{\text{initial}}$. Set $\CF_4(x_4):=a_4$ and $\CF_6(x_6):=a_6$.
\item Select $x_7$, $x_8$ and $x_9$ uniformly and aborts the procedure if any one of them is in $T_{\text{initial}}$. Set $\CF_7(x_7):=x_6 \oplus x_8$ and $\CF_8(x_8):=x_7 \oplus x_9$.
\item Evaluate $\CFt_9(x_9)$.
\item  Select $\CF_2(x_2)$ uniformly(use the existing value if it has been evaluated), set $x_3:=\CF_2(x_2) \oplus x_1$, $x_5:=a_4 \oplus x_3$, $\CF_6(x_6):=x_7 \oplus x_5$, and aborts the procedure if either $x_3$ or $x_5$ is in $T_{\text{initial}}$.
\end{enumerate} 
\end{itemize}

A quick thought reveals that the above two procedures are equivalent
in terms of the distribution of the chain (and, furthermore, the
probability they abort is negligible because of Lemma~\ref{Lemma:add
  one point at a time}). We use the second procedure to analyze the
distribution of the first one. In the second procedure, we can see
that if $(2,x_2) \notin Q_9(x_9)$, then $\CF_2(x_2)$ is still uniform
conditioned on $Q_9(x_9)$, which implies that
$x_5=a_4 \oplus x_3=a_4 \oplus a_2 \oplus x_1$ is uniform. Therefore,
if $(2,x_2) \notin Q_9(x_9)$, $(5,x_5) \in Q_9(x_9)$ with negligible
probability.

The formal proof can be achieved by replacing the concrete numbers in the last example by more general parameters $s$, $r$, $t$ and $k$ and taking the union bound over the various values of these parameters.
\end{proof}

\begin{lemma}\label{Lemma:late does not query early}
If $G_5$ is efficient, then with overwhelming probability, for any unsubverted increasing chain $c=(s,x_s,\ldots,x_{s+r})$ in $\cM^3.\CF$ and any index $i,j$ with $7<i<j \leq r$, $(s+i,x_{s+i}) \notin Q_{s+j}(x_{s+j})$(if $\CFt_{s+j}(x_{s+j})$ is defined).
\end{lemma}
\begin{proof}
The lemma is derived directly from Lemma \ref{Lemma: No point covers length 8 unsubverted chain} and Lemma \ref{Lemma: late points querying odd point implies querying all the early even points}.
\end{proof}

\begin{lemma}\label{Lemma:once honest, honest forever}
If $G_5$ is efficient, then with overwhelming probability, there does not exist an unsubverted increasing chain $c=(i,x_i,\ldots,x_{i+8})$ in $\cM^3.\CF$ such that $\CFt_{i+8}(x_{i+8})$ is defined in $\cM^3.\CF$ and $(i+8,x_{i+8})$ is dishonest.
\end{lemma}
\begin{proof}
We say the distinguisher $\cD$ wins $G_5$ if it is able to find an unsubverted increasing chain $c$ in $\cM^3.\CF$ that satisfies the property in the lemma. By Lemma \ref{Lemma:early does not query late}, the probability that there is a length-9 unsubverted increasing chain $c=(i,x_i,\ldots,x_{i+8})$ with $(i+7,x_{i+7}) \in Q_{i+8}(x_{i+8})$ ($\CFt_{i+8}(x_{i+8})$) is negligible. We denote this negligible probability by $\delta$. 

To show the probability that $\cD$ wins is negligible, consider the following experiment with a distinguisher $\cD^{*}$:
\begin{mdframed}
\begin{center}
  \textsc{Exp*}
\end{center}
\begin{enumerate}
\item $\cD^{*}$ takes an arbitrary moment of $G_5$, stops the game and selects an arbitrary length-2 increasing chain $(i,x_i,x_{i+1})$ in $\cM^3.\CF$ such that $\CF_{i+2}(x_{i+2})$ is not evaluated for $x_{i+1}:=x_i \oplus \CF_i(x_i)$.
\item Then, $\cD^{*}$ extends $(i,x_i,x_{i+1})$ to $(i,x_i,\ldots,x_{i+7})$ by iteratively evaluating $\CF_{j-1}(x_{j-1})$ (selected uniformly) and $x_j:=x_{j-2} \oplus \CF_{j-1}(x_{j-1})$ for $i+3 \leq j \leq i+7$. The experiment aborts if $\CF_{i+7}(x_{i+7})$ is already evaluated.
\item For any term $(j,y)$, if $\CF_j(y)$ is still unevaluated and $(j,y) \neq (i+7,x_{i+7})$, $\cD^{*}$ selects $\CF_j(y)$ uniformly. 
\item Finally $\cD^{*}$ selects $\CF_{i+7}(x_{i+7})$ and check if $(i+8,x_{i+8})$ is dishonest for $x_{i+8}:=x_{i+6} \oplus \CF_{i+7}(x_{i+7})$. 
\item $\cD^{*}$ wins \textbf{Exp*} if the experiment does not abort in Step 2 and $(i+8,x_{i+8})$ is dishonest.
\end{enumerate}
\end{mdframed}

To prove $\cD$ wins $G_5$ negligibly, it is sufficient to show the probability that the experiment aborts in Step 2 or $\cD^{*}$ wins is negligible. We also stress that although \textbf{Exp*} is not $G_5$, the lemmas we proved in this section can still be applied because all the $\CF$ values here are also selected uniformly and independently.

\vspace{-4mm}
\begin{align*}
 & \mkern20mu  \Pr_{\textbf{Exp*}}[\text{The experiment aborts in Step 2 or $\cD^{*}$ wins.}]\\
 & \leq \Pr_{\textbf{Exp*}}[\text{The experiment aborts in Step 2.}]+\Pr_{\textbf{Exp*}}\left[\parbox{4cm}{$\cD^{*}$ wins and there are at least $\sqrt{\delta}2^n$ $n$-bit strings $x$ such that $({i+7},x_{i+7}) \in Q_{i+8}(x)$.}\right] \\
 & \mkern230mu +\Pr_{\textbf{Exp*}}\left[\parbox{6cm}{$\cD^{*}$ wins and there are fewer than $\sqrt{\delta}2^n$ $n$-bit strings $x$ such that $({i+7},x_{i+7}) \in Q_{i+8}(x)$.}\right]\\
 & < \negl(n)+\Pr_{\textbf{Exp*}} \left [\parbox{9cm}{$\cD^{*}$ wins and there are at least $\sqrt{\delta}2^n$ $n$-bit strings $x$ such that $({i+7},x_{i+7}) \in Q_{i+8}(x)$.} \right ]\\
 & \mkern100mu +\Pr_{\textbf{Exp*}}\left[\text{$\cD^{*}$ wins.} \;\middle|\; \parbox{7cm}{There are fewer than $\sqrt{\delta}2^n$ $n$-bit strings $x$ such that $({i+7},x_{i+7}) \in Q_{i+8}(x)$.}\right]\\
 & < \negl(n)+\sqrt{\delta}+\Pr_{\textbf{Exp*}}\left[\parbox{3cm}{$(i+8,x_{i+8})$ is dishonest and $({i+7},x_{i+7}) \in Q_{i+8}(x_{i+8})$.} \;\middle|\; \parbox{5cm}{There are fewer than $\sqrt{\delta}2^n$ $n$-bit strings $x$ such that $({i+7},x_{i+7}) \in Q_{i+8}(x)$.}\right]\\
 & \mkern150mu +\Pr_{\textbf{Exp*}}\left[\parbox{3cm}{$(i+8,x_{i+8})$ is dishonest and $({i+7},x_{i+7}) \notin Q_{i+8}(x_{i+8})$.} \;\middle|\; \parbox{4cm}{There are fewer than $\sqrt{\delta}2^n$ $n$-bit strings $x$ such that $({i+7},x_{i+7}) \in Q_{i+8}(x)$.}\right]\\
 & < \negl(n)+\sqrt{\delta}+\sqrt{\delta}+\epsilon = \negl(n).  \qed
\end{align*} 

\end{proof}

Theorem \ref{Th: nice properties of increasing chain}] follows from the combination of Lemma \ref{Lemma:early does not query late}, Lemma \ref{Lemma:late does not query early} and Lemma \ref{Lemma:once honest, honest forever}.

\begin{definition}[Quasi-honest]
In $G_5$, for any $i=9,\ldots,8n$ and $x \in \{0,1\}^n$, we say $(i,x)$ is \emph{quasi-honest} if there is an increasing chain $(s,x_s,\ldots,x_{s+r})$ ($r \geq 8$) in $\cM^3.\CF$ such that $(s+r,x_{s+r})=(i,x)$.
\end{definition}

Now we turn our attention to the dishonest terms on a subverted chain. We want to show that although, in general, there are some dishonest terms on a subverted chain, all of them gather in a small area.

\begin{definition}[Bad region]
For a subverted chain $c=(s,x_s,\ldots,x_{s+r})$ in $\cM^3.\CF$, we say a subchain $(i,x_i,\ldots,x_j)$ ($s \leq i<j \leq s+r$) of $c$ is a \emph{bad region} of $c$ if there is no sequence of 14 consecutive elements $(k,x_k,\ldots,x_{k+13})$ ($i \leq k \leq j-13$) that are honest.

For a subverted chain $c=(s,x_s,\ldots,x_{s+r})$ in $\cM^3.\CF$, we say two bad regions of $c$, $(i,x_i,\ldots,x_j)$ and $(i',x_{i'},\ldots,x_{j'})$ ($i<i',j<j'$) are separated if the subchain $(i,x_i,\ldots,x_j')$ of $c$ is not a bad region of $c$.
\end{definition}

\begin{lemma}\label{Lemma:bad region is short}
In $G_5$, with overwhelming probability, there does not exist a subverted chain $(s,x_s,\ldots,x_{s+r})$ in $\cM^3.\CF$ such that it has a bad region with length greater than $n/6$.
\end{lemma}
\begin{proof}
Consider proving the following stronger statement: with overwhelming probability over the uniform choice of $(a_i,b_i)$ ($i=1,\ldots,8n$) and values of $F_i(x)$ for all $i=1,\ldots,8n$ and $x \in \{0,1\}^n$, there is no bad region with length greater than $n/6$. Imagine we select $F_i(x)$ for all $i=\{1,\ldots,8n\}$ and $x \in \{0,1\}^n$ and leave $a_i$ and $b_i$ undetermined. Then, over the randomness of the choice of $a_i$ and $b_i$, we have
\begin{align*}
  & \mkern20mu \Pr[\text{There is a subverted chain $c$ with a bad region longer than $n/6$.}] \\
  & = \sum_{i=1}^{8n}\Pr \left [\parbox{10cm}{There is a subverted chain $c$ with a bad region longer than $n/6$ and the bad region begins at index $i$.} \right ] \\
  & = \sum_{i=1}^{8n}\sum_{x,x' \in \{0,1\}^n}\Pr\left[\parbox{9cm}{There is a subverted chain $c$ with a bad region longer than $n/6$. The bad region begins at index $i$ and its first two elements are $(i,x)$ and $(i+1,x')$.}\right]\\
  & < \sum_{i=1}^{8n}\sum_{x,x' \in \{0,1\}^n}\Pr\left[\parbox{9cm}{There is a subverted chain $c=(i,x,x',\ldots,x_r)$ such that its first element has index $i$ and $r-i >n/6$. Moreover, for any length 14 subchain of $c$ in the form of $(14k,x_{14k},\ldots,x_{14k+13})$, at least one of 14 elements is dishonest.}\right]\\
  & < \sum_{i=1}^{8n}(2^n)^2 \cdot (14\epsilon)^{n/84-1} = 8n \cdot 2^{2n} \cdot (14\epsilon)^{n/84-1} = \negl(n). 
\end{align*}\qedhere
\end{proof}

\begin{lemma}\label{Lemma:there is no separated bad region}
If $G_5$ is efficient, then with overwhelming probability, there is no subverted chain $c=(s,x_s,\ldots,x_{s+r})$ in $\cM^3.\CF$ that has two separated bad regions. 
\end{lemma}
\begin{proof}
The lemma is implied directly by Lemma \ref{Lemma:add one point at a time} and Lemma \ref{Lemma:once honest, honest forever}.
\end{proof}

\subsection{Bounding Bad Events}\label{section:security}

Now we proceed to show the security: assuming $G_5$ is efficient, the probability that $G_5$ aborts is negligible.

There are two bad events that can cause $G_5$ to abort: $\BadComplete_5$ and $\BadEval_5$. For some technical reasons, we divide $\BadComplete_5$ into three smaller bad events:
\begin{itemize}
\item $\cM.\BadComplete_5$: $\cM.\BadComplete_5$ happens when $G_5$ aborts during the execution of $\cM^3.\text{Complete}$.
\item $\cS.\NewBadComplete_5$: $\cS.\NewBadComplete_5$ happens when $G_5$ aborts during the execution of $\cS^3.\text{Complete}$ on a chain $(i,x_i,x_{i+1},u)$ and
  \[
    (i,x_i,x_{i+1}) \notin \cM^3.\text{CompletedChains}\,,
  \]
  which means the chain $(i,x_i,x_{i+1})$ has not been completed by
  $\cM^3$.
\item $\cS.\ExistingBadComplete_5$: $\cS.\ExistingBadComplete_5$ happens when $G_5$ aborts during the execution of $\cS^3.\text{Complete}$ on a chain $(i,x_i,x_{i+1},u)$ and
  \[
    (i,x_i,x_{i+1}) \in \cM^3.\text{CompletedChains}\,,
  \]
  which means $\cS^3$ aborts when it is completing an existing chain in $\cM^3.\CF$. By definition of $\cS^3.\text{Complete}$, this bad event can happen only if $u=7n$.
\end{itemize}

For any bad event $A$ (e.g., $\cM.\BadComplete_5$, $\BadEval_5$, etc.) and a positive integer $k \leq q_{\cD}$, we denote by $A[k]$ the event that bad event $A$ causes $G_5$ to abort before the end of the $k$-th round of the interaction between $\cD$ and the simulators. We also denote the table $\cM^3.\CF$ ($\cS^3.\CF$) at the end of the $k$-th interaction of $G_5$ by $\cM^3.\CF[k]$ ($\cS^3.\CF[k]$).

\begin{definition}[Completed chains]
For any subverted chain $c=(i,x_i,\ldots,x_j)$ in $\cM^3.\CF$, we say it is $\cM.\text{Completed}$ if $(i,x_i,x_{i+1}) \in \cM^3.\text{CompletedChains}$. We denote the set of full $\cM.\text{Completed}$ chains by $C_{\cM.\text{FComp}}$.

The concepts of $\cS.\text{Completed}$ chains, $C_{\cS.\text{FComp}}$ are defined similarly. 

For any set $E$ we mentioned above (e.g., $\cM.\text{CompletedChains}$, $\cS.\text{CompletedChains}$, $C_{\cM.\text{FComp}}$, etc.) and a positive integer $k \leq q_{\cD}$, we denote by $E[k]$ the set $E$ at the end of $k$-th round of the game.
\end{definition}

To understand the probability of the bad events, we introduce the following property of the status of the game:

\begin{definition}[Good Status]
For any $0 <k \leq q_{\cD}$, we say $G_5$ has a \emph{good} status at round $k$ (denoted by $\GoodStatus[k]$) if at the end of the $k$-th round of the game:
\begin{itemize}
\item The game does not abort with any bad events.
\item For any increasing chain $(i,x_i,x_{i+1},x_{i+2})$ in $\cM^3.\CF$, if $(i+1,x_{i+1}) \notin \cS^3.\CF$, then
\begin{enumerate}
\item $(i+2,x_{i+2}) \notin \cS^3.\CF$, and
\item for any $c \in C_{\cM.\text{FComp}}[k]$ such that $(i+2,x_{i+2}) \in Q_c$, $(i+1,x_{i+1}) \in Q_c$.
\end{enumerate}
\item For any decreasing chain $(x_i,x_{i+1},x_{i+2})$ in $\cM^3.\CF$, if $(i+1,x_{i+1}) \notin \cS^3.\CF$, then
\begin{enumerate}
\item $(i,x_{i}) \notin \cS^3.\CF$, and
\item for any $c \in C_{\cM.\text{FComp}}[k]$ such that $(i,x_{i}) \in Q_c$, $(i+1,x_{i+1}) \in Q_c$.
\end{enumerate}
\end{itemize}
\end{definition}

We will use an induction proof to show the following three theorems:

\begin{theorem}[$\BadComplete_5$ is negligible.]\label{Th: complete}
If $G_5$ is efficient and for some positive integer $T<q_{\cD}$,
\[
  \Pr[\text{$\GoodStatus[k]$ does not happen}]= \negl(n),
\]
then
\[
  \Pr[\text{$\BadComplete_5[k+1]$ happens}]= \negl(n).
\]
\end{theorem}

\begin{theorem}[$\BadEval_5$ is negligible.]\label{Th: eval}
If $G_5$ is efficient and for some positive integer $k<q_{\cD}$,
\[
  \Pr[\text{$\GoodStatus[k]$ does not happen}]= \negl(n),
\]
then
\[
  \Pr[\text{$\BadEval_5[k+1]$ happens}]= \negl(n).
\]
\end{theorem}

\begin{theorem}\label{Th: goodstatus}
If $G_5$ is efficient and for some positive integer $k<q_{\cD}$,
\[
  \Pr[\text{$\GoodStatus[k]$ does not happen}]= \negl(n),
\]
then
\[
  \Pr[\text{$\GoodStatus[k+1]$ does not happen}]= \negl(n).
\]
\end{theorem}

The structure of the proof is:
\\
\begin{figure}
  \centering
\begin{tikzpicture}[node distance=20pt,font=\scriptsize]
  \node[draw, rounded corners]                        (Goodk)   {$\GoodStatus[k]$};
  \node[draw, align=center, below=of Goodk]             (Mcom)  {$\cM.\BadComplete_5[k+1]$\\ is negligible};
  \node[draw, align=center, right=of Mcom]     (Snewcom)  {$\cS.\NewBadComplete_5[k+1]$\\ is negligible};
  \node[draw, rounded corners, align=center, left= of Mcom]                   (Eval)  {$\BadEval_5[k+1]$\\ is negligible};
  \node[draw, rounded corners, left= of Eval]                   (Goodkone)  {$\GoodStatus[k+1]$};
  \node[draw, align=center, below= of Eval]  (Sexisting)     {$\cS.\ExistingBadComplete_5[k+1]$\\ is negligible};
  \node[draw, rounded corners, align=center, right=of Sexisting]                    (Com)  {$\BadComplete_5[k+1]$\\ is negligible};

  \draw[->] (Goodk) -- node[right]{(2)} (Mcom);
  \draw[->] (Mcom) -- (Com);
  \draw[->] (Mcom) -- node[above]{(4)} (Eval);
  \draw[->] (Eval) -- node[right]{(3)} (Sexisting);
  \draw[->] (Sexisting) -- (Com);
  \draw[->] (Eval) -- node[above]{(5)} (Goodkone);
  \draw[->] (Goodk) -- (Goodk-|Snewcom)  node[right]{(1)}-> (Snewcom);  
  \draw[->] (Snewcom) -- (Snewcom|-Com) -> (Com);   
\end{tikzpicture}
\caption{A diagram showing the flow of the proof.}
\end{figure}
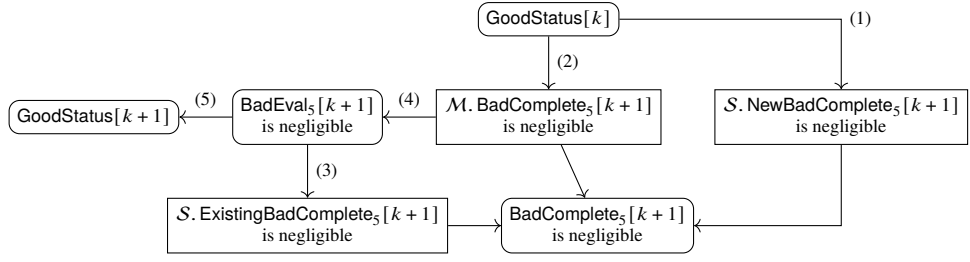
\\

Each arrow in the diagram above corresponds to a lemma we use to show the three main theorems. In section \ref{subsub: Complete}, we will prove Theorem \ref{Th: complete} by showing lemma (1), (2), (3) are true and assuming lemma (4) is true. In section \ref{subsub: Eval} we will prove Theorem \ref{Th: eval} by proving lemma (4). In section \ref{subsub: Goodstatus}, we will proving Theorem \ref{Th: goodstatus} by lemma (5).  

Putting the three theorems together, we have:
\begin{theorem}\label{Th:security}
If $G_5$ is efficient, the probability that it aborts because of bad events is negligible.
\end{theorem}
\begin{proof}
The theorem is derived from Theorem \ref{Th: complete}, Theorem \ref{Th: eval}, Theorem \ref{Th: goodstatus} and the fact that $\GoodStatus[k]$ is true for $k=1$.
\end{proof}

\subsubsection{$\BadComplete_5[k+1]$ is negligible}\label{subsub: Complete}
We first prove $\cS.\NewBadComplete_5[k+1]$ is negligible when $\GoodStatus[k]$ does not happen. Consider the following experiment \textbf{Exp-LongGeneration$[k]$} with a distinguisher $\cD_{\text{Long}}$:

\begin{mdframed}
\begin{center}
  \textsc{Exp-LongGeneration$[k]$}
\end{center}
For a positive integer $k <q_{\cD}$, if $G_5$ does not abort by the end of the $k$-th round of the game, $\cD_{\text{Long}}$ takes the table $\cM^3.\CF[k]$ and renames it by $T_1$. As the experiment goes, $\cD_{\text{Long}}$ will add new $\CF$ values to $T_1$. All the newly added $\CF$ values will be evaluated uniformly.
\begin{enumerate}
    \item Define $C_{\text{long}}$ to be the set of unsubverted chains in $T_1$ with length $n/10-22$. Define the set $S:=C_{\text{long}} \times \{1,2\}$. 
    \item $\cD_{\text{Long}}$ uniformly selects three distinct elements in $s_1,s_2,s_3 \in S$. Run the following For-loop:
    \begin{align*}
        & \mkern30mu \textbf{For } \text{$i=1,2,3$:}\\
        & \mkern30mu \textbf{If } \text{$s_i$ is in the form of $(c,1)$ for some $c=(r,x_r,\ldots,x_{r+n/10-23})$ } \textbf{then}\\
        & \mkern50mu \text{Select an integer $m$ uniformly from $(0,1,2,\ldots,22)$.}\\
        & \mkern50mu \text{Extend the chain $c$ to the right by $m$ terms and to the left by $22-m$ terms. }\\
        & \mkern50mu \text{Call the extended chain $c_i$. Evaluate $\CFt_j(x_j)$ for each $(j,x_j) \in c_i$.}\\
        & \mkern30mu \textbf{If } \text{$s_i$ is in the form of $(c,2)$ for some $c=(r,x_r,\ldots,x_{r+n/10-23})$ } \textbf{then}\\
        & \mkern50mu \text{Evaluate $\CFt_j(x_j)$ for $j=r,\ldots,r+n/10-23$.}\\
        & \mkern50mu \textbf{If } \text{$\CFt_i(x_i) \neq \CF_i(x_i)$ for some $r \leq i \leq r+n/10-23$ } \textbf{then}\\
        & \mkern70mu \text{The experiment aborts.}\\
        & \mkern50mu \textbf{Else}\\ 
        & \mkern70mu \text{Evaluate the full subverted chain that contains $c$ as a subchain and}\\
        & \mkern70mu \text{name the full chain by $c_i$.}
    \end{align*}
The new table after the For-loop is called $T_2$.
\end{enumerate}

$\cD_{\text{Long}}$ \emph{wins} if there is an increasing chain $c'=(i,x_i,x_{i+1},x_{i+2})$ in $T_2$ such that
    \begin{itemize}
      \item none of the elements of $c'$ is in $T_1$;
      \item $(i+2,x_{i+2}) \in Q_{c_3}$;
      \item $(i+1,x_{i+1}) \notin Q_{c_3}$. 
    \end{itemize}     

\end{mdframed}

\begin{lemma}\label{Lemma:query one point must query previous points}
If $G_5$ is efficient, 
\[
  \Pr[\text{$\cD_{\text{Long}}$ wins \textbf{Exp-LongGeneration$[k]$}}]= \negl(n).
\]
\end{lemma}
\begin{proof}
We first notice that, although \textbf{Exp-LongGeneration$[k]$} is not $G_5$, the lemmas (in Section~\ref{Subsec: preparation}) about the properties of monotone chains can still be applied since, in \textbf{Exp-LongGeneration$[k]$}, all the $\CF$ values are selected uniformly and independently.

Suppose \textbf{Exp-LongGeneration$[k]$} fails with non-negligible probability. Consider the following realization of \textbf{Exp-LongGeneration$[k]$}: with non-negligible probability in \textbf{Exp-LongGeneration$[k]$}, we can select $T_1$, $s_1$, $s_2$ and $s_3$ so that the probability in the lemma is non-negligible; conditioned on $T_1$, $s_1$, $s_2$ and $s_3$, with non-negligible probability over the choice of $T_2$, we can find an increasing chain $c'=(i,x_i,x_{i+1},x_{i+2})$ in $T_2$ that makes \textbf{Exp-LongGeneration$[k]$} fail. 

Consider the following realization of \textbf{Exp-LongGeneration$[k]$} which have the same non-negligible probability with the above one. First we select the same $T_1$. Then we select the same three elements of $S$ in the order $s_3$, $s_1$, $s_2$. Finally the same $T_2$ is selected. Now we want to find a contradiction by showing that $(i,x_i,x_{i+1})$ is increasing and $(i+1,x_{i+1},x_{i+2})$ is decreasing, which violates Lemma \ref{Lemma:add one point at a time}. Notice that in the last realization of \textbf{Exp-LongGeneration$[k]$}, $c'=(i,x_i,x_{i+1},x_{i+2})$ is increasing, $(i+2,x_{i+2}) \in Q_{c_3}$ and $(i+1,x_{i+1}) \notin Q_{c_3}$. It is clear that $\CF_{i+1}(x_{i+1})$ is evaluated after $\CF_{i+2}(x_{i+2})$. We want to show $\CF_{i+1}(x_{i+1})$ is also evaluated after $\CF_{i}(x_{i})$.
\begin{itemize}
\item Case 1: if $(i,x_{i}) \in Q_{c_3}$, $\CF_{i+1}(x_{i+1})$ is evaluated after $\CF_{i}(x_{i})$ since $(i+1,x_{i+1}) \notin Q_{c_3}$.
\item Case 2: if $(i,x_{i}) \notin Q_{c_3}$, $\CF_{i+1}(x_{i+1})$ is evaluated after $\CF_{i}(x_{i})$ since $c'$ is increasing.
\end{itemize}
\end{proof}

\begin{lemma}\label{Lemma: all k+1 completed chains have a long subchain in k}
If $G_5$ is efficient and does not abort by the end of $k$-th round of the game for some positive integer $k<q_{\cD}$, then with overwhelming probability, for any unsubverted chain $c=(s,x_s,\ldots,x_{s+n/10-1})$ that is processed by the procedure $\cS^3.\text{HonestyCheck}$ in the $(k+1)$-th round of $G_5$, there is a length $(n/10-22)$ subchain $c'=(i,x_i,\ldots,x_{i+n/10-23})$ of $c$ such that each element of $c'$ is in $\cM^3.\CF[k]$.
\end{lemma}
\begin{proof}
Suppose that all the chains that are processed by $\cS^3.\text{HonestyCheck}$ before $c=(s,x_s,\ldots,x_{s+n/10-1})$ have the properties in the lemma (i.e., has a length $(n/10-22)$ subchain that is included in $\cM^3.\CF[k]$). We want to show that either $c$ will not be processed by $\cS^3.\text{HonestyCheck}$ or $c$ has a length $(n/10-22)$ subchain in $\cM^3.\CF[k]$.

If $c$ does not have a length $(n/10-22)$ subchain in $\cM^3.\CF[k]$, then, without loss of generality, we assume the chain $(s+n/10-11,x_{s+n/10-11},\ldots,x_{s+n/10-1})$ is increasing and none of the elements in the chain is in $\cM^3.\CF[k]$.

Denote by $t$ the moment when $\cS^3$ is about to run the procedure $\cS^3.\text{Check}$ on $c$. We will focus on the status of the game at moment $t$. Define $S_1$ to be the set of the chains that have been processed at $\cS^3.\text{HonestyCheck}$ at $t$ in the $k+1$-th round of $G_5$. Define $S_2:=C_{\cS.\text{FComp}}/C_{\cS.\text{FComp}}[k]$, where $C_{\cS.\text{FComp}}$ is the set of full $\cS.\text{Completed}$ chains at $t$. Since $\cS^3.\text{HonestyCheck}$ and $\cS^3.\text{Complete}$ are the only two procedures that add new $\CF$ values, there exists an element $c^{*} \in S_1 \cup S_2$ such that $(s+n/10-1,x_{s+n/10-1}) \in Q_{c^{*}}$. By Lemma \ref{Lemma:query one point must query previous points}, $(i,x_i) \in Q_{c^{*}}$ for $i=s+n/10-10,\ldots,s+n/10-1$. Then, because of Lemma \ref{Lemma: No subverted chain covers a long honest chain}, $c^{*}$ and $(s+n/10-10,x_{s+n/10-10},\ldots,x_{s+n/10-1})$ are not disjoint. Finally, by definition of $\cS^3.\text{Check}$, $c$ will not be  processed by $\cS^3.\text{HonestyCheck}$.
\end{proof}

Next we prove a lemma that describes the property of $\cM^3.\CF[k]$. 

\begin{lemma}\label{Lemma:long chain in k is already completed}
If $G_5$ is efficient and $\GoodStatus[k]$ happens for some positive integer $k<q_{\cD}$, then for any length $3n/10+30$ increasing chain $c=(s,x_s,\ldots,x_{s+3n/10+29})$ in $\cM^3.\CF[k]$, $(s+3n/10+28,x_{s+3n/10+28},x_{s+3n/10+29})$ is in $\cM^3.\text{CompletedChains}[k]$ or $\cS^3.\text{CompletedChains}[k]$.  
\end{lemma}
\begin{proof}
We analyze the following two cases:
\begin{itemize}
\item \textbf{Case 1:} $(s+3n/10+10,x_{s+3n/10+10}) \in \cS^3.\CF[k]$.  

Because of $\GoodStatus[k]$, the chain $(s,x_s,\ldots,x_{s+3n/10+10})$ is in $\cS^3.\CF[k]$. Then suppose $(j,x_j)$ ($s+n/10+11 \leq j \leq s+3n/10+10$) is the last element of $(s+n/10+11,x_{s+n/10+11},\ldots,x_{s+3n/10+10})$ that is evaluated in $\cS^3.\CF[k]$. 
\begin{enumerate}
\item if $j\leq s+2n/10+11$, then $c_1=(j,x_{j},\ldots,x_{j+n/10-1})$ is a length $n/10$ subchain of $(s+n/10+11,x_{s+n/10+11},\ldots,x_{s+3n/10+10})$, which means $c_1$ was checked by the procedure $\cS^3.\text{Check}$ at some moment before the end of round $k$ of the game. Then, since any unsubverted chain that is not disjoint with $c_1$ consists of quasi-honest points, $(j,x_j,x_{j+1}) \in \cS^3.\text{CompletedChains}[k]$,
  therefore
  \[
    (s+3n/10+28,x_{s+3n/10+28},x_{s+3n/10+29}) \in \cS^3.\text{CompletedChains}[k]\,.
    \]
\item if $j\geq s+2n/10+12$, similar analysis can be used to show that $(j-1,x_{j-1},x_{j}) \in \cS^3.\text{CompletedChains}[k]$ and
  \[
    (s+3n/10+28,x_{s+3n/10+28},x_{s+3n/10+29}) \in \cS^3.\text{CompletedChains}[k]\,.
    \]
\end{enumerate}

\item \textbf{Case 2:} $(s+3n/10+10,x_{s+3n/10+10}) \notin \cS^3.\CF[k]$.

  Because of $\GoodStatus[k]$, none of the elements of the chain
  \[
    (s+3n/10+10,x_{s+3n/10+10},\ldots,x_{s+3n/10+29})
  \]
  is in $\cS^3.\CF[k]$. Suppose $(s+3n/10+29,x_{s+3n/10+29}) \in Q_{c'}$ for some $c' \in C_{\cM.\text{FComp}}$. By $\GoodStatus[k]$, all elements in the chain
  \[
    c_2=(s+3n/10+19,x_{s+3n/10+19},\ldots,x_{s+3n/10+29})
  \]
  are in $Q_{c'}$. By Lemma \ref{Lemma: No subverted chain covers a long honest chain}, $c_2$ and $c'$ are not disjoint. Since all the elements of $c_2$ are quasi-honest, $c_2 \subset c'$ and therefore,
  \[
    (s+3n/10+28,x_{s+3n/10+28},x_{s+3n/10+29}) \in \cM^3.\text{CompletedChains}[k]\,. \qedhere
    \]
\end{itemize}
\end{proof}

\begin{lemma}\label{Lemma:SNewBadComplete is negl}
If $G_5$ is efficient and for some positive integer $k<q_{\cD}$,
\[
  \Pr[\text{$\GoodStatus[k]$ does not happen}]= \negl(n),
\]
then
\[
  \Pr[\text{$\cS.\NewBadComplete_5[k+1]$}]=\negl(n).
\]
\end{lemma}
\begin{proof}
Suppose at some moment during the $(k+1)$-th interaction of $G_5$, the simulator $\cS^3$ is starting to execute procedure $\cS^3.\text{Complete}$ on a chain $(s,x_s,x_{s+1},u)$ such that $(s,x_s,x_{s+1}) \notin \cM^3.\text{CompletedChains}$. We will show that, with overwhelming probability, $\cS.\NewBadComplete_5$ does not happen in this execution.

Denote by $T_{\text{initial}}$ the table $\cM^3.\CF$ right before the execution of $\cS^3.\text{Complete}$ on $(s,x_s,x_{s+1},u)$. Assume, without loss of generality, that $s<u$. (Then, by our convention of the simulator, $u-s>n$.) Imagine we run $\cS^3.\text{Complete}$ on $(s,x_s,x_{s+1},u)$ and stop when the procedure just finished evaluating $\CFt_{s+9n/10+8}(x_{s+9n/10+8})$. We call the table $\cM^3.\CF$ at this moment $T_{\text{final}}$. 

Consider the subverted chain $c=(s,x_{s},\ldots,x_{s+9n/10+8})$ and its subchain $c'=(s+9n/10,x_{s+9n/10},\ldots,x_{s+9n/10+8})$. To prove $\cS.\NewBadComplete_5[k+1]$ does not happen during the completion of $(s,x_s,x_{s+1},u)$, it is sufficient to show that $c'$ has the following properties:
\begin{itemize}
\item $c'$ is honest, and therefore can be viewed as an unsubverted chain;
\item $c'$ is increasing;
\item None of the elements in $c'$ is in $T_{\text{initial}}$.
\end{itemize}

Before we prove these properties, it is helpful to see how these properties imply that the bad event $\cS.\NewBadComplete_5$ does not happen.
\begin{itemize}
\item First, because of Lemma \ref{Lemma:add one point at a time} and Theorem \ref{Th: nice properties of increasing chain}, the future subverted chain $(s+9n/10,x_{s+9n/10},\ldots,x_{8n})$ is honest and increasing (as an unsubverted chain). By Theorem \ref{Th: nice properties of increasing chain} again, for $m=u,u+1$, $(m,x_m)$ will not be queried by any $\CFt_i(x_i)$ with $s+9n/10+8<i\leq 8n$ and $i \neq m$.
\item Second, imagine evaluating the subverted chain $c''=(1,x_1,\ldots,x_{s+9n/10+8})$ (by uniformly assigning new $\CF$ values in need). We will show $(m,x_m) \notin Q_{c''}$ for $m=u,u+1$. Suppose $(m,x_m) \in Q_{c''}$. Since $c'$ is increasing, $(s+9n/10+9,x_{s+9n/10+9},\ldots,x_{m})$ is also increasing by Lemma~\ref{Lemma:add one point at a time}. This implies all the elements of $(s+9n/10+9,x_{s+9n/10+9},\ldots,x_{m})$ are in $Q_{c''}$, which is contradictory to Lemma \ref{Lemma: No subverted chain covers a long honest chain}. (Although the order we generate $\CF$ values here is not same as that in $G_5$, we still select $\CF$ values uniformly, which allows us to use the lemmas we showed about increasing chains.) 
\item Finally, since $(s+9n/10,x_{s+9n/10}) \notin T_{\text{initial}}$, $(m-1,x_{m-1}) \notin T_{\text{initial}}$ for $m=u,u+1$. By the randomness of $\CF_{m-1}(x_{m-1})$, $(m,x_m)$ is not in $T_{\text{initial}}$. 
\end{itemize}
   
We prove the above properties of $c'$ in three steps. 
\begin{itemize}
\item \textbf{Step 1:} We first show the chain $(s+n/2,x_{s+n/2},\ldots,x_{s+9n/10+8})$ is honest and increasing as an unsubverted chain.
\begin{itemize}
\item \textbf{Case 1:} if $(s,x_{s},\ldots,x_{s+3n/10+39})$ is not honest, then by Lemma \ref{Lemma:bad region is short} and Lemma \ref{Lemma:there is no separated bad region}, the chain
  \[
    (s+3n/10+n/6+39,x_{s+3n/10+n/6+39},\ldots,x_{s+9n/10+8})
  \]
  is honest. By Theorem \ref{Th: nice properties of increasing chain} and Lemma \ref{Lemma:add one point at a time}, the chain
  \[
    (s+3n/10+n/6+47,x_{s+3n/10+n/6+47},\ldots,x_{s+9n/10+8})
  \]
  is increasing.
\item \textbf{Case 2:} if $(s,x_{s},\ldots,x_{s+3n/10+39})$ is honest, then by Lemma~\ref{Lemma:add one point at a time} and Lemma \ref{Lemma:long chain in k is already completed},
  \[
    (s+3n/10+30,x_{s+3n/10+30},\ldots,x_{s+3n/10+39})
  \]
  is increasing. By Theorem \ref{Th: nice properties of increasing chain}, the chain
  \[
    (s+3n/10+39,x_{s+3n/10+30},\ldots,x_{s+9n/10+8})
  \]
  is increasing and honest.
\end{itemize}

\item \textbf{Step 2:} Next we show none of the elements in
  \[
    (s+4n/5+30,x_{s+4n/5+30},\ldots,x_{s+9n/10+8})
  \]
  is in $\cM^3.\CF[k]$. 

It suffices to show that $(s+4n/5+30,x_{s+4n/5+30}) \notin \cM^3.\CF[k]$. Suppose $(s+4n/5+30,x_{s+4n/5+30}) \in \cM^3.\CF[k]$, then by Lemma \ref{Lemma:long chain in k is already completed}, $(s+4n/5+30,x_{s+4n/5+29},x_{s+4n/5+30})$ is in $\cM^3.\text{CompletedChains}[k]$ or $\cS^3.\text{CompletedChains}[k]$. This contradicts the fact that $\cS^3.\text{Complete}$ only processes uncompleted chains and our assumption that $(s,x_s,x_{s+1}) \notin \cM^3.\text{CompletedChains}$.

\item \textbf{Step 3:} Lastly, we show that no elements in $(s+4n/5+40,x_{s+4n/5+40},\ldots,x_{s+9n/10+8})$ are in $T_{\text{initial}}$.

  It suffices to show that $(s+4n/5+40,x_{s+4n/5+40}) \notin T_{\text{initial}}$. Suppose $(s+4n/5+40,x_{s+4n/5+40}) \in T_{\text{initial}}$. Since $(s+4n/5+30,x_{s+4n/5+30}) \notin \cM^3.\CF[k]$, the chain
  \[
    (s+4n/5+30,x_{s+4n/5+30},\ldots,x_{s+4n/5+40})
  \]
  is in $T_{\text{initial}}$, which, by the proof of Lemma \ref{Lemma: all k+1 completed chains have a long subchain in k}, means that
  \[
    (s+4n/5+30,x_{s+4n/5+30},\ldots,x_{s+4n/5+40})
  \]
  is already a $\cS.\text{Completed}$ chain. This contradicts the fact that $\cS^3.\text{Complete}$ only processes uncompleted chains. \qedhere
\end{itemize} 
\end{proof}

\begin{lemma}\label{Lemma:MBadComplete is negl}
If $G_5$ is efficient and for some positive integer $k<q_{\cD}$,
\[
  \Pr[\text{$\GoodStatus[k]$ does not happen}]=\negl(n),
\]
then
\[
  \Pr[\text{$\cM.\BadComplete_5[k+1]$}]=\negl(n).
\]
\end{lemma}
\begin{proof}
The proof is similar to (and simpler than) that of Lemma~\ref{Lemma:SNewBadComplete is negl}.
\end{proof}

\begin{lemma}\label{Lemma:SExistingBadComplete is negl}
If $G_5$ is efficient and for some positive integer $k<q_{\cD}$,
\[
  \Pr[\text{$\BadEval_5[k+1]$}]=\negl(n),
\]
then
\[
  \Pr[\text{$\cS.\ExistingBadComplete_5[k+1]$}]=\negl(n).
\]
\end{lemma}
\begin{proof}
Suppose at some moment during the $(k+1)$-th interaction of $G_5$, the simulator $\cS^3$ is starting to execute procedure $\cS^3.\text{Complete}$ on a chain $(s,x_s,x_{s+1},u)$ such that $(s,x_s,x_{s+1}) \in \cM^3.\text{CompletedChains}$. 

Because of $\cS^3$'s convention of choosing $u$ and the fact that $\BadEval_5[k+1]$ is negligible, $u$ is equal to $4n$ in this case, which avoids the bad event $\cS.\ExistingBadComplete_5[k+1]$.
\end{proof}

\begin{proof}[Proof of Theorem \ref{Th: complete}]
The theorem is proved by combining Lemma \ref{Lemma:SNewBadComplete is negl}, Lemma \ref{Lemma:MBadComplete is negl}, Theorem \ref{Th: eval} and Lemma \ref{Lemma:SExistingBadComplete is negl}.
\end{proof}

\subsubsection{$\BadEval_5[k+1]$ is negligible}\label{subsub: Eval}
We assume $\GoodStatus[k]$ throughout this section to prove $\BadEval_5[k+1]$ is negligible. Consider the following experiment \textbf{Exp-Eval$[k+1]$} where an adaptive distinguisher $\cD_1$ interacts with $(\cS^3,\cM^3)$ and tries to trigger $\BadEval_5[k+1]$:

\begin{mdframed}
\begin{center}
  \textsc{Exp-Eval$[k+1]$}
\end{center}
\begin{enumerate}
\item $\cD_1$ makes $a$ queries for some $0 \leq a<k$ and outputs $a$.
\item At the $a+1$-th round of the game, $\cD_1$ queries $RF(x_0,x_1)$ for some pair of strings $(x_0,x_1)$. (By definition, $\cM^3$ will complete a subverted chain determined by $(x_0,x_1)$. We call this full subverted chain $c$.) 
\item $\cD_1$ makes another $b$ queries for $b=k-a$.
\end{enumerate}

We say $\cD_1$ wins \textbf{Exp-Eval$[k+1]$} if $c$ is not a $\cS.\text{Completed}$ chain by the end of the $k$-th round of the experiment and \textbf{Exp-Eval$[k+1]$} aborts at the $(k+1)$-th round with the bad event $\BadEval_5[k+1]$: $\cS^3.\CF^{\text{Inner}}$(or $\cM^3.\CF^{\text{Inner}}$) calls a term $(i,x)$ such that $3n \leq i \leq 5n$ and $(i,x) \in c$.

Here, without loss of generality, we do not consider the case where $\cD_1$ makes a $RF^{-1}(\cdot)$ query in round $a+1$.
\end{mdframed}

It is easy to see that Theorem \ref{Th: eval} is equivalent to:
\begin{theorem}[$\BadEval_5$ is negligible.]\label{Th: eval2}
If $G_5$ is efficient and for some positive integer $k<q_{\cD}$,
\[
  \Pr[\text{$\GoodStatus[k]$ does not happen}]=\negl(n),
\]
for any distinguisher $\cD$, then, for any distinguisher $\cD_1$,
\[
  \Pr[\text{$\cD_1$ wins \textbf{Exp-Eval$[k+1]$}}]=\negl(n).
\]
\end{theorem}

We take two steps to prove the theorem. First, notice that the only
information $\cD_1$ has to predict $(i,x) \in c$ is the answer to the
query $RF(x_0,x_1)$. We turn this observation into a lemma that says
$\cD_1$ can not win \textbf{Exp-Eval$[k+1]$} without receiving the
answer to the query $RF(x_0,x_1)$. Second, we prove that whether
receiving the answer or not does not affect too much the probability
that $\cD_1$ wins.

Consider the following modified version of \textbf{Exp-Eval$[k+1]$} where the distinguisher $\cD^2$ does not receive the answer to $RF(x_0,x_1)$.

\begin{mdframed}
\begin{center}
  \textsc{Exp-EvalNoAnswer$[k+1]$}
\end{center}
\begin{enumerate}
\item $\cD_2$ makes $a$ queries for some $0 \leq a<k$ and outputs $a$.
\item At the $a+1$-st round of the game, $\cD_2$ queries $RF(x_0,x_1)$ for some pair of strings $(x_0,x_1)$, but does not receive the answer. (By definition, $\cM^3$ will complete a subverted chain determined by $(x_0,x_1)$. We call this full subverted chain $c$.) 
\item $\cD_2$ makes another $b$ queries for $b=k-a$.
\end{enumerate}

The winning condition of \textbf{Exp-EvalNoAnswer$[k+1]$} is same as that of \textbf{Exp-Eval$[k+1]$}. 
\end{mdframed}

To prove $\cD_2$ wins \textbf{Exp-EvalNoAnswer$[k+1]$} with negligible probability, we consider the following distinguisher $\cD_3$, which uses $\cD_2$ in \textbf{Exp-EvalNoAnswer$[k+1]$} to trigger the bad event $\cM.\BadComplete_5[k+1]$.

\begin{mdframed}
\begin{center}
  \textsc{Exp-BadMComplete$[k+1]$}
\end{center}
\begin{enumerate}
\item $\cD_3$ uses $\cD_2$ to make $a$ queries for some $0 \leq a<k$ and outputs $a$. For each round of the game, $\cD_3$ receives the query from and gives the answer to $\cD_2$.
\item $\cD_3$ remembers the $a+1$-th query $RF(x_0,x_1)$ of $\cD_2$ but does not make the query $(x_0,x_1)$ to $RF(\cdot).$
\item $\cD_3$ uses $\cD_2$ to make another $b$ queries for $b=k-a$. Again, for each round of the game, $\cD_3$ receives the query from and gives the answer to $\cD_2$.
\item $\cD_3$ makes its last query $RF(x_0,x_1)$.
\end{enumerate}

We say $\cD_3$ wins \textbf{Exp-BadMComplete$[k+1]$} if the subverted chain corresponding to $(x_0,x_1)$ is not in $\cS.\text{Completed}$ at the end of the $k$-th round of the experiment and \textbf{Exp-BadMComplete$[k+1]$} aborts at the $(k+1)$-th round with the following event: when $\cM^3.\text{Complete}$ is completing the subverted chain $c$ determined by $(x_0,x_1)$, there is a term $(i,x) \in c$ such that $3n \leq i \leq 5n$ and $(i,x)$ is queried in the $k$-th round of the game.
\end{mdframed}

\begin{lemma}\label{Lemma:Mbadcomplete implies EvalNoAnswer is negl}
If for some positive integer $k<q_{\cD}$,
\[
  \Pr[\text{$\cM.\BadComplete_5[k+1]$ happens}]=\negl(n),
\]
for any distinguisher $\cD$, then, for any distinguisher $\cD_2$,
\[
  \Pr[\text{$\cD_2$ wins \textbf{Exp-EvalNoAnswer$[k+1]$}}]=\negl(n).
\]
\end{lemma}
\begin{proof}
Conditioned on the same randomness of $\CF$ in \textbf{Exp-EvalNoAnswer$[k+1]$} and \textbf{Exp-BadMComplete$[k+1]$}, we can see that if $\cD_2$ wins \textbf{Exp-EvalNoAnswer$[k+1]$}, $\cD_3$ will win \textbf{Exp-BadMComplete$[k+1]$}, which by assumption, is negligible.
\end{proof}

Now we proceed to show if $\cD_2$ wins \textbf{Exp-EvalNoAnswer$[k+1]$} with negligible probability, $\cD_1$ wins \textbf{Exp-Eval$[k+1]$} with negligible probability, too. 

Choose an arbitrary integer $m$ with $4n \leq m < 5n$. Consider the following variation of \textbf{Exp-Eval$[k+1]$} for $\cD_1'$.

\begin{mdframed}
\begin{center}
  \textsc{Exp-EvalRight$[k+1,m]$}
\end{center}
Select uniformly a full table of $\CF$ values $\CF_{\text{full}}$ and a $2n$-bit string $\alpha$. In step 1 and 3 of the experiment, all the $\CF$ values queried by the simulators are taken from $\CF_{\text{full}}$ instead of being selected uniformly as usual. In step 2, $\CF$ values queried by the simulators are generated with a special convention explained below.
\begin{enumerate}
\item Same as step 1 of \textbf{Exp-Eval$[k+1]$}. 
\item Same as step 2 of \textbf{Exp-Eval$[k+1]$} and the answer to the query $RF(x_0,x_1)$ is $\alpha$. 
\item Same as step 3 of \textbf{Exp-Eval$[k+1]$}. 
\end{enumerate}

In step 2, $\cM^3$ uses the following convention to fill in its table:
\begin{enumerate}
\item Set $(x_{8n},x_{8n+1}):=\alpha$. 
\item For $i=2,\ldots,m$, define $x_i:=x_{i-2} \oplus \CFt_{i-1}(x_{i-1})$. For $i=m+3,\ldots,8n+1$, define $x_{i-2}:=x_{i} \oplus \CFt_{i-1}(x_{i-1})$. All the $\CF$ values are taken from $\CF_{\text{full}}$.
\item Set $\CF_m(x_m):=x_{m-1} \oplus x_{m+1}$ and $\CF_{u+1}(x_{u+1}):=x_{u} \oplus x_{u+2}$. The game aborts if there is an index $j$ such that $3n \leq j \leq 5n$ and $(j,x_j)$ is in $\cM^3.\CF[a]$ or $\bigcup_{i=1}^{8n}Q_i(x_i)/Q_j(x_j)$.
\item Evaluate $\CFt_m(x_m)$ and $\CFt_{m+1}(x_{m+1})$ and the game aborts if there is an index $j$ such that $3n \leq j \leq 5n$ and $(j,x_j)$ is dishonest.
\end{enumerate}

The winning condition of the experiment is same as that of \textbf{Exp-Eval$[k+1]$} except that in this experiment, $\cD_1'$ wins if the index $i$ of the term $(i,x)$ causing the bad event has the range \textcolor{red}{$3n \leq i \leq m$}. For simplicity, we say $W_{\text{EvalRight}}(\CF_{\text{full}},\alpha,m)=1$ if $\cD_1'$ wins with the choice of $(\CF_{\text{full}},\alpha)$ in the experiment.
\end{mdframed}

From the proof of Lemma \ref{Lemma:4 vs 5}, we can see the distribution of $\cM^3.\CF$ in \textbf{Exp-EvalRight$[k+1,m]$} is same as that of \textbf{Exp-Eval$[k+1]$}. The only difference between the two experiments is that the new version has a more strict winning condition than the original one.

To prove the probability that $\cD_1$ wins \textbf{Exp-Eval$[k+1]$} is negligible, we first show, for any distinguisher $\cD_1'$, the probability that $\cD_1'$ wins \textbf{Exp-EvalRight$[k+1,m]$} is negligible. Consider the following rewrite of \textbf{Exp-EvalNoAnswer$[k+1]$}, where the randomness of $\CF$ is set like \textbf{Exp-EvalRight$[k+1,m]$} and $\cD_2$ uses $\cD_1'$ as an oracle to play the game.

\begin{mdframed}
\begin{center}
  \textsc{Exp-EvalNoAnswer$[k+1]$}
\end{center}

Select uniformly a full table of $\CF$ values $\CF_{\text{full}}$ and a pair of $2n$-bit strings $\alpha$ and $\beta$. The randomness of $\CF$ and $RF$ are set the same way as they are set in \textbf{Exp-EvalRight$[k+1,m]$}.

\begin{enumerate}
\item $\cD_2$ uses $\cD_1'$ to make $a$ queries for some $0 \leq a<k$ and outputs $a$. For each round of the game, $\cD_2$ receives the query from and gives the answer to $\cD_1'$.
\item At the $a+1$-st round of the game, $\cD_2$ uses $\cD_1'$ to query $RF(x_0,x_1)$ for some pair of strings $(x_0,x_1)$. $\alpha:=RF(x_0,x_1)$ is evaluated by $\cM^3$ but not returned to $\cD_2$. (By definition, $\cM^3$ will complete a subverted chain determined by $(x_0,x_1)$. We call this full subverted chain $c$.) $\cD_2$ selects $\beta$ uniformly and gives it to $\cD_1'$. 
\item $\cD_2$ uses $\cD_1'$ to make $b$ queries ($b=k-a$). For each round of the game, $\cD_2$ receives the query from and gives the answer to $\cD_1'$.
\end{enumerate}

For simplicity, we say $W_{\text{EvalNA}}(\CF_\text{full}, \alpha, \beta)=1$ if $\cD_2$ wins for the choice of $(\CF_\text{full}, \alpha, \beta)$. 
\end{mdframed}

\begin{lemma}\label{Lemma:EvalNoAnswer implies EvalRight}
If for some positive integer $k<q_{\cD}$,
\[
  \Pr[\text{$\cD_2$ wins \textbf{Exp-EvalNoAnswer$[k+1]$}}]=\negl(n)
\]
over the randomness of ($\CF_{\text{full}}$, $\alpha$, $\beta$), then, for any integer $4n \leq m <5n$,
\[ 
  \Pr[\text{$\cD_1'$ wins \textbf{Exp-EvalRight$[k+1,m]$}}]=\negl(n)
\]
over the randomness of ($\CF_{\text{full}}$, $\alpha$).
\end{lemma}
\begin{proof}
  Suppose that the first probability is negligible and the second is
  not. Then, with non-negligible probability, uniformly selecting a
  table $\CF_\text{full}$ and two $2n$-bit strings $\alpha$, $\beta$
  yields:
\begin{itemize}
\item $W_{\text{EvalRight}}(\CF_\text{full},\beta)=1$;
\item $W_{\text{EvalNA}}(\CF_\text{full},\alpha,\beta) \neq 1$.
\end{itemize}

For convenience, we call the two experiments above $E_1$ and $E_2$. Suppose the full subverted chain generated in the $(a+1)$-th round of $E_1$ is $c_1=(1,x_1,\ldots,x_m,y_{m+1},\ldots,y_{8n})$, and the chain in $E_2$ is $c_2=(1,x_1,\ldots,x_m,z_{m+1},\ldots,z_{8n})$. Notice that the first $m$ terms of $c_1$ and $c_2$ are same because the two experiments share the same table $\CF_{\text{full}}$. 

We will find a contradiction by showing that the queries to $\CF$ (made by the distinguisher or the simulators) are identical in $E_1$ and $E_2$ after the $(a+1)$-th round of the games. It is sufficient to show that $E_1$ or $E_2$ does not query $(m,x_m)$, $(m+1,y_{m+1})$ and $(m+1,z_{m+1})$, the only three terms that have different $\CF$ values in $E_1$, $E_2$ and $\CF_\text{full}$. Since $W_{\text{EvalNA}}(\CF_\text{full},\alpha,\beta) \neq 1$, $E_2$ does not query $(m,x_m)$ or $(m+1,z_{m+1})$. Since $W_{\text{EvalRight}}(\CF_\text{full},\beta)=1$, $E_1$ does not query $(m+1,y_{m+1})$. 

Since $W_{\text{EvalRight}}(\CF_\text{full},\beta)=1$ and $E_1,E_2$ make same queries to $\CF$ after the $(a+1)$-th round of the games, $W_{\text{EvalNA}}(\CF_\text{full},\alpha,\beta) = 1$. A contradiction.
\end{proof}

A similar proof can be used to show that, for any $3n < m \leq 4n$, no distinguisher can win the following game with non-negligible probability,
\begin{mdframed}
\begin{center}
  \textsc{Exp-EvalLeft$[k+1,m]$}
\end{center}
\begin{enumerate}
\item $\cD_1'$ makes $a$ queries for some $0 \leq a<k$ and outputs $a$.
\item At the $a+1$-th round of the game, $\cD_1'$ queries $RF(x_0,x_1)$ for some pair of strings $(x_0,x_1)$. (By definition, $\cM^3$ will complete a subverted chain determined by $(x_0,x_1)$. We call this full subverted chain $c$.) 
\item $\cD_1'$ makes another $b$ queries for $b=k-a$.
\end{enumerate}

We say $\cD_1'$ wins \textbf{Exp-Eval$[k+1]$} if $c$ is not a $\cS.\text{Completed}$ chain by the end of the $k$-th round of the experiment and \textbf{Exp-Eval$[k+1]$} aborts at the $(k+1)$-th round with the bad event $\BadEval_5[k+1]$: $\cS^3.\CF^{\text{Inner}}$(or $\cM^3.\CF^{\text{Inner}}$) calls a term $(i,x)$ such that \textcolor{red}{$m \leq i \leq 5n$} and $(i,x) \in c$.
\end{mdframed}

Summarizing the results above, we have:
\begin{lemma}\label{Lemma:EvalNoAnswer is negl implies Eval is negl}
If for some positive integer $k<q_{\cD}$,
\[
  \Pr[\text{$\cD_2$ wins \textbf{Exp-EvalNoAnswer$[k+1]$}}]=\negl(n),
\]
then, for any distinguisher $\cD_1$,
\[
  \Pr[\text{$\cD_1$ wins \textbf{Exp-Eval$[k+1]$}}]=\negl(n).
\]
\end{lemma}

Putting Lemma \ref{Lemma:Mbadcomplete implies EvalNoAnswer is negl} and Lemma \ref{Lemma:EvalNoAnswer is negl implies Eval is negl} together, we have:
\begin{lemma}\label{Lemma:Mbadcomplete implies Eval is negl}
If for some positive integer $k<q_{\cD}$,
\[
  \Pr[\text{$\cM.\BadComplete_5[k+1]$ happens}]=\negl(n),
\]
for any distinguisher $\cD$, then, for any distinguisher $\cD_1$,
\[
  \Pr[\text{$\cD_1$ wins \textbf{Exp-Eval$[k+1]$}}]=\negl(n).
\]
\end{lemma}

\begin{proof}[Proof of Theorem \ref{Th: eval2}]
The theorem is implied by combining Lemma \ref{Lemma:MBadComplete is negl} and Lemma \ref{Lemma:Mbadcomplete implies Eval is negl}.
\end{proof}

\subsubsection{$\GoodStatus[k]$ is overwhelming}\label{subsub: Goodstatus}
In this section we prove that, assuming $\GoodStatus[k]$, $\GoodStatus[k+1]$ happens with overwhelming probability. We introduce the following experiment between $\cD_4$ and $\cS^3,\cM^3$ to formulate the bad event:

\begin{mdframed}
\begin{center}
  \textsc{Exp-Status$[k+1]$}
\end{center}
\begin{enumerate}
\item $\cD_4$ makes $a$ queries for some $0 \leq a<k$ and outputs $a$.
\item At the $a+1$-th round of the game, $\cD_4$ queries $RF(x_0,x_1)$ for some pair of strings $(x_0,x_1)$. (By definition, $\cM^3$ will complete a subverted chain determined by $(x_0,x_1)$. We call this full subverted chain $c$.) 
\item $\cD_4$ makes another $b$ queries for $b=k-a$.
\end{enumerate}

We say $\cD_4$ wins \textbf{Exp-Status$[k+1]$} if there is an increasing chain $(i,x_i,x_{i+1},x_{i+2})$ in $\cM^3.\CF[k+1]$ such that $(i+1,x_{i+1}) \in Q_c$, $(i+1,x_{i+1}) \notin \cS^3.\CF[k+1]$, and:
\begin{itemize}
\item $(i+2,x_{i+2}) \in \cS^3.\CF[k+1]$, or
\item there is a chain $c' \in C_{\cM.\text{FComp}}[k+1]$ such that $(i+2,x_{i+2}) \in Q_{c'}$ and $(i+1,x_{i+1}) \notin Q_{c'}$.
\end{itemize}

Here, without loss of generality, we ignore the case when $(i,x_i,x_{i+1},x_{i+2})$ is decreasing.
\end{mdframed}

Notice that to prove Theorem \ref{Th: goodstatus}, it is sufficient to show the probability that $\cD_4$ wins \textbf{Exp-Status$[k+1]$} is negligible. We will restrict our attention to the experiment in the rest of the section.

\begin{definition}[Covering Index]
For any positive integer $k<q_{\cD}$, any index $i$ with $1 \leq i \leq 8n$, and any $x \in \{0,1\}^n$, we say the \emph{covering index} of $(i,x)$ at round $k$, $CI_{k}(i,x)$, is equal to $t$ ($t \geq 1$) if $(i,x) \in \cM^3.\CF[k]$, $(i,x) \notin \cS^3.\CF[k]$, and there are exactly $t$ elements $c_1,\ldots,c_{t} \in C_{\cM.\text{FComp}}[k]$ such that $(i,x) \in Q_{c_j}$ for $j=1,\ldots,t$. Otherwise, we say $CI_{k}(i,x)=0$. 
\end{definition} 

Suppose, with non-negligible probability, $\cD_4$ can win \textbf{Exp-Status$[k+1]$} with an increasing chain $(i,x_i,x_{i+1},x_{i+2})$ that has the property $CI_{k+1}(i+1,x_{i+1}) \geq 2$. Then, without loss of generality, we can assume $CI_{a}(i+1,x_{i+1}) \geq 1$, which means $(i+1,x_{i+1})$ has been evaluated before the $(a+1)$-st query from $\cD^4$. Using this observation, we rewrite \textbf{Exp-Status$[k+1]$} as:

\begin{mdframed}
\begin{center}
  \textsc{Exp-Status$[k+1]$}
\end{center}
\begin{enumerate}
\item $\cD_4$ makes $a$ queries for some $0 \leq a<k$ and outputs $a$.
\item At the $a+1$-th round of the game, $\cD_4$ queries $RF(x_0,x_1)$ for some pair of strings $(x_0,x_1)$. (By definition, $\cM^3$ will complete a subverted chain determined by $(x_0,x_1)$. We call this full subverted chain $c$.) 
\item $\cD_4$ makes another $b$ queries for $b=k-a$.
\end{enumerate}

We say $\cD_4$ wins \textbf{Exp-Status$[k+1]$} if there is an increasing chain $(i,x_i,x_{i+1},x_{i+2})$ in $\cM^3.\CF[k+1]$ such that $(i+1,x_{i+1}) \in Q_c$, $(i+1,x_{i+1}) \notin \cS^3.\CF[k+1]$, and
\begin{enumerate}
\item $CI_{k+1}(i+1,x_{i+1}) =1$, or $CI_{k+1}(i+1,x_{i+1}) \geq 2$ and $CI_{a}(i+1,x_{i+1}) \geq 1$;
\item $(i+2,x_{i+2}) \in \cS^3.\CF[k+1]$, or there is a chain $c' \in C_{\cM.\text{FComp}}[k+1]$ such that $(i+2,x_{i+2}) \in Q_{c'}$ and $(i+1,x_{i+1}) \notin Q_{c'}$.
\end{enumerate}

Here, without loss of generality, we ignore the case when $(i,x_i,x_{i+1},x_{i+2})$ is decreasing.
\end{mdframed}

As in the last section, we take two steps to prove $\cD_4$ wins \textbf{Exp-Status$[k+1]$} negligibly. First, notice that the only information that helps $\cD_4$ to generate a chain $(i,x_i,x_{i+1},x_{i+2})$ that breaks the good status is the answer to the query in the $(a+1)$-th round of the game. We turn this observation into a lemma that says, without receiving the answer in the $(a+1)$-th round, the distinguisher can not win the experiment. Second, we prove whether receiving the answer or not does not affect too much the probability that $\cD_4$ wins.

Consider the following modified version of \textbf{Exp-Status$[k+1]$} where the distinguisher $\cD_5$ does not receive the answer in round $(a+1)$.

\begin{mdframed}
\begin{center}
  \textsc{Exp-StatusNoAnswer$[k+1]$}
\end{center}
\begin{enumerate}
\item $\cD_5$ makes $a$ queries for some $0 \leq a<k$ and outputs $a$.
\item At the $a+1$-st round of the game, $\cD_5$ queries $RF(x_0,x_1)$
  for some pair of strings $(x_0,x_1)$, but does not receive the
  answer. (By definition, $\cM^3$ will complete a subverted chain
  determined by $(x_0,x_1)$. We call this full subverted chain $c$.)
\item $\cD_5$ makes another $b$ queries for $b=k-a$.
\end{enumerate}

The winning condition of \textbf{Exp-StatusNoAnswer$[k+1]$} is same as that of \textbf{Exp-Status$[k+1]$}. 
\end{mdframed}

\begin{lemma}\label{Lemma:StatusK implies StatusNoAnswerK+1 is negl}
If $G_5$ is efficient and for some positive integer $k<q_{\cD}$,
\[
  \Pr[\text{$\GoodStatus[k]$ does not happen}]=\negl(n),
\]
for any distinguisher $\cD$, then, for any distinguisher $\cD_1$,
\[
  \Pr[\text{$\cD_5$ wins \textbf{Exp-StatusNoAnswer$[k+1]$}}]=\negl(n).
\]
\end{lemma}
\begin{proof}
Suppose $\cD_5$ wins \textbf{Exp-StatusNoAnswer$[k+1]$} with some realization of the randomness of $\CF$: at the end of the experiment, there is a increasing chain $(i,x_i,x_{i+1},x_{i+2})$ in $\cM^3.\CF[k+1]$ such that $(i+1,x_{i+1}) \in Q_c$, $(i+1,x_{i+1}) \notin \cS^3.\CF[k+1]$, and $(i+2,x_{i+2}) \in \cS^3.\CF[k+1]$. (Without loss of generality, we ignore the case where $(i+2,x_{i+2}) \notin \cS^3.\CF[k+1]$.) 

Consider the following experiment:
\begin{mdframed}
\begin{center}
  \textsc{Exp-Monotone$[k+1]$}
\end{center}
\begin{enumerate}
\item $\cD_6$ uses $\cD_5$ to make $a$ queries for some $0 \leq a<k$ and outputs $a$. For each round of the game, $\cD_6$ receives the query from and gives the answer to $\cD_5$.
\item $\cD_6$ remembers the $a+1$-th query $RF(x_0,x_1)$ of $\cD_5$ but does not make the query $(x_0,x_1)$ to $RF(\cdot).$
\item $\cD_6$ uses $\cD_5$ to make another $b$ queries for $b=k-a$. Again, for each round of the game, $\cD_6$ receives the query from and gives the answer to $\cD_5$.
\item $\cD_6$ makes its last query $RF(x_0,x_1)$.
\end{enumerate}
\end{mdframed}

Imagine \textbf{Exp-Monotone$[k+1]$} has the same randomness of $\CF$ with \textbf{Exp-StatusNoAnswer$[k+1]$}. Then, $(i,x_i,x_{i+1},x_{i+2})$ is also a chain of $\cM^3.\CF[k+1]$ in \textbf{Exp-Monotone$[k+1]$}. And,
\begin{itemize}
\item if $CI_{k+1}(i+1,x_{i+1}) =1$, then $\CF_{i+1}(x_{i+1})$ is evaluated by $\cM^3$ after $\CF_i(x_1)$ and $\CF_{i+2}(x_{i+2})$, which violates Lemma \ref{Lemma:add one point at a time}; 
\item if $CI_{k+1}(i+1,x_{i+1}) \geq 2$, $CI_{a}(i+1,x_{i+1}) \geq 1$, then $\GoodStatus[k]$ does not happen.
\end{itemize}

The analysis above shows that, conditioned on the fact that $\GoodStatus[k]$ happens negligibly, $\cD_5$ wins negligibly.
\end{proof}

Now we proceed to show if $\cD_5$ wins \textbf{Exp-StatusNoAnswer$[k+1]$} with negligible probability, $\cD_4$ wins \textbf{Exp-Status$[k+1]$} with negligible probability, too. 

Choose an arbitrary integer $m$ with $4n \leq m < 5n$. Consider the following variation of \textbf{Exp-Status$[k+1]$} for $\cD_4'$.

\begin{mdframed}
\begin{center}
  \textsc{Exp-StatusRight$[k+1,m]$}
\end{center}
Select uniformly a full table of $\CF$ values $\CF_{\text{full}}$ and a $2n$-bit string $\alpha$. In step 1 and 3 of the experiment, all the $\CF$ values queried by the simulators are taken from $\CF_{\text{full}}$ instead of being selected uniformly as usual. In step 2, $\CF$ values queried by the simulators are generated with a special convention explained below.
\begin{enumerate}
\item Same as step 1 of \textbf{Exp-Status$[k+1]$}. 
\item Same as step 2 of \textbf{Exp-Status$[k+1]$} and the answer to the query $RF(x_0,x_1)$ is $\alpha$. 
\item Same as step 3 of \textbf{Exp-Status$[k+1]$}. 
\end{enumerate}

In step 2, $\cM^3$ uses the following convention to fill in its table:
\begin{enumerate}
\item Set $(x_{8n},x_{8n+1}):=\alpha$. 
\item For $i=(2,\ldots,m)$, define $x_i:=x_{i-2} \oplus \CFt_{i-1}(x_{i-1})$. For $m+3 \leq i \leq 8n+1$ define $x_{i-2}:=x_{i} \oplus \CFt_{i-1}(x_{i-1})$. All the $\CF$ values are taken from $\CF_{\text{full}}$.
\item Set $\CF_m(x_m):=x_{m-1} \oplus x_{m+1}$ and $\CF_{u+1}(x_{u+1}):=x_{u} \oplus x_{u+2}$. The game aborts if there is an index $j$ such that $3n \leq j \leq 5n$ and $(j,x_j)$ is in $\cM^3.\CF[a]$ or $\bigcup_{i=1}^{8n}Q_i(x_i)/Q_j(x_j)$.
\item Evaluate $\CFt_m(x_m)$ and $\CFt_{m+1}(x_{m+1})$ and the game aborts if there is an index $j$ such that $3n \leq j \leq 5n$ and $(j,x_j)$ is dishonest.
\end{enumerate}

The winning condition of the experiment is same as that of \textbf{Exp-Status$[k+1]$} except that in this experiment, $\cD_4'$ wins if the chain $(i,x_i,x_{i+1},x_{i+2})$ causing the bad event has the index $i$ with \textcolor{red}{$1 \leq i+1 \leq m$}. For simplicity, we say $W_{\text{StatusRight}}(\CF_{\text{full}},\alpha,m)=1$ if $\cD_4'$ wins for the choice of $(\CF_{\text{full}},\alpha)$ in the experiment.
\end{mdframed}

From the proof of Lemma \ref{Lemma:4 vs 5}, we can see the distribution of $\CF$ in \textbf{Exp-EvalRight$[k+1,m]$} is same as that of \textbf{Exp-Eval$[k+1]$}. The only difference between the two experiments is that the new version has a more strict winning condition than the original one.

To prove the probability that $\cD_1$ wins \textbf{Exp-Status$[k+1]$} is negligible, we first show, for any distinguisher $\cD_4'$, the probability that $\cD_4'$ wins \textbf{Exp-StatusRight$[k+1,m]$} is negligible. Consider the following rewrite of \textbf{Exp-StatusNoAnswer$[k+1]$}, where the randomness of $\CF$ is set like \textbf{Exp-StatusRight$[k+1,m]$} and $\cD_5$ uses $\cD_4'$ as an oracle to play the game.

\begin{mdframed}
\begin{center}
  \textsc{Exp-StatusNoAnswer$[k+1]$}
\end{center}

Select uniformly a full table of $\CF$ values $\CF_{\text{full}}$ and a pair of $2n$-bit strings $\alpha$ and $\beta$. The randomness of $\CF$ and $RF$ are set the same way as they are set in \textbf{Exp-StatusRight$[k+1,m]$}.

\begin{enumerate}
\item $\cD_5$ uses $\cD_4'$ to make $a$ queries for some $0 \leq a<k$ and outputs $a$. For each round of the game, $\cD_5$ receives the query from and gives the answer to $\cD_4'$.
\item At the $a+1$-st round of the game, $\cD_5$ uses $\cD_4'$ to query $RF(x_0,x_1)$ for some pair of strings $(x_0,x_1)$. $\alpha:=RF(x_0,x_1)$ is evaluated by $\cM^3$ but not returned to $\cD_5$. (By definition, $\cM^3$ will complete a subverted chain determined by $(x_0,x_1)$. We call this full subverted chain $c$.) $\cD_5$ selects $\beta$ uniformly and gives it to $\cD_4'$. 
\item $\cD_5$ uses $\cD_4'$ to make $b$ queries ($b=k-a$). For each round of the game, $\cD_5$ receives the query from and gives the answer to $\cD_4'$.
\end{enumerate}

For simplicity, we say $W_{\text{StatusNA}}(\CF_\text{full}, \alpha, \beta)=1$ if $\cD_5$ wins for the choice of $(\CF_\text{full}, \alpha, \beta)$. 
\end{mdframed}

\begin{lemma}\label{Lemma:StatusNoAnswer implies StatusRight}
If for some positive integer $k<q_{\cD}$,
\[
  \Pr[\text{$\cD_5$ wins \textbf{Exp-StatusNoAnswer$[k+1]$}}]=\negl(n)
\]
over the randomness of ($\CF_{\text{full}}$, $\alpha$, $\beta$), and 
\[
  \Pr[\text{$\BadEval_5[k+1]$}]=\negl(n),
\] 
for any distinguisher $\cD$, then, for any integer $4n \leq m <5n$,
\[ 
  \Pr[\text{$\cD_4'$ wins \textbf{Exp-StatusRight$[k+1,m]$}}]=\negl(n)
\]
over the randomness of ($\CF_{\text{full}}$, $\alpha$).
\end{lemma}
\begin{proof}
  Suppose that the first probability is negligible and the third is
  not. Then, with non-negligible probability, uniformly selecting a
  table $\CF_\text{full}$ and two $2n$-bit strings $\alpha$, $\beta$
  yields:
\begin{itemize}
\item $W_{\text{StatusRight}}(\CF_\text{full},\beta)=1$ and $\BadEval_5[k+1]$ does not happen in \textbf{Exp-StatusRight$[k+1,m]$} with the parameters $(\CF_\text{full},\beta)$;
\item $W_{\text{StatusNA}}(\CF_\text{full},\alpha,\beta) \neq 1$ and $\BadEval_5[k+1]$ does not happen in \textbf{Exp-StatusNoAnswer$[k+1]$} with the parameters $(\CF_\text{full},\alpha,\beta)$.
\end{itemize}

For convenience, we call the two experiments above $E_1$ and $E_2$. Suppose the full subverted chain generated in $(a+1)$-st round of $E_1$ is $c_1=(1,x_1,\ldots,x_m,y_{m+1},\ldots,y_{8n})$, and the chain in $E_2$ is $c_2=(1,x_1,\ldots,x_m,z_{m+1},\ldots,z_{8n})$. Notice that the first $m$ terms of $c_1$ and $c_2$ are same because the three experiments share the same table $\CF_{\text{full}}$. 

We will find a contradiction by showing that the queries to $\CF$ (made by the distinguisher or the simulators) are identical in $E_1$ and $E_2$ after the $(a+1)$-st round of the games. It is sufficient to show $E_1$ or $E_2$ does not query $(m,x_m)$, $(m+1,y_{m+1})$ and $(m+1,z_{m+1})$, the only three terms that have different $\CF$ values in $E_1$($E_2$) and $\CF_\text{full}$. This is directly implied by the fact that $\BadEval_5[k+1]$ does not happen in $E_1$ and $E_2$.

Since $W_{\text{StatusRight}}(\CF_\text{full},\beta)=1$ and $E_1,E_2$ make same queries to $\CF$ after the $(a+1)$-st round of the games, $W_{\text{StatusNA}}(\CF_\text{full},\alpha,\beta) = 1$. A contradiction.
\end{proof}

A similar proof can be used to show that, for any $3n < m \leq 4n$, no distinguisher can win the following game with non-negligible probability,
\begin{mdframed}
\begin{center}
  \textsc{Exp-StatusLeft$[k+1,m]$}
\end{center}
\begin{enumerate}
\item $\cD_4'$ makes $a$ queries for some $0 \leq a<k$ and outputs $a$.
\item At the $a+1$-st round of the game, $\cD_4'$ queries $RF(x_0,x_1)$ for some pair of strings $(x_0,x_1)$. (By definition, $\cM^3$ will complete a subverted chain determined by $(x_0,x_1)$. We call this full subverted chain $c$.) 
\item $\cD_4'$ makes another $b$ queries for $b=k-a$.
\end{enumerate}

We say $\cD_4'$ wins \textbf{Exp-Status$[k+1]$} if $c$ is not a $\cS.\text{Completed}$ chain by the end of the $k$-th round of the experiment and \textbf{Exp-Status$[k+1]$} aborts at the $(k+1)$-th round with the bad event $\GoodStatus_5[k+1]$: $\cS^3.\CF^{\text{Inner}}$(or $\cM^3.\CF^{\text{Inner}}$) calls a term $(i,x)$ such that \textcolor{red}{$m \leq i \leq 8n$} and $(i,x) \in c$.
\end{mdframed}

Summarize the results above and we have:
\begin{lemma}\label{Lemma:StatusNoAnswer is negl implies Ststus is negl}
If for some positive integer $k<q_{\cD}$,
\[
  \Pr[\text{$\cD_5$ wins \textbf{Exp-StatusNoAnswer$[k+1]$}}]=\negl(n)
\]
over the randomness of ($\CF_{\text{full}}$, $\alpha$, $\beta$), and 
\[
  \Pr[\text{$\BadEval_5[k+1]$}]=\negl(n),
\]
for any distinguisher $\cD$, then, for any distinguisher $\cD_4$,
\[
  \Pr[\text{$\cD_4$ wins \textbf{Exp-Status$[k+1]$}}]=\negl(n).
\]
\end{lemma}

\begin{proof}[Proof of Theorem \ref{Th: goodstatus}]
  The theorem is implied by combining Lemma \ref{Lemma:StatusK implies StatusNoAnswerK+1 is negl}, Theorem \ref{Th: eval} and Lemma \ref{Lemma:StatusNoAnswer is negl implies Ststus is negl}.
\end{proof}

\subsection{Efficiency of $G_5$}\label{section:efficiency}
In this section, we are going to show that the number of the elements in $\cS^3.\CF$ and $\cM^3.\CF$ are bounded by a polynomial function if the distinguisher $\cD$ makes at most $q_{\cD}$ ($q_{\cD}$ is polynomial) queries to $\CF$ or $RF$ (and $RF^{-1}$). 

\begin{lemma}\label{Lemma:No chain covers 3 disjoint length 11 chain}
If $G_5$ is efficient, then with overwhelming probability, there are not a chain $c=(1,w_1,\ldots,w_{8n})$ and three pairwise disjoint increasing chains $c_1=(i,x_{i},\ldots,x_{i+10})$, $c_2=(j,y_{j},\ldots,y_{j+10})$ and $c_3=(k,z_{k},\ldots,z_{k+10})$ in $\cM^3.\CF$, such that
\begin{itemize}
\item for all $(i,x) \in c$, $\CFt_i(x)$ is defined;
\item $c$ is disjoint with $c_1,c_2,c_3$;
\item $({i+10},x_{i+10}),({j+10},y_{j+10}),({k+10},z_{k+10}) \in Q_c$.
\end{itemize}
\end{lemma}
\begin{proof}
According to Lemma \ref{Lemma: No subverted chain covers a long honest chain}, if $({i+10},x_{i+10})\in Q_c$, then there exists an index $m$ ($i+1 \leq m \leq i+9$) such that in the length 3 monotone increasing chain $(x_{i-1},x_i,x_{i+1})$, $(i,x_i) \notin Q_c$ but $(i+1,x_{i+1}) \in Q_c$. Now we turn this observation into a proof.

Consider the following experiment in $G_5$. Take an arbitrary pair of $n$-bit strings $(w_1,w_2)$. $\cD$ tries to find a subverted chain $c$ starting with $(w_1,w_2)$ (w.l.o.g., we only consider subverted chain for convenience) and a length 3 increasing chain $(x_{i-1},x_i,x_{i+1})$ such that $(i,x_i) \notin Q_c$ and $(i+1,x_{i+1}) \in Q_c$. A quick thought reveals that the probability that $\cD$ wins is negligible ($8nq_{\cA}/s^n$): Suppose, without loss of generality, $\cD$ queries all the elements in $Q_c$ at the beginning of $G_5$. At some moment of the experiment, $\cD$ will select a pair of terms $(i-1,x_{i-1},x_i)$ as the starting pair of target length 3 chain. It is easy to see that, since $(i,x_i) \notin Q_c$, $(i+1,x_{i+1}) \in Q_c$ with probability not greater than $\poly(n) \cdot 8nq_{\cA}/2^n$, where $\poly(n)$ denotes the upper bound of the number of the terms in $\cM^3.\CF$.

For any pair $(w_1,w_2)$, we define the event:
\[
  E(w_1,w_2):= \left\{ \parbox{7cm}{
        There are three monotone increasing unsubverted chains $c_1=(i,x_{i},\ldots,x_{i+10})$, $c_2=(j,y_{j},\ldots,y_{j+10})$ and $c_3=(k,z_{k},\ldots,z_{k+10})$ in $\cM^3.\CF$, such that $c$ is disjoint with $c_1,c_2,c_3$ and $({i+10},x_{i+10}),({j+10},y_{j+10}),({k+10},z_{k+10}) \in Q_c$, where $c$ is the subverted starting with $(w_1,w_2)$}
    \right\}
  \]
  Finally we have 
\[
\sum_{(w_1,w_2) \in \{0,1\}^{2n}}\Pr[E(w_1,w_2)] < 2^{2n} \cdot (\poly(n) \cdot 8nq_{\cA}/2^n)^3 = \negl(n). \qedhere
\]
\end{proof}

\begin{lemma}\label{Lemma:No chain covers 21 length 11 chain}
Suppose $G_5$ is efficient. Let $C_{\text{11}}$ be a set of length 11 increasing chains and $c$ be a chain in $\cM^3.\CF$ such that $c$ is disjoint with any element in $C_{\text{11}}$, and for all $(i,x) \in c$, $\CFt_i(x)$ is defined. Then, with overwhelming probability, there are at most 20 chains $c'=(i,x_i,\ldots,x_{i+10}) \in C_{\text{11}}$ such that $(i+10,x_{i+10}) \in Q_c$.
\end{lemma}
\begin{proof}
Suppose there are 21 chains in $C_{\text{11}}$ that satisfy the property in the lemma. Notice that for each length 11 chain $c'$, there are at most 9 other length 11 chains that are not disjoint with $c'$. Then, among the 21 chains satisfying the property in the lemma, we can find 3 pairwise disjoint chains. This contradicts Lemma~\ref{Lemma:No chain covers 3 disjoint length 11 chain}.
\end{proof}

\begin{definition}[Order of a chain]
We define the order of an unsubverted chain $c=(s,x_s,\ldots,x_{s+r})$ in $\cM^3.\CF$ to be:
\[
  O_{\cM^3}(c):=\min_{k=s,\ldots,s+r-1}\Bigl\{\max\{O_{\cM^3}(k,x_k),O_{\cM^3}(k+1,x_{k+1})\}\Bigr\}
\]
Intuitively speaking, the order of a chain describes the time when a chain is ``determined.''
\end{definition}

\begin{lemma}\label{Lemma:the number of the elements are linear in the number of the unsubverted chains}
Suppose $G_5$ is efficient. Let $C_{\text{Disj}}$ be a set of pairwise disjoint unsubverted chains with length greater than or equal to 4 in $\cM^3.\CF$. Define the set $A$ to be the set of the elements of the chains in $C_{\text{Disj}}$. Then, with overwhelming probability, $|A| \geq \sum_{c \in C_{\text{Disj}}}(L(c)-3)$.
\end{lemma}
\begin{proof}
For any $c \in C_{\text{Disj}}$ and a term $(i,x)$ in $c$, we say $(i,x)$ is \emph{original} in $c$ if there does not exist a different element $c' \in C_{\text{Disj}}$ such that $c$ and $c'$ intersects at $(i,x)$ and $O_{\cM^3}(c) \geq O_{\cM^3}(c')$. Notice that a term $(i,x)$ can be original in at most one chain.

Now we are going to show that, with overwhelming probability, each element in $C_{\text{Disj}}$ contains at most 3 non-original terms. Suppose there is a chain $c=(s,x_s,\ldots,x_{s+r})$ that has four non-original terms. Then there are two non-original elements, $(i,x_i)$ and $(j,x_j)$, such that $s \leq i < j-2 \leq s+r-2$. Because of Lemma \ref{Lemma:add one point at a time}, without loss of generality, we assume
\[
  O_{\cM^3}(i,x_i)>O_{\cM^3}(i+1,x_{i+1})>O_{\cM^3}(i+2,x_{i+2})\,.
\]
Since $(i,x_i)$ is non-original in $c$, there is a chain $c' \neq c$ such that $(i,x_i) \in c'$ and $O_{\cM^3}(c) \geq O_{\cM^3}(c')$. Since $c$ and $c'$ are disjoint, $(i+1,x_{i+1}) \notin c'$. Then, since $O_{\cM^3}(c) \geq O_{\cM^3}(c')$ and $O_{\cM^3}(i+1,x_{i+1})>O_{\cM^3}(i+2,x_{i+2})$, we have $O_{\cM^3}(c) > O_{\cM^3}(c')$, which means $(i+1,x_{i+1})$ is not evaluated when $c'$ has been determined. Finally, because $\cM^3.\CF(i+1,x_{i+1})$ is selected uniformly, $(i,x_i) \in c'$ with negligible probability. A contradiction.

Going back to the proof of the lemma, since each term is original in at most one chain and each chain in $C_{\text{Disj}}$ has all but 3 original elements, $|A|$ is lower bounded by the sum of the original terms in the elements of $C_{\text{Disj}}$, which is not less than $\sum_{c \in C_{\text{Disj}}}(L(c)-3)$.
\end{proof}

\begin{theorem}\label{Th:efficiency}[Efficiency of $G_5$]
For any positive integer $k \leq q_{\cD}$, with overwhelming probability, at the end of the $k$-th round of $G_5$, there are fewer than $(88q_{\cA}+1)k$ terms in $\cS^3.\CF$ and fewer than $8nq_{\cA}k$ terms in $\cM^3.\CF$.
\end{theorem}

\paragraph{Remark.} In the proof of Theorem \ref{Th:efficiency}, we will make use of Lemma \ref{Lemma:No chain covers 21 length 11 chain} and Lemma \ref{Lemma:the number of the elements are linear in the number of the unsubverted chains}. However, these lemmas already take efficiency of $G_5$ as their assumptions. To reassure the reader that there is not a circular argument here, we imagine that the $k$-th round of the game is forced to end when $\cS^3.\CF$ contains more than $(88q_{\cA}+1)k$ elements or $\cM^3.\CF$ contains more than $8nq_{\cA}k$ elements. In this way, we can also feel free to reason about the tables of the simulators at the end of $k$-th round of the game (e.g., $\cS^3.\CF[k]$, $C_{\cS.\text{FComp}}[k]$).

\begin{proof}

In $\cS^3.\CF[k]$, for any unsubverted chain $c$, we call $c$ a \emph{generator} if $c$ was processed by the procedure $\cS^3.\text{HonestCheck}$. We denote by $C_G$ the set of generators. We define a function $g$ from $C_{\cS.\text{FComp}}[k]$ to $C_G$: for each $c_1 \in C_{\cS.\text{FComp}}[k]$ and $c \in C_G$, we say $g(c_1)=c$ if $c_1 \subset c$.

Define
\[
  G:=\{(i,x)\mid \text{there is $c \in  C_G$ such that $(i,x) \in c$.}\}\,.
\]
Since $C_G$ is a set of pairwise disjoint chains, by Lemma~\ref{Lemma:the number of the elements are linear in the number of the unsubverted chains},
\begin{equation}\label{G lower bound}
  |G| \geq \sum_{c \in C_G}(L(c)-3)=(n/10-3) \cdot |C_G|.
\end{equation}

To understand the structure of $G$, we define several subsets of $G$. We say a point $(i,x) \in G$ is a \emph{tail point} if there is an increasing $c_2=(s,x_s,\ldots,x_{s+10})$ in $\cM^3.\CF$ (w.l.o.g., we only consider the increasing case) and a chain $c \in C_G$ such that $(i,x)=(s+10,x_{s+10})$ and $c_2 \subset c$. We say a point $(i,x) \in G$ is a \emph{head point} if it is not a tail point. We denote the sets of the head points and tail points by $G_{\text{Head}}$ and $G_{\text{Tail}}$, respectively. For any point $(i,x) \in G_{\text{Tail}}$ and any chain $c \in C_G$, we say $c$ covers $(i,x)$ ($(i,x) \notin c$) if $(i,x) \in Q_c$ or $(i,x) \in Q_{f^{-1}(c)}$ (if $c$ has a preimage in function $f$). We define $G_{\text{Query}}$ to be the set of the points in $G_{\text{Tail}}$ that are not covered by any element in $C_G$. Notice that any element in $G_{\text{Query}}$ was queried directly by the distinguisher $\cD$. Our goal is to show the size of the set $G_{\text{Query}}$ is big.

By Lemma \ref{Lemma:add one point at a time}, the number of the elements in $G_{\text{Head}}$ is easily bounded by 
\begin{equation}\label{upper bound of Head}
  |G_{\text{Head}}| \leq 19 \cdot |C_G|.
\end{equation}

By Lemma \ref{Lemma:No chain covers 21 length 11 chain}
\begin{equation}\label{upper bound of Tail}
  |G_{\text{Tail}}/G_{\text{Query}}| \leq 20 \cdot |C_G|.  
\end{equation}

Summarizing Equation \ref{G lower bound},  \ref{upper bound of Head} and \ref{upper bound of Tail}, we have 
\begin{align*}
 & \mkern20mu |G_{\text{Query}}|\\
 & = |G|-|G_{\text{Head}}|-|G_{\text{Tail}}/G_{\text{Query}}|\\
 & \geq (n/10-3)|C_G|-19|C_G|-20|C_{Fu}|\\
 & = (n/10-42)|C_G|\,.
\end{align*}

This implies that
\begin{align*}
 & \mkern20mu |\cS^3.\CF[k]|\\
 & \leq 8n \cdot q_{\cA} \cdot |C_G| +k\\
 & \leq 8n \cdot q_{\cA} \cdot |G_{\text{Query}}|/(n/10-42) +k\\
 & \leq 8n \cdot q_{\cA} \cdot k/(n/10-42) +k\\
 & \leq 8n \cdot q_{\cA} \cdot k/(n/11) +k\\
 & = (88q_{\cA}+1)k.\\
\end{align*}

Suppose $\cD$ makes $t$ ($0 \leq t \leq k$) queries to the ideal object and $k-t$ queries to $\CF$, then 
\begin{align*}
 & \mkern20mu |\cM^3.\CF[k]|\\
 & \leq 8n \cdot q_{\cA} \cdot (k-t)/(n/10-42) +(k-t) + 8n \cdot q_{\cA} \cdot t\\
 & < (88q_{\cA}+1)(k-t) + 8nq_{\cA}t\\
 & \leq 8nq_{\cA}k
\end{align*}
when $n$ is large.

We remark that all the statements in the proof are true with overwhelming probability, we omit ``with overwhelming probability'' for simplicity.
\end{proof}

\subsection{Crooked indifferentiability in the full model}
\label{section:full model}

Now we show the simulator $\cS$ achieving abbreviated crooked indifferentiability can be lifted to a simulator that achieves full indifferentiability (Definition~\ref{def:abbrev-indiff-crooked}).

\begin{theorem}\label{Abbr. to full indiff}
If the construction in Section \ref{subsec:our contribution and construction} is $(n',n,q_{\cD},q_{\cA},r,\epsilon')$-Abbreviated-$H$-crooked-indifferentiable from a random oracle $F$, it is $(n',n,q_{\cD},q_{\cA},r,\epsilon'+8n \cdot q_{\cD}^{n/10} \cdot  2^{(2n-n^2/10)})$-$H$-crooked-indifferentiable from $F$.
\end{theorem}

\begin{proof}
  Consider the following simulator $\cS_{F}$ built on $\cS$:
\begin{enumerate}
    \item In the first phase, $\cS_{F}$ answers $f_i(x)$ ($1 \leq i \leq 8n$) queries uniformly. 
    \item The second phase, after which $\cS_F$ receives $R$, is
      divided into two sub-phases.
    \begin{itemize}
    \item First, $\cS_F$ simulates $\cS$ in $G_1$. It then plays the
      role of the distinguisher, and asks $\cS$ all the questions that
      were actually asked by the the distinguisher in the first
      phase. $\cS_F$ \emph{aborts} the game if, in this sub-phase,
      there are $n/10$ (simulated) queries such that they form a length $n/10$ unsubverted chain.
    \item Second, $\cS_{F}$ simulates $\cS$ and answers the
      second-phase questions from the distinguisher.
    \end{itemize}
\end{enumerate}
For an arbitrary full model distinguisher $\cD_F$, we construct the an
abbreviated model distinguisher $\cD$ as follows. The proof will show
that, with high probability, the execution that takes place between
$\cD$ and $\cS$ can be ``lifted'' to an associated execution between
$\cD_F$ and $\cS_F$.
\begin{enumerate}
\item Prior to the game, $\cD$ must publish a subversion algorithm
  $\cA$. This program is constructed as follows. To decide how
  to subvert a certain term $f_i(x)$, $\cA$ first simulates the
  first phase of $\cD_F$; all queries made by this simulation are
  asked as regular queries by $\cA$ and, at the conclusion, this
  first phase of $\cD_F$ produces, as output, a subversion algorithm
  $\cA_F$. $\cA$ then simulates the algorithm
  $\cA_F$ on the term $h_i(x)$.
\item In the game, $\cD$ simulates the queries of $\cD_{F}$ in
  $\cD_{F}$'s first phase. After that, $\cD$ continues to simulate
  $\cD_{F}$ in the second phase. (Note that at the point in
  $\cD_{F}$'s game where it produces the subversion algorithm
  $\cA$, this is simply ignored by $\cD$.)
\end{enumerate}

Now we are ready to prove $\cS_{F}$ is secure against the arbitrarily
chosen distinguisher $\cD_F$. We organize the proof around four different transcripts:
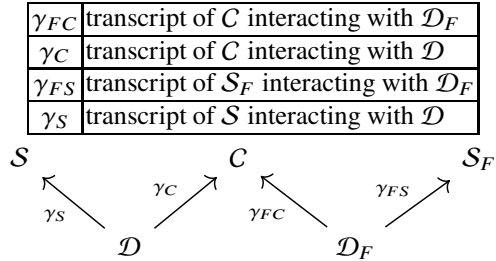
\begin{center}
\begin{tabular}{|c|l|}
  \hline
  $\gamma_{FC}$ & transcript of $\cC$ interacting with $\cD_{F}$\\ \hline
  $\gamma_C$ & transcript of $\cC$ interacting with $\cD$\\ \hline
  $\gamma_{FS}$ & transcript of $\cS_{F}$ interacting with $\cD_F$\\ \hline
  $\gamma_{S}$  & transcript of $\cS$ interacting with $\cD$\\ \hline
\end{tabular}
\qquad
\begin{tikzcd}
  \cS  &     & \cC &       & \cS_F\\
  & \cD \arrow[ul,"\gamma_S"] \arrow[ur,"\gamma_C"] &     & \cD_F \arrow[ul,"\gamma_{FC}"] \arrow[ur,"\gamma_{FS}"]& 
\end{tikzcd}
\end{center}
(Here $\cC$ denotes the construction, as usual.)
  %
Since
\[
  \|\gamma_{FC}-\gamma_{FS}\|_{\tv} \leq
  \|\gamma_{FC}-\gamma_C\|_{\tv}+\|\gamma_C-\gamma_S\|_{\tv}+\|\gamma_S-\gamma_{FS}\|_{\tv}\,,
\]
it is sufficient to prove the three terms in the right-hand side of
the inequality are all negligible.
\begin{itemize}
\item[-]  $\|\gamma_{FC}-\gamma_C\|_{\tv}=0$.
  This is obvious by observing that $\gamma_{FC}=\gamma_C$ when the underlying values of $H$ are the same.
\item[-]  $\|\gamma_C-\gamma_S\|_{\tv}= \epsilon'$.
  This is true because $\cS$ achieves abbreviated crooked indifferentiability.
\item[-] $\|\gamma_S-\gamma_{FS}\|_{\tv}= 8n \cdot q_{\cD}^{n/10} \cdot  2^{(2n-n^2/10)}$.  To prove this
  statement, we suppose both the full model game and the abbreviated
  model game select all randomness \emph{a priori} (as in the
  descriptions above). For the game between $\cS_F$ and $\cD_F$(the full game)
  or the game between $\cS$ and $\cD$(the abbreviated game), suppose we select a table $T_F$ of $F(i,x)$ values for all $1 \leq i \leq 8n$ and $x \in \{0,1\}^n$. When the simulator(in the full or abbreviated game) needs to assign a certain term uniformly, it takes the value from the table $T_F$. Suppose the full and the abbreviated game share the same table $T_F$ and $R$. Notice that the two games have same
  transcripts unless, $\cS_F$ aborts the game in the first sub-phase
  of the second phase. We denote this bad event by \textbf{LongChain},
  which by the following lemma \ref{Lemma:No long chain selected before R is given}, is negligible. \qedhere
 \end{itemize}
\end{proof}

N.b. While the description of the simulator above calls for all
randomness to be generated in advance, it is easy to see that the
simulator can in fact be carried out lazily with tables.

\begin{lemma}\label{Lemma:No long chain selected before R is given}
For any distinguisher $\cD_F$, the probability that \textbf{LongChain} happens is less than $8n \cdot q_{\cD}^{n/10} \cdot  2^{(2n-n^2/10)}$.
\end{lemma}
\begin{proof}
For any pair of $n$-bit strings $(x,y)$ and $i$ with $1 \leq i  \leq 8n$, over the randomness of $R$, the probability that \textbf{LongChain} happens with a length $n/10$ unsubverted chain starting with $(i,x,y)$ is bounded by $(q_{\cD} \cdot 2^{-n})^{n/10}$, which is equal to  $q_{\cD}^{n/10} \cdot  2^{-n^2/10}$. The lemma follows by taking the union bound over the choice of $(x,y)$ and $i$.
\end{proof}

\subsection{Lower Bound of the Number of Round Functions}

\paragraph{Remark.} Crooked security still holds if $\ell=8n$ is replaced by $\ell = 2000n/\log(1/\epsilon)$. Formally speaking, in order to fit the games, lemmas and theorems in previous sections to the parameter $\ell = 2000n/\log(1/\epsilon)$, we only need to replace all the linear parameters $c \cdot n$ by $c/8 \cdot 2000n/\log(1/\epsilon)$ in these games, lemmas, and theorems. For example, when $\ell = 2000n/\log(1/\epsilon)$, the simulator programs at $u=4/8 \cdot 2000n/\log(1/\epsilon)$ or $u=7/8 \cdot 2000n/\log(1/\epsilon)$ (instead of $u=4n$ or $u=7n$ when $\ell=8n$). 

To see why we need $\ell \geq 2000n/\log(1/\epsilon)$, notice that in Lemma \ref{Lemma:bad region is short}, when we set $\ell = 2000n/\log(1/\epsilon)$, the upper bound of the length of a bad region is $1/48 \cdot 2000n/\log(1/\epsilon)$. Therefore, the probability in the proof will be $8n \cdot 2^{2n} \cdot (14\epsilon)^{(1/(48 \cdot 14)) \cdot 2000n/\log(1/\epsilon)-1}=\negl(n).$


\bibliographystyle{abbrv}

\newpage

\appendix

\section{Detailed descriptions of the security games}

\subsection{Game 1}\label{Game 1}

\begin{algorithm}[H]
\DontPrintSemicolon

  \SetKwInput{KwSystemS}{System $\cS$}
  \SetKwInput{KwVariable}{Variable}{}

  \KwSystemS \;
  \KwVariable
  {\;
  Queue: $Q_{\cS}$ \;
  Tables: $\cS.\CF_1,...,\cS.\CF_{8n}$\;
  Order function: $O_{\cS}$\;
  Set $\cS.\text{CompletedChains}:=\emptyset$\;
  Set $\cS.\text{HonestyCheckedChains}:=\emptyset$\;
  Sets $Q_i(z):=\emptyset$ for all $i \in \{1,\ldots,8n\}$ and $x \in \{0,1\}^n$\;
  Sets $Q_i:=\emptyset$ for all $i \in \{1,\ldots,8n\}$\;
  Hashtable $P \subset \{\uparrow,\downarrow\} \times \{0,1\}^{2n} \times \{0,1\}^{2n}$ \;}
  \BlankLine
  \BlankLine
  
  \SetKwProg{public}{public procedure}{:}{}
  \SetKwProg{private}{private procedure}{:}{}

  \public{$\CF(i,x)$}
    { 
      $\cS.\CF^{\text{Inner}}(i,x)$\;
      \While{$\neg Q_{\cS}.\text{Empty}()$}
         {
            $(s,x_s,...,x_{s+n/10-1}):=Q_{\cS}.\text{Dequeue}()$\;
            \If{$\cS.\text{Check}(s,x_s,...,x_{s+n/10-1})=(s,x_s,x_{s+1},u)$}
               {
                 $\cS.$Complete $(s,x_s,x_{s+1},u)$
               }
         }
        Return $\cS.\CF_i(x)$
    }

  \BlankLine
  \BlankLine

  \public{$P(x_0,x_1)$}
    {
      \While{$(\downarrow,x_0,x_1) \notin P$}
        {
          $x_{8n} \leftarrow_R \{0,1\}^n$\;
          $x_{8n+1} \leftarrow_R \{0,1\}^n$\;
          \If{$(\uparrow,x_{8n},x_{8n+1}) \notin P$}
            {
              $P(\downarrow,x_0,x_1):=(x_{8n},x_{8n+1})$\;
              $P(\uparrow,x_{8n},x_{8n+1}):=(x_0,x_1)$
            } 
        }
      Return $P(\downarrow,x_0,x_1)$\;
    }
    
  \BlankLine
  \BlankLine
  
  \public{$P^{-1}(x_{8n},x_{8n+1})$}
    {
      \While{$(\uparrow,x_{8n},x_{8n+1}) \notin P$}
        {
          $x_0 \leftarrow_R \{0,1\}^n$\;
          $x_1 \leftarrow_R \{0,1\}^n$\;
          \If{$(\downarrow,x_0,x_1) \notin P$}
            {
              $P(\downarrow,x_0,x_1):=(x_{8n},x_{8n+1})$\;
              $P(\uparrow,x_{8n},x_{8n+1}):=(x_0,x_1)$ 
            }        
        }
      Return $P(\uparrow,x_{8n},x_{8n+1})$\;
    }  

  \BlankLine
  \BlankLine
 
  \private{$\cS.\CF^{\text{Inner}}(i,x)$}
    {
      \If{$x \notin \cS.\CF_i$}
        {
          $\cS.\CF_i(x) \leftarrow_R\{0,1\}^n$\;
          $\cS.\text{EnqueueNewChain}(i,x)$
        }
      Return $\cS.\CF_i(x)$    
    }

  \BlankLine
  \BlankLine
  
    \private{$\cS.\CFt^{\text{Inner}}(i,x)$}
      {
        \While{$\cA(i,x)$ queries $\CF_j(y)$} 
          {
            $Q_i(x):=Q_i(x) \cup \{(j,y)\}$\;
            $\cS.\CF^{\text{Inner}}(j,y)$
          }
        \tcc{Simulates the subversion algorithm $\cA$ on input $(i,x)$}
        Return $\cA(i,x)$
      }
      
    \BlankLine
    \BlankLine

\end{algorithm}

\begin{algorithm}
\DontPrintSemicolon

  \SetKwProg{public}{public procedure}{:}{}
  \SetKwProg{private}{private procedure}{:}{}
  
  \setcounter{AlgoLine}{33}

    \private{$\cS.\text{EnqueueNewChain}(i,x)$}
      {
        \ForAll{$(x_{i-n/10+1},...,x) \in \cS.\CF_{i-n/10+1} \times \cdot\cdot\cdot \times \cS.\CF_i$}
          {
            $Q_{\cS}.\text{Enqueue}(i-n/10+1,x_{i-n/10+1},...,x)$
          }  
        \ForAll{$(x,x_{i+1}...,x_{i+n/10-1}) \in \cS.\CF_i \times \cdot\cdot\cdot \times \cS.\CF_{i+n/10-1}$}
          {
            $Q_{\cS}.\text{Enqueue}(i,x,x_{i+1}...,x_{i+n/10-1})$
          } 
      }
      
    \BlankLine
    \BlankLine      

    \private{$\cS.\text{Check}(s,x_s,...,x_{s+n/10-1})$}
      {
        \If
        {$(i,x_i,x_{i+1}) \notin \cS.\text{Completedchains} \cup \cS.\text{HonestyCheckedChains}$ for all $s \leq i \leq s+n/10-2$}
        {$\cS.\text{HonestyCheck}(s,x_s,...,x_{s+n/10-1})$}
      }
      
    \BlankLine
    \BlankLine  
    
    \private{$\cS.\text{HonestyCheck}(s,x_s,...,x_{s+n/10-1})$}
      {  
        k:=1\;
        i:=s\;
        \While{$s \leq i \leq s+n/10-2$}
               {
                $\cS.\text{HonestyCheckedChains}:=\cS.\text{HonestyCheckedChains} \cup (i,x_i,x_{i+1})$\;
                \If{$\cS.\CFt^{\text{Inner}}(i,x_i)=\cS.\CF_i(x_i)$ and $\cS.\CFt^{\text{Inner}}(i+1,x_{i+1})=\cS.\CF_{i+1}(x_{i+1})$}{k:=k+1}
                i:=i+1
               }
          \If{k=n/10}{
            \If{$i+n/10-1<3n$ or $i>5n$}
                  {
                    Return $(s,x_s,x_{s+1},4n)$
                  }
                \Else{Return $(s,x_s,x_{s+1},7n)$}    
          }
      }
      
    \BlankLine
    \BlankLine

  \private{$\cS.$Complete $(s,x_s,x_{s+1},u)$}
    {
    
      $(x_{u-2},x_{u-1}):=\cS.\text{EvaluateForward}(s,x_s,x_{s+1},u)$\;
      $(x_{u+2},x_{u+3}):=\cS.\text{EvaluateBackward}(s,x_s,x_{s+1},u)$\;
      $\cS.\text{Adapt}(x_{u-2},x_{u-1},x_{u+2},x_{u+3},u)$
    }
    
  \BlankLine
  \BlankLine

    \private{$\cS.\text{EvaluateForward}(s,x_s,x_{s+1},u)$}
      {
        $\cS.\text{CompletedChains}:=\cS.\text{CompletedChains} \cup \{(s,x_{s},x_{s+1})\}$\;
        \While{$s \neq u-1$}
          {
            \If{$s=8n$}
              {
                $(x_0,x_1):=P^{-1}(x_{8n},x_{8n+1})$\;
                $s:=0$
              }
            \Else
            {
              $\cS.\text{CompletedChains}:=\cS.\text{CompletedChains} \cup \{(s+1,x_{s+1},x_{s+2})\}$\;
              $x_{s+2}:=x_s \oplus \cS.C \tilde{F}^{\text{Inner}}(s+1,x_{s+1})$\;
              $Q_{s+1}:=Q_{s+1}(x_{s+1})$\;
              $s:=s+1$
            }
          }
        Return $(x_{s-1},x_s)$  
        } 
        
      \BlankLine
      \BlankLine

        \private{$\cS.\text{EvaluateBackward}(s,x_s,x_{s+1},u)$}
        {
          \While{$s \neq u+1$}
          {
            \If{$s=0$}
              {
                $(x_{8n},x_{8n+1}):=P(x_0,x_1)$\;
                $s:=8n$
              }
            \Else
              {
                $\cS.\text{CompletedChains}:=\cS.\text{CompletedChains} \cup \{(s-1,x_{s-1},x_{s})\}$\;
                $x_{s-1}:=x_{s+1} \oplus \cS.\CFt^{\text{Inner}}(s,x_{s})$\;
                $Q_{s}:=Q_{s}(x_{s})$\;
                $s:=s-1$
              }
          }
          Return $(x_{s+1},x_{s+2})$
      }  

      \BlankLine
      \BlankLine

\end{algorithm}

\begin{algorithm}
\DontPrintSemicolon

  \SetKwProg{public}{public procedure}{:}{}
  \SetKwProg{private}{private procedure}{:}{}
  
  \setcounter{AlgoLine}{81}

  \private{$\cS.$Adapt$(x_{u-2},x_{u-1},x_{u+2},x_{u+3},u)$}
    {
      $x_u:=x_{u-2} \oplus \cS.\CF^{\text{Inner}}(u-1,x_{u-1})$\;
      $x_{u+1}:=x_{u+3} \oplus \cS.\CF^{\text{Inner}}(u+2,x_{u+2})$\;
      \If{$x_u \notin \cS.\CF_u$ and $x_{u+1} \notin \cS.\CF_{u+1}$}
        {
          $\cS.\CF_u(x_u) \leftarrow x_{u-1} \oplus x_{u+1}$\;
          $\cS.\CF_{u+1}(x_{u+1}) \leftarrow x_{u} \oplus x_{u+2}$
        }
      \Else{The game aborts.}

      \If{$\cS.\CFt^{\text{Inner}}(u,x_{u})=\cS.\CF_u(x_u)$ and $\cS.\CFt^{\text{Inner}}(u+1,x_{u+1})=\cS.\CF_{u+1}(x_{u+1})$}
         {
           $\cS.\text{CompletedChains}:=\cS.\text{CompletedChains} \cup \{(u,x_u,x_{u+1})\}$
         }
      \Else{The game aborts.}
      $Q_u:=Q_u(x_u)$ and $Q_{u+1}:=Q_{u+1}(x_{u+1})$\;
      \If{$(u,x_u) \in \cup_{j=1}^{8n}Q_j/Q_u$ or $(u+1,x_{u+1}) \in \cup_{j=1}^{8n}Q_j/Q_{u+1}$}{The game aborts.}
    }

\end{algorithm}

\subsection{Game 2}\label{Game 2}

The Game 2 is same as Game 1 except that the random permutation is replaced by the following two-sided random function $RF$.

\begin{algorithm}[H]
\DontPrintSemicolon

  \SetKwInput{KwSystemRF}{System $RF$}
  \SetKwInput{KwVariable}{Variable}{}

  \KwSystemRF \;
  \KwVariable
  {\;
  Hashtable $RF \subset \{\uparrow,\downarrow\} \times \{0,1\}^{2n} \times \{0,1\}^{2n}$ \;}
  
  \BlankLine
  \BlankLine
  
  \SetKwProg{public}{public procedure}{:}{}
  \SetKwProg{private}{private procedure}{:}{}
  
  \public{$RF(x_0,x_1)$}
    {
      \If{$(\downarrow,x_0,x_1) \notin RF$}
        {
          $x_{8n} \leftarrow_R \{0,1\}^n$\;
          $x_{8n+1} \leftarrow_R \{0,1\}^n$\;
          $RF(\downarrow,x_0,x_1):=(x_{8n},x_{8n+1})$\;
          $RF(\uparrow,x_{8n},x_{8n+1}):=(x_0,x_1)$ \tcc{May over write an entry}
        }
      Return $RF(\downarrow,x_0,x_1)$\;
    }
    
  \BlankLine
  \BlankLine
  
  \public{$RF^{-1}(x_{8n},x_{8n+1})$}
    {
      \If{$(\uparrow,x_{8n},x_{8n+1}) \notin RF$}
        {
          $x_0 \leftarrow_R \{0,1\}^n$\;
          $x_1 \leftarrow_R \{0,1\}^n$\;
          $RF(\downarrow,x_0,x_1):=(x_{8n},x_{8n+1})$ \tcc{May over write an entry}
          $RF(\uparrow,x_{8n},x_{8n+1}):=(x_0,x_1)$ 
        }
      Return $RF(\uparrow,x_{8n},x_{8n+1})$\;
    }  
   
\end{algorithm}

\subsection{Game 3}\label{Game 3}

The Game 3 has two systems, $\cS^1$ and $\cM^1$. We will describe all the public procedures, all the private procedures of $\cM^1$, and the private procedures of $\cS^1$ that are different from their counterparts of $\cS$.

\begin{algorithm}
\DontPrintSemicolon

  \SetKwInput{KwSystemMoneS}{System $\cM^1, \cS^1, RF$}
  \SetKwInput{KwVariable}{Variable}{}

  \KwSystemMoneS \;
  \KwVariable
  {\;
  Queue: $Q_{\cS^1}$ \;
  Tables of $\cM^1$: $\cM^1.\CF_1,...,\cM^1.\CF_{8n}$\;  
  Tables of $\cS^1$: $\cS^1.\CF_1,...,\cS^1.\CF_{8n}$\;
  Order function: $O_{\cM^1}, O_{\cS^1}$\;
  Set $\cM^1.\text{CompletedChains}, \cS^1.\text{CompletedChains}:=\emptyset$\;
  Set $\cS^1.\text{HonestyCheckedChains}:=\emptyset$\;
  Set $\cM^1.\text{MiddlePoints}:=\emptyset$\;
  Set $\cM^1.\text{AdaptedPoints}:=\emptyset$\;
  Sets $Q_i(z):=\emptyset$ for all $i \in \{1,\ldots,8n\}$ and $x \in \{0,1\}^n$\;
  Sets $Q_i:=\emptyset$ for all $i \in \{1,\ldots,8n\}$\;
  Strings $y_i:=\{0,1\}^n$ for all $i \in \{1,\ldots,8n\}$\;
  Hashtable $RF \subset \{\uparrow,\downarrow\} \times \{0,1\}^{2n} \times \{0,1\}^{2n}$\;}
  \BlankLine
  \BlankLine
  
  \SetKwProg{public}{public procedure}{:}{}
  \SetKwProg{private}{private procedure}{:}{}

  \public{$\CF(i,x)$}
    { 
      $\cS^1.\CF^{\text{Inner}}(i,x)$\;
      \While{$\neg Q_{\cS^1}.\text{Empty}()$}
         {
            $(s,x_s,...,x_{s+n/10-1}):=Q_{\cS^1}.\text{Dequeue}()$\;
            \If{$\cS^1.\text{Check}(s,x_s,...,x_{s+n/10-1})=(s,x_s,x_{s+1},u)$}
               {
                 $\cS^1.$Complete $(s,x_s,x_{s+1},u)$
               }
         }
        Return $\cS^1.\CF_i(x)$
    }

\BlankLine
\BlankLine

  \public{$RF(x_0,x_1)$}
    {
      \If{$(\downarrow,x_0,x_1) \notin RF$}
        {
          $x_{8n} \leftarrow_R \{0,1\}^n$\;
          $x_{8n+1} \leftarrow_R \{0,1\}^n$\;
          $RF(\downarrow,x_0,x_1):=(x_{8n},x_{8n+1})$\;
          $RF(\uparrow,x_{8n},x_{8n+1}):=(x_0,x_1)$\tcc{May over write an entry}
          $\cM^1.$Complete $(0,x_0,x_1,4n)$
        }
      Return $RF(\downarrow,x_0,x_1)$\;
    }
    
  \BlankLine
  \BlankLine
  
  \public{$RF^{-1}(x_{8n},x_{8n+1})$}
    {
      \If{$(\uparrow,x_{8n},x_{8n+1}) \notin RF$}
        {
          $x_0 \leftarrow_R \{0,1\}^n$\;
          $x_1 \leftarrow_R \{0,1\}^n$\;
          $RF(\downarrow,x_0,x_1):=(x_{8n},x_{8n+1})$ \tcc{May over write an entry}
          $RF(\uparrow,x_{8n},x_{8n+1}):=(x_0,x_1)$ \;
          $\cM^1.$Complete $(8n,x_{8n},x_{8n+1},4n)$
        }
      Return $RF(\uparrow,x_{8n},x_{8n+1})$\;
    }  
    
  \BlankLine
  \BlankLine   
  
  \private{$\cM^1.$Complete $(s,x_s,x_{s+1},u)$}
    {
      $(x_{u-2},x_{u-1}):=\cM^1.\text{EvaluateForward}(s,x_s,x_{s+1},u)$\;
      $(x_{u+2},x_{u+3}):=\cM^1.\text{EvaluateBackward}(s,x_s,x_{s+1},u)$\;
      $\cM^1.\text{Adapt}(x_{u-2},x_{u-1},x_{u+2},x_{u+3},u)$
    }

\end{algorithm}  

\begin{algorithm}
\DontPrintSemicolon

  \SetKwInput{KwSystemMoneS}{System $\cM^1, \cS^1, RF$}
  \SetKwInput{KwVariable}{Variable}{}

  \SetKwProg{public}{public procedure}{:}{}
  \SetKwProg{private}{private procedure}{:}{}
  
  \setcounter{AlgoLine}{27}

    \private{$\cM^1.\text{EvaluateForward}(s,x_s,x_{s+1},u)$}
      {
        $\cM^1.\text{CompletedChains}:=\cM^1.\text{CompletedChains} \cup \{(s,x_{s},x_{s+1})\}$\;
        \While{$s \neq u-1$}
          {
            \If{$s=8n$}
              {
                $(x_0,x_1):=RF(\uparrow,x_{8n},x_{8n+1})$\;
                $s:=0$
              }
            \Else
              {
                $\cM^1.\text{CompletedChains}:=\cM^1.\text{CompletedChains} \cup \{(s+1,x_{s+1},x_{s+2})\}$\;
                $x_{s+2}:=x_s \oplus \cM^1.\CFt^{\text{Inner}}(s+1,x_{s+1})$\;
                \If{$3n \leq s+1 \leq 5n$}{$M^1.\text{MiddlePoints}=M^1.\text{MiddlePoints} \cup ({s+1},x_{s+1})$}
                $s:=s+1$
              }
          }
        Return $(x_{s-1},x_s)$  
      }

     \BlankLine
     \BlankLine
     
      \private{$\cM^1.\text{EvaluateBackward}(s,x_s,x_{s+1},u)$}
      {
        \While{$s \neq u+1$}
        {
          \If{$s=0$}
          {
            $(x_{8n},x_{8n+1}):=RF(\downarrow,x_0,x_1)$\;
            $s:=8n$
          }
          \Else
          {
            $\cM^1.\text{CompletedChains}:=\cM^1.\text{CompletedChains} \cup \{(s-1,x_{s-1},x_{s})\}$\;
            $x_{s-1}:=x_{s+1} \oplus \cM^1.\CFt^{\text{Inner}}(s,x_{s})$\;
            \If{$3n \leq s \leq 5n$}{$\cM^1.\text{MiddlePoints}=\cM^1.\text{MiddlePoints} \cup ({s},x_{s})$}
            $s:=s-1$
          }
        } 
        Return $(x_{s+1},x_{s+2})$ 
      }    

  \BlankLine
  \BlankLine
  
  \private{$\cM^1.$Adapt$(x_{u-2},x_{u-1},x_{u+2},x_{u+3},u)$}
    {
      $x_u:=x_{u-2} \oplus \cM^1.\CF^{\text{Inner}}(u-1,x_{u-1})$\;
      $x_{u+1}:=x_{u+3} \oplus \cM^1.\CF^{\text{Inner}}(u+2,x_{u+2})$\;
      \If{$x_u \notin \cM^1.\CF_u$}
        {
          $\cM^1.\CF_u(x_u) \leftarrow x_{u-1} \oplus x_{u+1}$\;
          $\cM^1.\text{MiddlePoints}=\cM^1.\text{MiddlePoints} \cup (u,x_u)$\;
          $\cM^1.\text{AdaptedPoints}=\cM^1.\text{AdaptedPoints} \cup (u,x_u)$\;   
        }
      \If{$x_{u+1} \notin \cM^1.\CF_{u+1}$}
        {
          $\cM^1.\CF_{u+1}(x_{u+1}) \leftarrow x_{u} \oplus x_{u+2}$\;
          $\cM^1.\text{MiddlePoints}=\cM^1.\text{MiddlePoints} \cup ({u+1},x_{u+1})$\;
          $\cM^1.\text{AdaptedPoints}=\cM^1.\text{AdaptedPoints} \cup ({u+1},x_{u+1})$\;
        }
      \If{$\cM^1.\CFt^{\text{Inner}}(u,x_{u})=\cM^1.\CF_u(x_u)$ and $\cM^1.\CFt^{\text{Inner}}(u+1,x_{u+1})=\cM^1.\CF_{u+1}(x_{u+1})$}
         {
           $\cM^1.\text{CompletedChains}:=\cM^1.\text{CompletedChains} \cup \{(u,x_u,x_{u+1})\}$
         }
    }
    
  \BlankLine
  \BlankLine
    
  \private{$\cM^1.\CF^{\text{Inner}}(i,x)$}
    {
      \If{$x \notin \cM^1.\CF_i$}
        {
          $\cM^1.\CF_i(x) \leftarrow_R\{0,1\}^n$
        }
      Return $\cM^1.\CF_i(x)$    
    }    

  \BlankLine
  \BlankLine

\end{algorithm}

\begin{algorithm}
\DontPrintSemicolon
  
  \SetKwProg{public}{public procedure}{:}{}
  \SetKwProg{private}{private procedure}{:}{}
  
  \setcounter{AlgoLine}{71}

  \private{$\cM^1.\CFt^{\text{Inner}}(i,x)$}
    {
      \While{$\cA(i,x)$ queries $\CF_j(y)$} 
        {
          $Q_i(x):=Q_i(x) \cup \{(j,y)\}$\;
          $\cM^1.\CF^{\text{Inner}}(j,y)$      \tcc{Simulates the subversion algorithm $\cA$ on input $(i,x)$}
        }
      Return $\cA(i,x)$
    }

  \BlankLine
  \BlankLine
    
  \private{$\cS^1.\CF^{\text{Inner}}(i,x)$}
    {
      \If{$x \notin \cS^1.\CF_i$}
        {\If{$x \notin \cM^1.\CF_i$}
          {
            $\cS^1.\CF_i(x) \leftarrow_R\{0,1\}^n$\;
            $\cM^1.\CF_i(x) \leftarrow \cS^1.\CF_i(x)$
          }
        \ElseIf{$x \in \cM^1.\CF_i$ and $(i,x) \notin M^1.\text{AdaptedPoints}$}
          {
            $\cS^1.\CF_i(x) \leftarrow \cM^1.\CF_i(x)$
          }
        \Else
          {
            $\cS^1.\CF_i(x) \leftarrow_R\{0,1\}^n$
          }  
        $\cS^1.\text{EnqueueNewChain}(i,x)$
        }
      Return $\cS^1.\CF_i(x)$    
    }

  \BlankLine
  \BlankLine

    \private{$\cS^1.\text{EvaluateForward}(s,x_s,x_{s+1},u)$}
      {
        $\cS^1.\text{CompletedChains}:=\cS^1.\text{CompletedChains} \cup \{(s,x_{s},x_{s+1})\}$\;
        \While{$s \neq u-1$}
          {
            \If{$s=8n$}
              {
                 \If{$(\uparrow,x_{8n},x_{8n+1}) \notin RF$}
                   {
                     $x_0 \leftarrow_R \{0,1\}^n$\;
                     $x_1 \leftarrow_R \{0,1\}^n$\;
                     $RF(\downarrow,x_0,x_1):=(x_{8n},x_{8n+1})$ \tcc{May over write an entry}
                     $RF(\uparrow,x_{8n},x_{8n+1}):=(x_0,x_1)$
                   }
                 \Else{$(x_0,x_1):=RF(\uparrow,x_{8n},x_{8n+1})$}  
                $s:=0$
              }
            \Else
              {
                $\cS^1.\text{CompletedChains}:=\cS^1.\text{CompletedChains} \cup \{(s+1,x_{s+1},x_{s+2})\}$\;
                $x_{s+2}:=x_s \oplus \cS^1.\CFt^{\text{Inner}}(s+1,x_{s+1})$\;
                $Q_{s+1}:=Q_{s+1}(x_{s+1})$\;
                $s:=s+1$
              }
          }
        Return $(x_{s-1},x_s)$  
      } 
 \BlankLine   
 \BlankLine

\end{algorithm}  

\begin{algorithm}[H]
\DontPrintSemicolon
  
  \SetKwProg{public}{public procedure}{:}{}
  \SetKwProg{private}{private procedure}{:}{}
  
  \setcounter{AlgoLine}{115}

    \private{$\cS^1.\text{EvaluateBackward}(s,x_s,x_{s+1},u)$}
      {
        \While{$s \neq u+1$}
          {
            \If{$s=0$}
              {  
                \If{$(\downarrow,x_0,x_1) \notin RF$}
                  {
                    $x_{8n} \leftarrow_R \{0,1\}^n$\;
                    $x_{8n+1} \leftarrow_R \{0,1\}^n$\;
                    $RF(\downarrow,x_0,x_1):=(x_{8n},x_{8n+1})$\;
                    $RF(\uparrow,x_{8n},x_{8n+1}):=(x_0,x_1)$\tcc{May over write an entry}
                  }
                \Else{$(x_{8n},x_{8n+1}):=RF(\downarrow,x_0,x_1)$}
                $s:=8n$
              }
            \Else
              {
                $\cS^1.\text{CompletedChains}:=\cS^1.\text{CompletedChains} \cup \{(s-1,x_{s-1},x_{s})\}$\;
                $x_{s-1}:=x_{s+1} \oplus \cS^1.\CFt^{\text{Inner}}(s,x_{s})$\;
                $Q_{s}:=Q_{s}(x_{s})$\;
                $s:=s-1$
              }
          }
        Return $(x_{s+1},x_{s+2})$  
      }      

 \BlankLine 
  \BlankLine   
        
  \private{$\cS^1.Adapt(x_{u-2},x_{u-1},x_{u+2},x_{u+3},u)$}
    {
      $x_u:=x_{u-2} \oplus \cS^1.\CF^{\text{Inner}}(u-1,x_{u-1})$\;
      $x_{u+1}:=x_{u+3} \oplus \cS^1.\CF^{\text{Inner}}(u+2,x_{u+2})$\;
      \If{$x_u \notin \cS^1.\CF_u$ and $x_{u+1} \notin \cS^1.\CF_{u+1}$}
        {
          $\cS^1.\CF_u(x_u) \leftarrow x_{u-1} \oplus x_{u+1}$\;
          $\cS^1.\CF_{u+1}(x_{u+1}) \leftarrow x_{u} \oplus x_{u+2}$\;
          $\cM^1.\CF_u(x_u) \leftarrow \cS^1.\CF_u(x_u)$\tcc{May over write an entry} 
          $\cM^1.\CF_{u+1}(x_{u+1}) \leftarrow \cS^1.\CF_{u+1}(x_{u+1})$\tcc{May over write an entry}  
        }
      \Else{The game aborts.}
      \If{$\cS^1.\CFt^{\text{Inner}}(u,x_{u})=\cS^1.\CF_u(x_u)$ and $\cS^1.\CFt^{\text{Inner}}(u+1,x_{u+1})=\cS^1.\CF_{u+1}(x_{u+1})$}
         {
           $\cS^1.\text{CompletedChains}:=\cS^1.\text{CompletedChains} \cup \{(u,x_u,x_{u+1})\}$
         }
      \Else{The game aborts.}
      $Q_u:=Q_u(x_u)$ and $Q_{u+1}:=Q_{u+1}(x_{u+1})$\;
      \If{$(u,x_u) \in \cup_{j=1}^{8n}Q_j/Q_u$ or $(u+1,x_{u+1}) \in \cup_{j=1}^{8n}Q_j/Q_{u+1}$}{The game aborts.}
    }

\end{algorithm}

\subsection{Game 4}\label{Game 4}
To obtain $G_4$, we just need to add some abortion conditions to several procedures of $G_3$. Below we only show the procedures of $\cS^2$ and $\cM^2$ that are different from their counterparts of $\cS^1$ and $\cM^1$. We use red color to stress the extra abortion conditions.

\begin{algorithm}
\DontPrintSemicolon
 
     \BlankLine
     \BlankLine   
  
  \SetKwProg{public}{public procedure}{:}{}
  \SetKwProg{private}{private procedure}{:}{}

  \private{$\cM^2.\CF^{\text{Inner}}(i,x)$}
    {
      \If{$x \notin \cM^2.\CF_i$}
        {
          $\cM^2.\CF_i(x) \leftarrow_R\{0,1\}^n$
        }
      \ElseIf{$(i,x) \in \cM^2.\text{MiddlePoints}$ and $(i,x) \notin \cS^2.\CF_i$}{\textcolor{red}{The game aborts.}}  
      Return $\cM^1.\CF_i(x)$    
    } 
    
     \BlankLine
     \BlankLine    
      
    \private{$\cM^2.\text{EvaluateForward}(s,x_s,x_{s+1},u)$}
      {
        $\cM^2.\text{CompletedChains}:=\cM^2.\text{CompletedChains} \cup \{(s,x_{s},x_{s+1})\}$\;
        \While{$s \neq u-1$}
          {
            \If{$s=8n$}
              {
                $(x_0,x_1):=RF(\uparrow,x_{8n},x_{8n+1})$\;
                $s:=0$
              }
            \Else
              {
                $\cM^2.\text{CompletedChains}:=\cM^2.\text{CompletedChains} \cup \{(s+1,x_{s+1},x_{s+2})\}$\;                
                \If{$3n \leq s+1 \leq 5n$}{$\cM^2.\text{MiddlePoints}=\cM^2.\text{MiddlePoints} \cup ({s+1},x_{s+1})$}
                \If{$3n \leq s+1 \leq 5n$ and $x_{s+1} \in \cM^2.\CF$}{\textcolor{red}{The game aborts.}}
                $x_{s+2}:=x_s \oplus \cM^2.\CFt^{\text{Inner}}(s+1,x_{s+1})$\;
                $Q_{s+1}:=Q_{s+1}(x_{s+1})$\;
                $y_{s+1}:=x_{s+1}$\;
                \If{$3n \leq s+1 \leq 5n$ and $\cM^2.\CF_{s+1}(x_{s+1}) \neq x_{s+2}\oplus x_s$}{\textcolor{red}{The game aborts.}}
                $s:=s+1$
              }
          }
        Return $(x_{s-1},x_s)$  
      }

     \BlankLine
     \BlankLine

\end{algorithm}

\begin{algorithm}
\DontPrintSemicolon
 
     \BlankLine
     \BlankLine   
  
  \SetKwProg{public}{public procedure}{:}{}
  \SetKwProg{private}{private procedure}{:}{}

  \setcounter{AlgoLine}{23}
       
      \private{$\cM^2.\text{EvaluateBackward}(s,x_s,x_{s+1},u)$}
      {
        \While{$s \neq u+1$}
        {
          \If{$s=0$}
          {
            $(x_{8n},x_{8n+1}):=RF(\downarrow,x_0,x_1)$\;
            $s:=8n$
          }
          \Else
          {
            $\cM^2.\text{CompletedChains}:=\cM^2.\text{CompletedChains} \cup \{(s-1,x_{s-1},x_{s})\}$\;
            \If{$3n \leq s \leq 5n$}{$\cM^2.\text{MiddlePoints}=\cM^2.\text{MiddlePoints} \cup ({s},x_{s})$}
            \If{$3n \leq s \leq 5n$ and $x_{s} \in \cM^2.\CF$}{\textcolor{red}{The game aborts.}}
            $x_{s-1}:=x_{s+1} \oplus \cM^2.\CFt^{\text{Inner}}(s,x_{s})$\;
            $Q_{s}:=Q_{s}(x_{s})$\;
            $y_{s}:=x_{s}$\;
            \If{$3n \leq s \leq 5n$ and $\cM^2.\CF_{s}(x_{s}) \neq x_{s+1}\oplus x_{s-1}$}{\textcolor{red}{The game aborts.}}
            $s:=s-1$
          }
        } 
        Return $(x_{s+1},x_{s+2})$ 
      } 
         
  \BlankLine
  \BlankLine
      
  \private{$\cM^2.$Adapt$(x_{u-2},x_{u-1},x_{u+2},x_{u+3},u)$}
    {
      $x_u:=x_{u-2} \oplus \cM^2.\CF^{\text{Inner}}(u-1,x_{u-1})$\;
      $x_{u+1}:=x_{u+3} \oplus \cM^2.\CF^{\text{Inner}}(u+2,x_{u+2})$\;
      \If{$x_u \notin \cM^2.\CF_u$ and $x_{u+1} \notin \cM^2.\CF_{u+1}$}
        {
          $\cM^2.\CF_u(x_u) \leftarrow x_{u-1} \oplus x_{u+1}$\;
          $\cM^2.\text{MiddlePoints}=\cM^2.\text{MiddlePoints} \cup (u,x_u)$\;    
          $\cM^2.\text{AdaptedPoints}=\cM^2.\text{AdaptedPoints} \cup ({u},x_{u})$\;
          $\cM^2.\CF_{u+1}(x_{u+1}) \leftarrow x_{u} \oplus x_{u+2}$\;
          $\cM^2.\text{MiddlePoints}=\cM^2.\text{MiddlePoints} \cup ({u+1},x_{u+1})$\;
          $\cM^2.\text{AdaptedPoints}=\cM^2.\text{AdaptedPoints} \cup ({u+1},x_{u+1})$\;
        }
      \Else{\textcolor{red}{The game aborts.}}  
      \If{$\cM^2.\CFt^{\text{Inner}}(u,x_{u})=\cM^2.\CF_u(x_u)$ and $\cM^2.\CFt^{\text{Inner}}(u+1,x_{u+1})=\cM^2.\CF_{u+1}(x_{u+1})$}
         {
           $\cM^2.\text{CompletedChains}:=\cM^2.\text{CompletedChains} \cup \{(u,x_u,x_{u+1})\}$
         }
      \Else{\textcolor{red}{The game aborts.}}
      $Q_u:=Q_u(x_u)$ and $Q_{u+1}:=Q_{u+1}(x_{u+1})$\;
      $y_{u}:=x_{u}$ and $y_{u+1}:=x_{u+1}$\;
      \While{$3n \leq i \leq 5n$}
      {
      \If{$(i,y_i) \in \cup_{j=1}^{8n}Q_j/Q_i$}{\textcolor{red}{The game aborts.}} 
      }    
    }

    \BlankLine
    \BlankLine

\end{algorithm}

\begin{algorithm}
\DontPrintSemicolon
  
  \SetKwProg{public}{public procedure}{:}{}
  \SetKwProg{private}{private procedure}{:}{}
  
  \setcounter{AlgoLine}{64}

  \private{$\cS^2.\CF^{\text{Inner}}(i,x)$}
    {
      \If{$x \notin \cS^2.\CF_i$}
        {\If{$x \notin \cM^2.\CF_i$}
          {
            $\cS^2.\CF_i(x) \leftarrow_R\{0,1\}^n$\;
            $\cM^2.\CF_i(x) \leftarrow \cS^2.\CF_i(x)$
          }
        \ElseIf{$x \in \cM^2.\CF_i$ and $(i,x) \notin \cM^2.\text{MiddlePoints}$}
          {
            $\cS^2.\CF_i(x) \leftarrow \cM^2.\CF_i(x)$
          }
        \Else
          {
            \textcolor{red}{The game aborts.}
          }  
        $\cS^2.\text{EnqueueNewChain}(i,x)$
        }
      Return $\cS^2.\CF_i(x)$    
    }
    
    \BlankLine
    \BlankLine
  
  \private{$\cS^2.$Complete $(s,x_s,x_{s+1},u)$}
    {
      \If{$(s,x_s,x_{s+1}) \in \cM^2.\text{CompletedChains}$ and $u = 7n$}{\textcolor{red}{The game aborts.}}
      \ElseIf{$(s,x_s,x_{s+1}) \in \cM^2.\text{CompletedChains}$ and $u = 4n$}
      {Copy the full subverted chain containing $(s,x_s,x_{s+1})$ in $\cM^2.\CF$ to $\cS^2.\CF$}
      \Else{$(x_{u-2},x_{u-1}):=\cS^2.\text{EvaluateForward}(s,x_s,x_{s+1},u)$\;
      $(x_{u+2},x_{u+3}):=\cS^2.\text{EvaluateBackward}(s,x_s,x_{s+1},u)$\;
      $\cS^2.\text{Adapt}(x_{u-2},x_{u-1},x_{u+2},x_{u+3},u)$}
    }

    \BlankLine
    \BlankLine    
        
  \private{$\cS^2.$Adapt$(x_{u-2},x_{u-1},x_{u+2},x_{u+3},u)$}
    {
      $x_u:=x_{u-2} \oplus \cS^2.\CF^{\text{Inner}}(u-1,x_{u-1})$\;
      $x_{u+1}:=x_{u+3} \oplus \cS^2.\CF^{\text{Inner}}(u+2,x_{u+2})$\;
      \If{$x_u \notin \cS^2.\CF_u$, $x_{u+1} \notin \cS^2.\CF_{u+1}$} 
        {
          $\cS^2.\CF_u(x_u) \leftarrow x_{u-1} \oplus x_{u+1}$\;
          $\cS^2.\CF_{u+1}(x_{u+1}) \leftarrow x_{u} \oplus x_{u+2}$\;
          \If{$x_u \notin \cM^2.\CF_u$ and $x_{u+1} \notin \cM^2.\CF_{u+1}$}
            {
              $\cM^2.\CF_u(x_u) \leftarrow \cS^2.\CF_u(x_u)$\;
              $\cM^2.\CF_u(x_{u+1}) \leftarrow \cS^2.\CF_{u+1}(x_{u+1})$
            }
          \Else{\textcolor{red}{The game aborts.}}   
        }
      \Else{The game aborts.}
      \If{$\cS^2.\CFt^{\text{Inner}}(u,x_{u})=\cS^2.\CF_u(x_u)$ and $\cS^2.\CFt^{\text{Inner}}(u+1,x_{u+1})=\cS^2.\CF_{u+1}(x_{u+1})$}
         {
           $\cS^2.\text{CompletedChains}:=\cS^2.\text{CompletedChains} \cup \{(u,x_u,x_{u+1})\}$
         }
      \Else{The game aborts.}
      $Q_u:=Q_u(x_u)$ and $Q_{u+1}:=Q_{u+1}(x_{u+1})$\;
      \If{$(u,x_u) \in \cup_{j=1}^{8n}Q_j/Q_u$ or $(u+1,x_{u+1}) \in \cup_{j=1}^{8n}Q_j/Q_{u+1}$}{The game aborts.}
    }
    
\end{algorithm}

\newpage

\subsection{Game 5}\label{Game 5}

$G_5$ is different from $G_4$ in the following procedures:

\begin{algorithm}
\DontPrintSemicolon

     \BlankLine
     \BlankLine  
       
  \SetKwProg{public}{public procedure}{:}{}
  \SetKwProg{private}{private procedure}{:}{}

  \public{$RF(x_0,x_1)$}
    {
      \If{$(\downarrow,x_0,x_1) \notin RF$}
        {
          $\cM^3.$Complete $(0,x_0,x_1,4n)$
        }
      Return $RF(\downarrow,x_0,x_1)$\;
    }
    
  \BlankLine
  \BlankLine
  
  \public{$RF^{-1}(x_{8n},x_{8n+1})$}
    {
      \If{$(\uparrow,x_{8n},x_{8n+1}) \notin RF$}
        {
          $\cM^3.$Complete $(8n,x_{8n},x_{8n+1},4n)$
        }
      Return $RF(\uparrow,x_{8n},x_{8n+1})$\;
    }  
    
  \BlankLine
  \BlankLine   
  
  \private{$\cM^3.$Complete $(s,x_s,x_{s+1},u)$}
    {
    \If{$s=0$}
      {\While{$s \neq 8n$}
          {
            $x_{s+2}:=x_s \oplus \cM^3.\CFt^{\text{Inner}}(s+1,x_{s+1})$\;
            \If{$3n \leq s+1 \leq 5n$ and $\cM^3.\CF_{s+1}(x_{s+1}) \neq x_{s}\oplus x_{s+2}$}
              {The game aborts.}
            \If{$3n \leq s+1 \leq 5n$}
              {$\cM^3.\text{MiddlePoints}=\cM^3.\text{MiddlePoints} \cup ({s+1},x_{s+1})$\;}              
            \If{$s+1 =u,u+1$}
              {$\cM^3.\text{AdaptedPoints}=\cM^3.\text{AdaptedPoints} \cup ({s+1},x_{s+1})$\;}                            
            $\cM^3.\text{CompletedChains}:=\cM^3.\text{CompletedChains} \cup \{(s+1,x_{s+1},x_{s+2})\}$\;
            $s:=s+1$
          }
          }
     \Else     
        {\While{$s \neq 0$}
          {
            $x_{s-1}:=x_{s+1} \oplus \cM^3.\CFt^{\text{Inner}}(s,x_{s})$\;
            \If{$3n \leq s \leq 5n$ and $\cM^3.\CF_{s}(x_{s}) \neq x_{s+1}\oplus x_{s-1}$}
              {The game aborts.}
            \If{$3n \leq s \leq 5n$}
              {$\cM^3.\text{MiddlePoints}=\cM^3.\text{MiddlePoints} \cup ({s},x_{s})$\;}              
            \If{$s =u,u+1$}
              {$\cM^3.\text{AdaptedPoints}=\cM^3.\text{AdaptedPoints} \cup ({s},x_{s})$\;}                                          
            $\cM^3.\text{CompletedChains}:=\cM^3.\text{CompletedChains} \cup \{(s-1,x_{s-1},x_{s})\}$\;
            $s:=s-1$
          } 
        }
      \While{$3n \leq i \leq 5n$}
      {
      \If{$(i,x_i) \in \cup_{j=1}^{8n}Q_j/Q_i$}{The game aborts.}
      }        
        $RF(\downarrow,x_0,x_1):=(x_{8n},x_{8n+1})$\;  
        $RF(\uparrow,x_{8n},x_{8n+1}):=(x_0,x_1)$      
    }
  
\end{algorithm}

\begin{algorithm}
\DontPrintSemicolon
  
  \SetKwProg{public}{public procedure}{:}{}
  \SetKwProg{private}{private procedure}{:}{}
  
  \setcounter{AlgoLine}{31}
  
  \private{$\cS^3.$Complete $(s,x_s,x_{s+1},u)$}
    {
    \If{$(s,x_s,x_{s+1}) \in \cM^3.\text{CompletedChains}$ and $u = 7n$}{The game aborts.}
    \ElseIf{$(s,x_s,x_{s+1}) \in \cM^3.\text{CompletedChains}$ and $u = 4n$}
      {Copy the full subverted chain containing $(s,x_s,x_{s+1})$ in $\cM^3.\CF$ to $\cS^3.\CF$}    
    \Else{
      $i:=s$\;
      \While{$s \leq i \leq 8n$}
          {
            \If{$i+1=u$ or $i+1=u+1$}
              {
                \If{$x_{i+1}$ is in $\cS^3.\CF_{i+1}$ or $\cM^3.\CF_{i+1}$}{The game aborts.}
              }
            $x_{i+2}:=x_i \oplus \cS^3.\CFt^{\text{Inner}}(i+1,x_{i+1})$\;
            \If{$i+1=u$ or $i+1=u+1$}
              {
                \If{$\cS^3.\CF^{\text{Inner}}(i+1,x_{i+1}) \neq \cS^3.\CFt^{\text{Inner}}(i+1,x_{i+1})$}{The game aborts.}
              }
            $\cS^3.\text{CompletedChains}:=\cS^3.\text{CompletedChains} \cup \{(i+1,x_{i+1},x_{i+2})\}$\;
            $i:=i+1$
          }
       $j:=s$\;      
        \While{$0 \leq j \leq s$}
          {
            \If{$j=u$ or $j=u+1$}
              {
                \If{$x_j$ is in $\cS^3.\CF_j$ or $\cM^3.\CF_j$}{The game abort.}
              }
            $x_{j-1}:=x_{j+1} \oplus \cS^3.\CFt^{\text{Inner}}(j,x_{j})$\;
            \If{$j=u$ or $j=u+1$}
              {
                \If{$\cS^3.\CF^{\text{Inner}}(j,x_{j}) \neq \cS^3.\CFt^{\text{Inner}}(j,x_{j})$}{The game aborts.}
              }
            $\cS^3.\text{CompletedChains}:=\cS^3.\text{CompletedChains} \cup \{(j-1,x_{j-1},x_{j})\}$\;
            $j:=j-1$ 
        }
      \If{$(u,x_u) \in \cup_{j=1}^{8n}Q_j/Q_u$ or $(u+1,x_{u+1}) \in \cup_{j=1}^{8n}Q_j/Q_{u+1}$}{The game aborts.}
     
     }       
        $RF(\downarrow,x_0,x_1):=(x_{8n},x_{8n+1})$\;  
        $RF(\uparrow,x_{8n},x_{8n+1}):=(x_0,x_1)$      
    }
  
\end{algorithm}


\end{document}